\documentclass[11pt]{article}
\usepackage[margin=1in]{geometry}
\usepackage{amsmath,amssymb,amsthm,bm,mathtools}
\usepackage{booktabs}
\usepackage{graphicx}
\usepackage{enumitem}
\usepackage{natbib}
\usepackage{url}
\usepackage{etoolbox}
\usepackage[colorlinks=true,citecolor=blue,linkcolor=blue,urlcolor=blue]{hyperref}
\allowdisplaybreaks
\numberwithin{equation}{section}

\newtheorem{theorem}{Theorem}
\newtheorem{proposition}{Proposition}
\newtheorem{corollary}{Corollary}
\theoremstyle{definition}
\newtheorem{assumption}{Assumption}
\newtheorem{remark}{Remark}
\newtheorem{definition}{Definition}
\theoremstyle{plain}
\newtheorem{lemma}{Lemma}

\makeatletter
\patchcmd{\l@section}{1.5em}{2.7em}{}{}
\patchcmd{\l@section}{1.0em \@plus\p@}{0.5em \@plus\p@}{}{}
\renewcommand*\l@subsection{\@dottedtocline{2}{2.7em}{3.4em}}
\patchcmd{\l@part}{2.25em}{1.2em}{}{}
\makeatother
\newcommand{\supppart}[1]{%
  \clearpage\phantomsection\addcontentsline{toc}{part}{#1}%
  \begin{center}\Large\bfseries #1\end{center}\medskip}

\newcommand{\R}{\mathbb{R}}
\newcommand{\E}{\mathbb{E}}
\newcommand{\Pp}{\mathbb{P}}
\newcommand{\ind}{\mathbf{1}}
\newcommand{\calT}{\mathcal{T}}
\newcommand{\bmu}{\bm{\mu}}
\newcommand{\btheta}{\bm{\theta}}

\newcommand{\bone}{\bm{1}}
\newcommand{\rights}{\rightsquigarrow}
\DeclareMathOperator{\Cov}{Cov}
\DeclareMathOperator{\Var}{Var}

\begin{document}

\def\spacingset#1{\renewcommand{\baselinestretch}%
{#1}\small\normalsize} \spacingset{1}
% Floats keep single spacing. The text uses the template spacing below.
\AtBeginEnvironment{table}{\spacingset{1}}
\AtBeginEnvironment{figure}{\spacingset{1}}

%%%%%%%%%%%%%%%%%%%%%%%%%%%%%%%%%%%%%%%%%%%%%%%%%%%%%%%%%%%%%%%%%%%%%%%%%%%%%%

\title{\bfseries Strong Simultaneous Inference for Restricted Mean Survival
  Time Processes in Multi-Arm Trials under Covariate-Adaptive Randomization}
\author{Guannan Zhai\qquad Feifang Hu\\[0.5em]
  {\normalsize Department of Statistics, The George Washington University}}
\date{}
\maketitle
\addtocontents{toc}{\protect\setcounter{tocdepth}{-1}}

\begin{abstract}
Confirmatory trials increasingly test several treatments against a common
control, and many balance prognostic factors through covariate-adaptive
randomization (CAR). For time-to-event endpoints, they usually rely
on the stratified log-rank test and the hazard ratio, but the former may
lose power and the latter its clear interpretation when hazards are not
proportional. The restricted mean survival time (RMST) provides an
interpretable alternative, yet inference for multiple RMST comparisons
under CAR remains largely unexplored. We first derive the joint asymptotic
distribution of the standard Kaplan--Meier RMST estimators in multi-arm
trials. Under Pocock--Simon minimization, their limiting covariance
contains a design term that arises only when the randomization factors
interact in their effect on survival. A stratum-standardized
estimator removes this term and remains valid when censoring depends on
the strata. We then develop test statistics that account for the shared
control and the randomization design. Calibrated by a Gaussian multiplier,
the resulting max-$T$ procedures control the family-wise error rate and
give simultaneous confidence bands over restriction times. In simulations,
the method controls the Type I error rate under various randomization
procedures and outperforms the stratified log-rank test when the benefit
wanes. An application based on the three-arm ICON6 trial
illustrates its practicality.

\end{abstract}

\spacingset{1.5}
\section{Introduction}
\label{sec:intro}

Randomized controlled trials are the gold standard for evaluating new
treatments, and a growing share of them compare several experimental
treatments with one shared control, which saves patients, time, and cost
relative to separate two-arm trials \citep{parmar2014more}. Examples include
seamless phase II/III trials \citep{bretz2009adaptive}, dose-optimization
trials \citep{fda2024optimizing}, and platform trials
\citep{woodcock2017master}. Most trials balance prognostic factors by
covariate-adaptive randomization (CAR)
\citep{pocock1975sequential,kahan2012reporting},
especially in oncology, where survival endpoints are standard and 85\% of
multi-arm phase III trials stratify the randomization
\citep{pond2010trends}.

For time-to-event endpoints, the primary analysis is usually a log-rank
test stratified by the randomization factors, with the effect reported as a
hazard ratio from a Cox model \citep{cox1972regression}. Among the phase III
oncology trials reviewed by \citet{batson2016review}, 88\% used Cox
regression and most did not report any check of proportional hazards.
However, modern therapies increasingly violate this assumption.
Immunotherapy curves often separate only after several months, and a
time-limited benefit can fade once treatment stops
\citep{alexander2018hazards,oza2015icon7}. If trials keep relying on these
traditional methods, the log-rank test can lose much of its power, and the
hazard ratio averages a time-varying effect with weights that depend on
follow-up and censoring \citep{hernan2010hazards,uno2014moving}.

The restricted mean survival time (RMST) is a natural alternative. It is
estimated from Kaplan--Meier curves without any model for the hazard. Its
treatment contrast is an interpretable and collapsible difference in mean
event-free time, which defines a model-free population-level estimand in
the sense of ICH E9(R1)
\citep{royston2013restricted,uno2014moving,ich2019estimands}. RMST has
become a prominent option for the primary analysis when nonproportional
hazards are anticipated \citep{freidlin2019methods}, and multi-arm trials
such as STAMPEDE, ICON6, and SORCE reported RMST contrasts after observing
nonproportional hazards
\citep{clarke2019stampede,ledermann2021icon6,eisen2020sorce}. The FDA
guidance on covariate adjustment notes that the hazard ratio can be
non-collapsible and that results from nonlinear models ``can be difficult
to interpret if the model is misspecified'' \citep{fda2023covariates}.

Design-aware inference under CAR has been studied extensively for means and
regression estimators \citep{ma2015testing,bugni2018inference,bugni2019inference,ma2020statistical,ye2022inference,zhai2027valid},
building on the asymptotic theory of adaptive randomization procedures
\citep{hu2012asymptotic,hu2020theory,hu2023multi,zhang2026multitreatment,zhai2024group,alkhnefr2025efficient}.
However, for survival endpoints under CAR, the available results mostly
concern two arms
\citep{ye2020robust,wang2023model,ye2024covariate,johnson2024validity,zhao2024consistent},
and most of them focus on the log-rank test and the Cox model. RMST methods
have also been developed, but they compare only two arms, do not account for
CAR, or consider a single restriction time \citep{munko2024rmst,krajewski2026randomization}. To our knowledge, model-free, design-aware joint inference for time-to-event outcomes in multi-arm CAR trials that provides both an interpretable treatment effect and valid p-values is still missing.

In this article, we investigate inference for right-censored time-to-event outcomes under multi-arm covariate-adaptive randomization (CAR). Specifically, we study two Kaplan--Meier RMST estimators, the unstratified estimator used in current practice and the stratum-standardized estimator from the direct-adjusted survival curve \citep{chang1982corrected,cupples1995age}, which we recommend as the primary analysis. Our contributions are twofold. 
\begin{itemize}
    \item[] First, we derive the joint asymptotic distribution of the RMST processes of all arms under CAR, for arbitrary treatment effects. One of our main theoretical findings concerns Pocock--Simon minimization. Under minimization, the limiting covariance of the unstratified estimator
      contains a design term, which is absent under procedures with bounded
      within-stratum imbalance, such as stratified permuted-block
      randomization and the Hu--Hu procedure. We show that this term is
      driven entirely by the interaction between the randomization factors
      in their effect on survival (Proposition~\ref{prop:additivity}).
      Ignoring it is harmless when the factors act additively, but it
      understates the variance and inflates the family-wise error rate when
      they interact. The stratum-standardized estimator removes the term
      altogether. Its limiting distribution is therefore the same under
      stratified permuted-block randomization and Pocock--Simon
      minimization, and it remains valid when censoring depends on the
      randomization factors. 
      \item[] Second, based on these limits, we propose two families of CAR-adjusted test
      statistics, one built on the unstratified contrasts with a
      design-corrected covariance and one built on the standardized
      contrasts. Their covariance accounts for both the shared control and
      the randomization design. Calibrated by a common Gaussian multiplier
      procedure, they yield max-$T$ tests with strong family-wise error
      control and simultaneous confidence bands, at a fixed restriction time or over a range of restriction times.
\end{itemize}

The remainder of the paper is organized as follows.
Section~\ref{sec:framework} introduces the framework, and Section~\ref{sec:main-results} gives the limit theorems and characterizes
when the minimization design term matters.
Section~\ref{sec:constructed-statistics} defines the test statistics, the multiplier calibration, the simultaneous bands, and the family-wise
error control. Sections~\ref{sec:simulation} and \ref{sec:application} present the simulation study and the ICON6-based application, and Section~\ref{sec:discussion} gives practical recommendations.

\section{Framework}
\label{sec:framework}

\subsection{Multi-arm trial and potential event times}

Suppose that $N$ subjects are assigned to one of $K+1$ study arms. Arm
$a=0$ is the shared control and arms $a=1,\ldots,K$ are the experimental
treatments. Let
\[
A_i\in\{0,1,\ldots,K\},
\qquad
A_{ia}=\ind(A_i=a),
\]
denote the treatment assignment and its corresponding indicators.

Let $Z_i\in\mathcal S=\{1,\ldots,J\}$ be the discrete stratum label formed
from the baseline covariates used by the randomization procedure, where $J$
is fixed. Write
\[
p_s=\Pp(Z_i=s)>0,
\qquad s\in\mathcal S.
\]

For subject $i$ and arm $a$, let $T_i(a)$ be the potential event time and
$C_i(a)$ the potential censoring time. Define
\[
X_i(a)=\min\{T_i(a),C_i(a)\},
\qquad
\Delta_i(a)=\ind\{T_i(a)\leq C_i(a)\}.
\]
Under consistency, the observed data are
\[
X_i=\sum_{a=0}^K A_{ia}X_i(a),
\qquad
\Delta_i=\sum_{a=0}^K A_{ia}\Delta_i(a),
\]
together with $(A_i,Z_i)$.

The marginal survival function under arm $a$ is
\[
S_a(t)=\Pp\{T_i(a)>t\},
\qquad t\geq 0.
\]
For a restriction time $\tau>0$, the arm-specific RMST is
\begin{equation}
\label{eq:rmst}
\mu_a(\tau)
=
\E\{\min(T_i(a),\tau)\}
=
\int_0^\tau S_a(t)\,dt.
\end{equation}
Thus, $\mu_a(\tau)$ is the expected amount of event-free time accumulated
during the first $\tau$ time units under arm $a$.

The treatment effect for experimental arm $k$ relative to the shared control
is
\begin{equation}
\label{eq:rmst-contrast}
\theta_k(\tau)=\mu_k(\tau)-\mu_0(\tau),
\qquad k=1,\ldots,K.
\end{equation}
A positive value indicates a gain in restricted mean survival under treatment
$k$. Let
\[
\bmu(\tau)
=
\{\mu_0(\tau),\mu_1(\tau),\ldots,\mu_K(\tau)\}^{T},
\qquad
\btheta(\tau)
=
\{\theta_1(\tau),\ldots,\theta_K(\tau)\}^{T}.
\]
With $\bm C=[-\bone_K,\ \bm I_K]$, we have $\btheta(\tau)=\bm C\bmu(\tau)$.

We study the entire RMST process on the prespecified compact interval
\[
\calT=[\tau_L,\tau_U],
\qquad
0<\tau_L<\tau_U<\infty.
\]
A fixed restriction time is recovered as a special case. We write
$\ell^\infty(\calT)$ for the space of bounded real functions on $\calT$
with the supremum norm, $\ell^\infty(\calT)^m$ for the product of $m$
copies, and $\rights$ for weak convergence in these spaces
\citep{vaart1996weak}.

\subsection{Covariate-adaptive randomization}

Let
\[
n(s)=\sum_{i=1}^N\ind(Z_i=s),
\qquad
n_a(s)=\sum_{i=1}^N A_{ia}\ind(Z_i=s)
\]
be the total and arm-specific sample sizes in stratum $s$. The target
allocation proportions are denoted by
\[
\pi_a\in(0,1),
\qquad
\sum_{a=0}^K\pi_a=1.
\]
The within-stratum imbalance for arm $a$ is
\begin{equation}
\label{eq:imbalance}
D_{N,a}(s)=n_a(s)-\pi_a n(s).
\end{equation}
Theorem~\ref{thm:bounded-balance-fclt} first treats strongly balancing CAR
procedures satisfying $D_{N,a}(s)=O_p(1)$. This class includes stratified
permuted-block randomization and the multi-arm Hu--Hu procedure with
positive within-stratum weight \citep{hu2012asymptotic,hu2023multi}. A later
extension allows $D_{N,a}(s)$ to have a nondegenerate $\sqrt N$-scale limit,
as under simple randomization and Pocock--Simon minimization
\citep{pocock1975sequential}.

\subsection{Assumptions}

Let
\[
\mathcal O_i^\star
=
\{Z_i,T_i(a),C_i(a):a=0,\ldots,K\}
\]
denote the full potential data for subject $i$.

\begin{assumption}[Sampling]
\label{asu:sampling}
The full-data vectors $\mathcal O_1^\star,\ldots,\mathcal O_N^\star$ are
independent and identically distributed. The number of CAR strata $J$ is
fixed and $\min_{s\in\mathcal S}p_s>0$.
\end{assumption}

\begin{assumption}[Randomization]
\label{asu:randomization}
Conditional on $(Z_1,\ldots,Z_N)$, the treatment allocation sequence is
independent of the potential event and censoring times. Equivalently, for
every measurable allocation event $\mathcal A$,
\[
\Pp\{(A_1,\ldots,A_N)\in\mathcal A
\mid\mathcal O_1^\star,\ldots,\mathcal O_N^\star\}
=
\Pp\{(A_1,\ldots,A_N)\in\mathcal A
\mid Z_1,\ldots,Z_N\}.
\]
The allocation rule targets the fixed proportions
$(\pi_0,\ldots,\pi_K)$.
\end{assumption}

\begin{assumption}[Strong within-stratum balance]
\label{asu:bounded-balance}
For every $a=0,\ldots,K$ and $s\in\mathcal S$,
\[
D_{N,a}(s)=O_p(1).
\]
\end{assumption}

\begin{assumption}[Independent censoring and positivity]
\label{asu:censoring}
For every arm $a$, $T_i(a)$ and $C_i(a)$ are independent and
$\Pp\{T_i(a)=C_i(a)\}=0$. In addition, for some constant $\epsilon>0$,
\[
\inf_{a=0,\ldots,K}\Pp\{X_i(a)\geq\tau_U\}\geq\epsilon.
\]
If the event-time distribution has jumps, its cumulative hazard satisfies
\[
\inf_{a=0,\ldots,K}\inf_{0<t\leq\tau_U}
\{1-\Delta\Lambda_a(t)\}\geq\epsilon.
\]
\end{assumption}

Assumption~\ref{asu:censoring} is the standard condition for the unadjusted
arm-specific Kaplan--Meier estimator. Subsection~\ref{sec:conditional-censoring}
weakens it to independent censoring within the randomization strata. No
further regularity condition is needed: under Assumption~\ref{asu:censoring}
the influence functions defined below are uniformly bounded and Lipschitz in
the restriction time (Lemma~\ref{lem:if-bounds} in the Supplementary Material).

\subsection{Estimators and influence functions}
\label{subsec:influence-function}

Let $\widehat S_a$ be the Kaplan--Meier estimator computed from the subjects
assigned to arm $a$. The estimated RMST and treatment-control contrast are
\begin{equation}
\label{eq:contrast-estimator}
\widehat\mu_a(\tau)=\int_0^\tau\widehat S_a(t)\,dt,
\qquad
\widehat\btheta(\tau)=\bm C\widehat\bmu(\tau),
\qquad
\widehat\bmu(\tau)=\{\widehat\mu_0(\tau),\ldots,\widehat\mu_K(\tau)\}^{T}.
\end{equation}
Write $N_i^{(a)}(t)=\ind\{X_i(a)\leq t,\Delta_i(a)=1\}$ and
$Y_i^{(a)}(t)=\ind\{X_i(a)\geq t\}$ for the potential counting and at-risk
processes, $y_a(t)=\Pp\{X_i(a)\geq t\}$, and $\Lambda_a$ for the cumulative
hazard of $T_i(a)$, and let
$M_i^{(a)}(t)=N_i^{(a)}(t)-\int_{(0,t]}Y_i^{(a)}(u)\,d\Lambda_a(u)$. The
influence function of $\widehat\mu_a(\tau)$ is
\begin{equation}
\label{eq:rmst-if}
\phi_{ia}(\tau)
=
-\int_0^\tau
S_a(t)
\int_{(0,t]}
\frac{dM_i^{(a)}(u)}
{y_a(u)\{1-\Delta\Lambda_a(u)\}}
\,dt,
\end{equation}
where the jump correction equals one for continuous event times. Its
derivation from the Nelson--Aalen and product-integral maps is reviewed in
Section~\ref{sec:supp-if}. The influence function has $\E\{\phi_{ia}(\tau)\}=0$. Its
conditional mean and covariance kernel given the CAR stratum are
\begin{equation}
\label{eq:conditional-mean}
m_a(s,\tau)
=
\E\{\phi_{ia}(\tau)\mid Z_i=s\},
\qquad
c_a(s;\tau,\tau')
=
\Cov\{\phi_{ia}(\tau),\phi_{ia}(\tau')\mid Z_i=s\}.
\end{equation}
Although $\phi_{ia}$ is marginally centered, $m_a(s,\tau)$ is nonzero when the
CAR strata are prognostic, and this is the source of all design effects below.

\section{Limit Theory for the Two RMST Estimators}
\label{sec:main-results}

This section gives the limit theory for the two estimators on which the
inference procedures of Section~\ref{sec:constructed-statistics} are built.
Subsections~\ref{sec:first-result}--\ref{sec:additivity} treat the
unstratified Kaplan--Meier RMST analysis, which is the analysis most often
reported in practice: bounded within-stratum imbalance, general
$\sqrt N$-order imbalance, which includes simple randomization and
Pocock--Simon minimization, and the
conditions under which the resulting design term is negligible.
Subsection~\ref{sec:conditional-censoring} treats the stratum-standardized
estimator, whose limit contains no design term.

\subsection{Joint RMST Process under Strong Balance}
\label{sec:first-result}

\begin{theorem}[Joint functional limit under bounded within-stratum imbalance]
\label{thm:bounded-balance-fclt}
Suppose Assumptions~\ref{asu:sampling}--\ref{asu:censoring} hold. Then, in
$\ell^\infty(\calT)^{K+1}$,
\begin{equation}
\label{eq:mu-fclt}
\sqrt N\{\widehat\bmu(\cdot)-\bmu(\cdot)\}
\rights
\mathbb G_\mu(\cdot),
\end{equation}
where $\mathbb G_\mu=(\mathbb G_{\mu,0},\ldots,\mathbb
G_{\mu,K})^T$ is a mean-zero tight Gaussian process with covariance kernel
\begin{equation}
\label{eq:mu-covariance}
\Gamma_{\mu,ab}(\tau,\tau')
=
\E\{m_a(Z_i,\tau)m_b(Z_i,\tau')\}
+
\ind(a=b)\frac{1}{\pi_a}
\E\{c_a(Z_i;\tau,\tau')\}
\end{equation}
for $a,b=0,\ldots,K$ and $\tau,\tau'\in\calT$.

Consequently, in $\ell^\infty(\calT)^K$,
\begin{equation}
\label{eq:theta-fclt}
\sqrt N\{\widehat\btheta(\cdot)-\btheta(\cdot)\}
\rights
\mathbb G_\theta(\cdot)
=
\bm C\mathbb G_\mu(\cdot).
\end{equation}
For $k,\ell=1,\ldots,K$, the covariance kernel of the
treatment-control process is
\begin{align}
\label{eq:theta-covariance}
\Gamma_{\theta,k\ell}(\tau,\tau')
={}&
\E\!\left[
\{m_k(Z_i,\tau)-m_0(Z_i,\tau)\}
\{m_\ell(Z_i,\tau')-m_0(Z_i,\tau')\}
\right] \notag\\
&+
\ind(k=\ell)\frac{1}{\pi_k}
\E\{c_k(Z_i;\tau,\tau')\}
+
\frac{1}{\pi_0}
\E\{c_0(Z_i;\tau,\tau')\}.
\end{align}
In particular, $\sqrt N\{\widehat\btheta(\tau)-\btheta(\tau)\}$ is
asymptotically normal with covariance matrix
$\{\Gamma_{\theta,k\ell}(\tau,\tau)\}_{k,\ell=1}^K$ at each fixed $\tau\in\calT$.
\end{theorem}

\begin{remark}[No global-null restriction]
\label{rmk:no-global-null}
Theorem~\ref{thm:bounded-balance-fclt} holds for arbitrary
$\{S_0,\ldots,S_K\}$ and therefore for arbitrary values of
$\{\theta_1,\ldots,\theta_K\}$. It is not a result only under
$H_0:\theta_1=\cdots=\theta_K=0$. This is essential for proving strong
family-wise error control under partial-null configurations.
\end{remark}

\begin{remark}[Three sources of covariance]
\label{rmk:covariance-interpretation}
Equation~\eqref{eq:theta-covariance} separates three components. The first
captures random variation in the sampled stratum composition. The second is
the treatment-specific outcome contribution and appears only on the diagonal.
The third is the control-arm contribution and appears in every pair of
treatment-control comparisons because all comparisons use the same control.
Thus, the shared-control dependence is retained rather than removed by
whitening the contrast vector.
\end{remark}

\begin{remark}[Relation to \citet{wang2023model}]
\label{rmk:wang}
For a single arm, the linear representation underlying
Theorem~\ref{thm:bounded-balance-fclt} agrees with the Kaplan--Meier result
of \citet{wang2023model} under stratified and biased-coin randomization, and
the RMST statement follows from theirs by integration over $[0,\tau]$. The
new content of Theorem~\ref{thm:bounded-balance-fclt} is the joint
distribution across all arms under CAR. Joint RMST limits across several
independently sampled groups were derived by \citet{munko2024rmst}, but
without randomization-induced dependence. The covariance in
Theorem~\ref{thm:bounded-balance-fclt} contains the cross-arm
composition term $\E\{m_a(Z_i,\tau)m_b(Z_i,\tau')\}$ and the shared-control
term, which together determine the dependence used by the max-$T$
procedures. Theorem~\ref{thm:general-car-fclt} then goes beyond
bounded-imbalance designs.
\end{remark}

\subsection{Joint RMST Process under General \texorpdfstring{$\sqrt N$}{root-N}-Order Imbalance}
\label{sec:general-car}

Many CAR procedures do not satisfy the bounded-balance condition of
Theorem~\ref{thm:bounded-balance-fclt}. Under simple randomization, every
within-stratum imbalance fluctuates on the $\sqrt N$ scale. Pocock--Simon
minimization strongly controls the overall and marginal imbalances, but its
joint-stratum imbalances may still fluctuate on the $\sqrt N$ scale, with a
limiting covariance that is design-specific and need not be dominated by its
counterpart under simple randomization
\citep{hu2020theory,zhao2024consistent,zhang2026multitreatment}. We therefore
allow a general $\sqrt N$-order imbalance.

Let $L=(K+1)J$ and stack all arm-by-stratum imbalances in the vector
\begin{equation}
\label{eq:scaled-imbalance-vector}
\bm d_N
=
\frac{1}{\sqrt N}
\operatorname{vec}
\left[
\{D_{N,a}(s):s=1,\ldots,J\}:a=0,\ldots,K
\right]
\in\R^L.
\end{equation}
Also define the centered stratum-count vector
\begin{equation}
\label{eq:stratum-count-vector}
\bm W_N
=
\left\{
\frac{n(s)-Np_s}{\sqrt N}:s=1,\ldots,J
\right\}^{T}.
\end{equation}

\begin{assumption}[Gaussian imbalance and asymptotic orthogonality]
\label{asu:general-imbalance}
The within-stratum allocations satisfy
\[
\max_{a,s}|D_{N,a}(s)|=O_p(\sqrt N)
\quad\text{and}\quad
\frac{n_a(s)}{n(s)}\xrightarrow{P}\pi_a.
\]
Moreover,
\begin{equation}
\label{eq:joint-design-clt}
\begin{pmatrix}
\bm W_N\\
\bm d_N
\end{pmatrix}
\xrightarrow{D}
\mathcal N_{J+L}
\left[
\bm 0,
\begin{pmatrix}
\bm\Sigma_Z&\bm 0\\
\bm 0&\bm\Omega_D
\end{pmatrix}
\right],
\qquad
\bm\Sigma_Z
=
\operatorname{diag}(\bm p)-\bm p\bm p^T,
\end{equation}
where $\bm p=(p_1,\ldots,p_J)^T$ and $\bm\Omega_D$ is a finite positive
semidefinite matrix. Thus the stratum-count fluctuation and the
randomization imbalance are asymptotically orthogonal. The matrix
$\bm\Omega_D$ is allowed to be singular because the imbalance vector obeys
exact within-stratum arm-sum constraints and may obey additional asymptotic
constraints induced by the design.
\end{assumption}

For $\tau\in\calT$, define the $(K+1)\times L$ matrix $\bm Q(\tau)$ by
\begin{equation}
\label{eq:Q-matrix}
Q_{a,(b,s)}(\tau)
=
\frac{\ind(a=b)}{\pi_a}m_a(s,\tau),
\qquad
a,b=0,\ldots,K,\quad s=1,\ldots,J,
\end{equation}
where $(b,s)$ indexes the columns, and let $\bm\Gamma_\mu(\tau,\tau')$ denote
the $(K+1)\times(K+1)$ matrix with entries \eqref{eq:mu-covariance}.

\begin{theorem}[Joint RMST process under $\sqrt N$-order CAR imbalance]
\label{thm:general-car-fclt}
Suppose Assumptions~\ref{asu:sampling}, \ref{asu:randomization},
\ref{asu:censoring}, and \ref{asu:general-imbalance} hold. Then
\[
\sqrt N\{\widehat\bmu(\cdot)-\bmu(\cdot)\}
\rights
\mathbb G_\mu^{D}(\cdot)
\quad\text{in }\ell^\infty(\calT)^{K+1},
\]
where $\mathbb G_\mu^{D}$ is a mean-zero tight Gaussian process with matrix
covariance kernel
\begin{equation}
\label{eq:general-mu-kernel}
\bm\Gamma_\mu^{D}(\tau,\tau')
=
\bm\Gamma_\mu(\tau,\tau')
+
\bm Q(\tau)\bm\Omega_D\bm Q(\tau')^T.
\end{equation}
Consequently,
\begin{equation}
\label{eq:general-theta-fclt}
\sqrt N\{\widehat\btheta(\cdot)-\btheta(\cdot)\}
\rights
\mathbb G_\theta^{D}(\cdot)
=
\bm C\mathbb G_\mu^{D}(\cdot)
\quad\text{in }\ell^\infty(\calT)^K,
\end{equation}
with covariance kernel
\begin{equation}
\label{eq:general-theta-kernel}
\bm\Gamma_\theta^{D}(\tau,\tau')
=
\bm C\bm\Gamma_\mu(\tau,\tau')\bm C^T
+
\bm R(\tau)\bm\Omega_D\bm R(\tau')^T,
\qquad
\bm R(\tau)=\bm C\bm Q(\tau).
\end{equation}
The $k$th row of $\bm R(\tau)$ has entries
\begin{equation}
\label{eq:R-row}
R_{k,(a,s)}(\tau)
=
\frac{\ind(a=k)}{\pi_k}m_k(s,\tau)
-
\frac{\ind(a=0)}{\pi_0}m_0(s,\tau).
\end{equation}
The theorem holds under arbitrary treatment effects.
\end{theorem}

\begin{remark}[Application to minimization]
\label{rmk:ps-application}
For equal-allocation, treatment-label-symmetric Pocock--Simon procedures,
the joint-stratum imbalance has a Gaussian limit of the required form
\citep{hu2020theory,hu2023multi,zhang2026multitreatment}. The orthogonality in
\eqref{eq:joint-design-clt} then holds exactly: label symmetry gives
$\E\{D_{N,a}(s)\mid Z_1,\ldots,Z_N\}=0$, so $\Cov(\bm W_N,\bm d_N)=\bm 0$ for
every $N$. Under bounded imbalance, $\bm\Omega_D=\bm 0$ and
Theorem~\ref{thm:general-car-fclt} reduces to
Theorem~\ref{thm:bounded-balance-fclt}.
\end{remark}

\begin{remark}[Simple randomization and other special cases]
\label{rmk:sr-special-case}
Under simple randomization, the assignments are independent of each other
and of the strata, so Assumption~\ref{asu:general-imbalance} holds with
\[
\bm\Omega_D^{\mathrm{SR}}
=
\{\operatorname{diag}(\bm\pi)-\bm\pi\bm\pi^T\}\otimes\operatorname{diag}(\bm p),
\qquad
\bm\pi=(\pi_0,\ldots,\pi_K)^T.
\]
Substituting it into \eqref{eq:general-mu-kernel} gives
$\Gamma_{\mu,ab}^{D}(\tau,\tau')
=\ind(a=b)\Cov\{\phi_{ia}(\tau),\phi_{ia}(\tau')\}/\pi_a$, the familiar
covariance of independent Kaplan--Meier curves.
Theorem~\ref{thm:general-car-fclt} therefore covers simple randomization,
the bounded-balance designs of Theorem~\ref{thm:bounded-balance-fclt}, and
minimization within one framework. It also relaxes the condition of
\citet{bugni2019inference}, which requires the imbalances to be independent
across strata with covariance
$\tau(s)\{\operatorname{diag}(\bm\pi)-\bm\pi\bm\pi^T\}$ for some
$\tau(s)\in[0,1]$. That condition holds for simple and stratified block
randomization but not for minimization.
\end{remark}

\subsection{When Does the Minimization Design Term Matter?}
\label{sec:additivity}

Pocock--Simon minimization balances the treatment counts within each level
of each randomization factor, but not within the joint strata. The next
result shows that, as a consequence, only the non-additive part of the
stratum-specific means $m_a(s,\tau)$ enters the design term.

Suppose the joint stratum is formed from $R$ randomization factors, so that
each $s\in\mathcal S$ corresponds to a level profile
$(s_1,\ldots,s_R)$. For factor $r$ and level $l$, define the marginal
imbalance
\begin{equation}
\label{eq:marginal-imbalance}
D_{N,a}^{(r)}(l)
=
\sum_{s\in\mathcal S:\,s_r=l}D_{N,a}(s).
\end{equation}
A function $f$ on $\mathcal S$ is called additive if
$f(s)=\sum_{r=1}^R f_r(s_r)$ for some functions $f_1,\ldots,f_R$ of the
individual factor levels, and constants are additive.

\begin{proposition}[Design term under marginal balance]
\label{prop:additivity}
Suppose the conditions of Theorem~\ref{thm:general-car-fclt} hold and that
the randomization procedure controls the marginal imbalances, in the sense
that $D_{N,a}^{(r)}(l)=O_p(1)$ for every arm $a$, factor $r$, and level $l$.
For each $a$ and $\tau$, write
$m_a(s,\tau)=m_a^{\mathrm{add}}(s,\tau)+\eta_a(s,\tau)$, where
$m_a^{\mathrm{add}}(\cdot,\tau)$ is any additive function, and let
$\bm Q^{\eta}(\tau)$ denote $\bm Q(\tau)$ in \eqref{eq:Q-matrix} with
$m_a$ replaced by $\eta_a$. Then, for all $\tau,\tau'\in\calT$,
\begin{equation}
\label{eq:interaction-only}
\bm Q(\tau)\bm\Omega_D\bm Q(\tau')^T
=
\bm Q^{\eta}(\tau)\bm\Omega_D\bm Q^{\eta}(\tau')^T.
\end{equation}
In particular, if $m_a(\cdot,\tau)$ is additive for every arm and every
$\tau\in\calT$, the design term vanishes and the covariance kernel in
Theorem~\ref{thm:general-car-fclt} equals the kernel
$\bm\Gamma_\mu$ of Theorem~\ref{thm:bounded-balance-fclt}.
\end{proposition}

Bounded marginal imbalance is the design goal of minimization and has been
established for Pocock--Simon-type procedures under conditions given by
\citet{hu2020theory}, \citet{hu2023multi}, and
\citet{zhang2026multitreatment}, and it can also be checked from the design
simulation. \citet{johnson2024validity} described the asymptotic correlation
matrix of the within-stratum imbalances for equal stratum prevalences. For the four-arm
rule of Section~\ref{sec:simulation}, the
variance of each unscaled marginal imbalance was about $0.45$ at both $N=800$
and $N=1200$, whereas the variance of the scaled interaction contrast of
$\bm d_N$ was about $0.19$, close to its simple-randomization value
$\pi_a(1-\pi_a)=0.1875$. Minimization thus removes the additive directions of
the imbalance but leaves the interaction direction essentially as variable
as under simple randomization.

Proposition~\ref{prop:nondegenerate} shows that, for marginal procedures
that keep every arm possible at every step and return to balance quickly
enough, minimization removes none of the interaction directions. Each of them keeps a nondegenerate variance on
the $\sqrt N$ scale, so that Proposition~\ref{prop:additivity} becomes a
characterization of exactly when the design term vanishes. For $n\leq N$,
let $D_{n,a}(s)$ and $D_{n,a}^{(r)}(l)$ denote the imbalances among the
first $n$ subjects, and let $\bm M_n$ collect the marginal imbalances
$\{D_{n,a}^{(r)}(l)\}$. Let $\mathcal H\subset\R^J$ be the space of additive
functions on $\mathcal S$, let $\bm P_\perp$ be the Euclidean orthogonal
projection of $\R^J$ onto $\mathcal H^\perp$, and let
$\eta_a^\perp(\cdot,\tau)=\bm P_\perp m_a(\cdot,\tau)$ be the residual of
the unweighted least-squares fit of $m_a(\cdot,\tau)$ by an additive
function. The changes of the within-stratum imbalances that leave every
marginal imbalance unchanged form the subspace
\[
\mathcal V
=
\Bigl\{\bm v\in\R^L:\
\sum_{a=0}^Kv_{(a,s)}=0\ \text{for all }s,\ \
\sum_{s:\,s_r=l}v_{(a,s)}=0\ \text{for all }a,r,l\Bigr\},
\]
which is trivial only when every function on $\mathcal S$ is additive, as
with a single factor.

\begin{proposition}[Nondegenerate design term under minimization]
\label{prop:nondegenerate}
Suppose the conditions of Theorem~\ref{thm:general-car-fclt} hold and
$\mathcal V\neq\{\bm 0\}$. Suppose also that the randomization procedure
satisfies the following conditions.
\begin{enumerate}[label=(\roman*)]
\item There are functions $\chi_a$ such that
$\Pp(A_i=a\mid Z_1,\ldots,Z_N,A_1,\ldots,A_{i-1})=\chi_a(Z_i,\bm M_{i-1})$
for every $i$ and $a$, and $\chi_a\geq\delta$ for some $\delta>0$.
\item The return time $\varrho=\min\{n\geq1:\bm M_n=\bm 0\}$ to the
balanced state satisfies $\E(\varrho^2)<\infty$.
\end{enumerate}
Let $\bm\zeta\in\R^L$ stack the imbalances $\{D_{\varrho,a}(s)\}$. Then
$\bm\Omega_D=\E(\bm\zeta\bm\zeta^T)/\E(\varrho)$, and the range of
$\bm\Omega_D$ is exactly $\mathcal V$. Consequently, if $\lambda_-$ and
$\lambda_+$ denote the smallest and largest positive eigenvalues of
$\bm\Omega_D$, then for every arm $a$ and every $\tau\in\calT$,
\begin{equation}
\label{eq:design-term-bounds}
\lambda_-\frac{K\|\eta_a^\perp(\cdot,\tau)\|^2}{(K+1)\pi_a^2}
\leq
\{\bm Q(\tau)\bm\Omega_D\bm Q(\tau)^T\}_{aa}
\leq
\lambda_+\frac{K\|\eta_a^\perp(\cdot,\tau)\|^2}{(K+1)\pi_a^2},
\end{equation}
where $\|\cdot\|$ is the Euclidean norm on $\R^J$. In particular, the
design term in the variance of $\widehat\mu_a(\tau)$ vanishes if and only
if $m_a(\cdot,\tau)$ is additive. The design term in the variance of
$\widehat\theta_k(\tau)$ vanishes if and only if $m_k(\cdot,\tau)$ and
$m_0(\cdot,\tau)$ are both additive when $K\geq2$, and if and only if
$m_1(\cdot,\tau)/\pi_1+m_0(\cdot,\tau)/\pi_0$ is additive when $K=1$.
\end{proposition}

Conditions~(i) and~(ii) imply that the marginal imbalances are bounded in
probability, because each subject changes them by a bounded amount and the
time since the last visit to the balanced state is tight when
$\E(\varrho)<\infty$. Proposition~\ref{prop:additivity} therefore applies.
Condition~(i) holds for Pocock--Simon rules whose imbalance measure is
computed from the deviations of the arm counts from their targets at each
level, such as the range, the variance, or the sum of absolute deviations
under equal allocation, provided every arm keeps a positive assignment
probability and the rule does not change with the subject index. Under
unequal allocation, these conditions do not by themselves imply
Assumption~\ref{asu:general-imbalance}, which must be checked separately.
Under equal allocation, and when every combination of factor levels has
positive probability, condition~(ii) holds for Pocock--Simon rules that use
a weighted variance with positive weights as imbalance measure and whose
assignment probabilities do not increase with the score and separate a
minimizing arm from a maximizing arm by a fixed margin. Examples are the
rule in which the minimizing arms share a probability $p\in(K/(K+1),1)$ and
the ranked rule of \citet{hu2020theory}, and for these rules $\varrho$ even
has an exponential moment (Lemma~\ref{lem:return-time-variance}). For the
sum of absolute deviations used in Section~\ref{sec:simulation}, we
checked condition~(ii) numerically: the simulated return times had mean
590 and a tail consistent with geometric decay
(Section~\ref{sec:supp-return-variance}).
For the variance measure, \citet[Theorem~4.1]{hu2020theory} showed that
every within-stratum imbalance has a positive asymptotic variance.
Proposition~\ref{prop:nondegenerate} extends this to every direction of
$\mathcal V$ and to any procedure satisfying (i) and (ii), and it gives a
renewal formula for $\bm\Omega_D$. For two arms, \citet{johnson2024validity}
showed under equal stratum prevalences that the imbalance covariance
vanishes outside the interaction directions $\mathcal V$, and
\citet{kuznetsova2027explicit} conjectured from simulations that it is
proportional to the projection onto $\mathcal V$.
Proposition~\ref{prop:nondegenerate} proves that its range is exactly
$\mathcal V$ for any number of arms and factors. Because $\bm P_\perp\eta_a=\eta_a^\perp$
for every remainder $\eta_a$ in Proposition~\ref{prop:additivity}, the
upper bound in \eqref{eq:design-term-bounds} also holds with $\eta_a$ in
place of $\eta_a^\perp$. With two binary factors,
$\|\eta_a^\perp(\cdot,\tau)\|=|m_a(1,\tau)-m_a(2,\tau)-m_a(3,\tau)+m_a(4,\tau)|/2$.
The bounds show that the variance added by minimization is of the same
order as the squared interaction component of $m_a(\cdot,\tau)$, so it can
be ignored only when that component is small.

\begin{remark}[Label-symmetric procedures]
\label{rmk:label-symmetric}
Suppose the procedure in Proposition~\ref{prop:nondegenerate} also uses
equal allocation and is label-symmetric, that is,
$\chi_{\sigma(a)}(s,\sigma\bm m)=\chi_a(s,\bm m)$ for every permutation
$\sigma$ of the arms, where $\sigma\bm m$ permutes the arm index of
$\bm m$. Then the distribution of $\bm d_N$ is invariant under permutations
of the arms, and Schur's lemma gives
$\bm\Omega_D=\{\bm I_{K+1}-\bone\bone^T/(K+1)\}\otimes\bm G$ for a
symmetric matrix $\bm G$ that vanishes on $\mathcal H$ and is positive
definite on $\mathcal H^\perp$. The design term for arm $a$ is then exactly
$K(K+1)\,\eta_a^\perp(\cdot,\tau)^T\bm G\,\eta_a^\perp(\cdot,\tau)$. With
two binary factors, $\mathcal H^\perp$ is spanned by the interaction
contrast and $\bm G$ reduces to a single constant $\omega>0$, so the design
term equals $K(K+1)\,\omega\,\|\eta_a^\perp(\cdot,\tau)\|^2$. For the
four-arm rule of Section~\ref{sec:simulation}, the design simulation gave
$\widehat\omega=0.0649$ at $N=800$ and $0.0653$ at $N=1200$, and
$5.4\times10^5$ simulated return cycles gave $0.0646$ (standard error
$0.0001$). The last value, which has the smallest Monte Carlo error, is
about 3\% above the value $0.0625$ that simple randomization places on the
same directions.
\end{remark}

\begin{remark}[Practical interpretation]
\label{rmk:additivity-interpretation}
Whether the design term matters is therefore a question about prognosis.
With two binary factors and strata ordered by the levels $(0,0)$, $(0,1)$,
$(1,0)$, $(1,1)$, the remainder $\eta_a$ is proportional to the interaction
contrast $m_a(1,\tau)-m_a(2,\tau)-m_a(3,\tau)+m_a(4,\tau)$. Under
approximately additive prognosis the bounded-balance kernel is adequate even
under minimization. Under strong interaction, for example a biomarker that
is prognostic in only one disease stage, omitting $\bm\Omega_D$ understates
the variance (Section~\ref{sec:simulation}). The standardized estimator of
Subsection~\ref{sec:conditional-censoring} avoids the issue.
\end{remark}

\subsection{The Stratum-Standardized Estimator}
\label{sec:conditional-censoring}

The unadjusted estimator identifies the marginal survival curve under
Assumption~\ref{asu:censoring}. In practice, censoring may depend on baseline
severity, region, or other factors used by the randomization procedure, so
that independence holds within each randomization stratum but not
marginally. We therefore also study a stratum-standardized estimator. Besides
weakening the censoring assumption, standardization has an important design
consequence: the stratum-specific influence functions are conditionally
centered, so the design term of Theorem~\ref{thm:general-car-fclt} does not
arise.

For arm $a$ and stratum $s$, define
\begin{align}
\label{eq:stratum-survival-rmst}
S_a(t\mid s)
&=
\Pp\{T_i(a)>t\mid Z_i=s\},&
\mu_a(s,\tau)
&=
\int_0^\tau S_a(t\mid s)\,dt.
\end{align}
The marginal RMST remains the target estimand and satisfies
\begin{equation}
\label{eq:standardization-identity}
\mu_a(\tau)
=
\sum_{s=1}^Jp_s\mu_a(s,\tau).
\end{equation}
Let
\begin{equation}
\label{eq:stratum-rmst-deviation}
g_a(s,\tau)
=
\mu_a(s,\tau)-\mu_a(\tau).
\end{equation}
Thus $\sum_sp_sg_a(s,\tau)=0$.

\begin{assumption}[Conditional independent censoring]
\label{asu:conditional-censoring}
For every arm $a$ and stratum $s$,
\[
T_i(a)\perp C_i(a)\mid Z_i=s
\quad\text{and}\quad
\Pp\{T_i(a)=C_i(a)\mid Z_i=s\}=0.
\]
For some $\epsilon>0$,
\[
\min_{a,s}
\Pp\{X_i(a)\geq\tau_U\mid Z_i=s\}
\geq\epsilon.
\]
If a stratum-specific event-time distribution has jumps, its cumulative
hazard $\Lambda_{a,s}$ satisfies
\[
\min_{a,s}\inf_{0<t\leq\tau_U}
\{1-\Delta\Lambda_{a,s}(t)\}
\geq\epsilon.
\]
\end{assumption}

For a subject in stratum $s$, let $\phi_{ia}^{s}(\tau)$ denote the RMST
influence function \eqref{eq:rmst-if} computed with
$(S_a,\Lambda_a,y_a)$ replaced by the stratum-specific quantities
$S_a(\cdot\mid s)$, its cumulative hazard $\Lambda_{a,s}$, and
$y_{a,s}(t)=\Pp\{X_i(a)\geq t\mid Z_i=s\}$ (Section~\ref{sec:supp-standardized-details}). Under
Assumption~\ref{asu:conditional-censoring}, it is conditionally centered,
$\E\{\phi_{ia}^{s}(\tau)\mid Z_i=s\}=0$, and we write
\begin{equation}
\label{eq:stratum-if-covariance}
v_a(s;\tau,\tau')
=
\E\{\phi_{ia}^{s}(\tau)\phi_{ia}^{s}(\tau')\mid Z_i=s\}.
\end{equation}

Let $\widehat S_a(\cdot\mid s)$ be the Kaplan--Meier estimator computed from
the subjects with $A_i=a$ and $Z_i=s$, and let
$\widehat\mu_a(s,\tau)=\int_0^\tau\widehat S_a(t\mid s)\,dt$.
With $\widehat p_s=n(s)/N$, define
\begin{align}
\label{eq:standardized-rmst-estimator}
\widehat\mu_a^{\mathrm{st}}(\tau)
&=
\sum_{s=1}^J
\widehat p_s\widehat\mu_a(s,\tau),\\
\label{eq:standardized-contrast-estimator}
\widehat\btheta^{\mathrm{st}}(\tau)
&=
\bm C\widehat\bmu^{\mathrm{st}}(\tau).
\end{align}
The same empirical stratum weights are used for every arm. Consequently,
the treatment comparisons refer to one common target population rather than
to different arm-specific mixtures of the strata. The survival curve
$\sum_s\widehat p_s\widehat S_a(\cdot\mid s)$ is the direct-adjusted
Kaplan--Meier estimator of \citet{chang1982corrected} and
\citet{cupples1995age}, and because the RMST is linear in the survival
function, $\widehat\mu_a^{\mathrm{st}}$ is its RMST. Related estimators,
reviewed by \citet{denz2023comparison}, include the at-risk-weighted direct
adjustment of \citet{gregory1988adjusting} and \citet{nieto1996adjusting},
the reweighted risk-set estimator of \citet{amato1988generalized}, and
inverse-probability-weighted versions
\citep{cole2004adjusted,xie2005adjusted}. These works treat subjects as independently sampled and do
not address covariate-adaptive randomization. Our contribution is the joint
process-level distribution of the standardized contrasts across several
arms under CAR, including minimization, when censoring is independent only
within the randomization strata.

\begin{theorem}[Standardized RMST process under conditional censoring]
\label{thm:standardized-fclt}
Suppose Assumptions~\ref{asu:sampling},
\ref{asu:randomization}, and \ref{asu:conditional-censoring} hold. Suppose also that
\begin{equation}
\label{eq:allocation-consistency-standardized}
\frac{n_a(s)}{n(s)}
\xrightarrow{P}\pi_a
\qquad
\text{for every }a=0,\ldots,K\text{ and }s=1,\ldots,J.
\end{equation}
Condition \eqref{eq:allocation-consistency-standardized} follows under
either Assumption~\ref{asu:bounded-balance} or
Assumption~\ref{asu:general-imbalance}.

Then, uniformly in $\tau\in\calT$ and jointly over arms,
\begin{equation}
\label{eq:standardized-linear-representation}
\sqrt N
\{\widehat\mu_a^{\mathrm{st}}(\tau)-\mu_a(\tau)\}
=
\frac{1}{\sqrt N}
\sum_{i=1}^N
\left[
g_a(Z_i,\tau)
+
\frac{A_{ia}}{\pi_a}
\phi_{ia}^{Z_i}(\tau)
\right]
+o_p(1).
\end{equation}
Consequently,
\[
\sqrt N
\{\widehat\bmu^{\mathrm{st}}(\cdot)-\bmu(\cdot)\}
\rights
\mathbb G_\mu^{\mathrm{st}}(\cdot)
\quad\text{in }\ell^\infty(\calT)^{K+1},
\]
where the limiting mean-zero Gaussian process has covariance kernel
\begin{equation}
\label{eq:standardized-mu-kernel}
\Gamma_{\mu,ab}^{\mathrm{st}}(\tau,\tau')
=
\E\{g_a(Z_i,\tau)g_b(Z_i,\tau')\}
+
\ind(a=b)\frac{1}{\pi_a}
\E\{v_a(Z_i;\tau,\tau')\}.
\end{equation}
Furthermore,
\[
\sqrt N
\{\widehat\btheta^{\mathrm{st}}(\cdot)-\btheta(\cdot)\}
\rights
\mathbb G_\theta^{\mathrm{st}}(\cdot)
=
\bm C\mathbb G_\mu^{\mathrm{st}}(\cdot)
\quad\text{in }\ell^\infty(\calT)^K.
\]
For $k,\ell=1,\ldots,K$,
\begin{align}
\label{eq:standardized-theta-kernel}
\Gamma_{\theta,k\ell}^{\mathrm{st}}(\tau,\tau')
={}&
\E\!\left[
\{g_k(Z_i,\tau)-g_0(Z_i,\tau)\}
\{g_\ell(Z_i,\tau')-g_0(Z_i,\tau')\}
\right]\notag\\
&+
\ind(k=\ell)\frac{1}{\pi_k}
\E\{v_k(Z_i;\tau,\tau')\}
+
\frac{1}{\pi_0}
\E\{v_0(Z_i;\tau,\tau')\}.
\end{align}
These results hold under arbitrary treatment effects.
\end{theorem}

\begin{remark}[Why the Pocock--Simon term disappears]
\label{rmk:standardized-no-imbalance}
The expansion \eqref{eq:standardized-linear-representation} contains no
term of the form $\bm Q(\tau)\bm d_N$. Within each stratum, the KM influence
function $\phi_{ia}^{s}$ has conditional mean zero, while the empirical
stratum weights account directly for random variation in the composition of
the enrolled population. Therefore, a $\sqrt N$-order fluctuation in
$n_a(s)-\pi_an(s)$ changes only the number of conditionally centered
observations contributing from that stratum and does not create an
additional first-order mean component. In particular,
$\bm\Omega_D$ need not be estimated for the standardized analysis, even
when Pocock--Simon minimization was used. This parallels the design-robust
behavior of stratum-standardized average-treatment-effect estimators for
noncensored outcomes under minimization \citep{ye2022inference} and the
universal applicability of covariate-adjusted log-rank tests across CAR
designs \citep{ye2024covariate}.
\end{remark}

\section{Inference Procedures}
\label{sec:constructed-statistics}

This section proposes two CAR-adjusted RMST testing procedures, one for each
estimator of Section~\ref{sec:main-results}. The first, the CAR-adjusted
marginal procedure, studentizes the unadjusted Kaplan--Meier contrast that is
reported in practice by a standard error that is valid under CAR. The second,
which we recommend as the primary analysis, does the same for the
stratum-standardized contrast. Subsection~\ref{sec:car-adjusted-statistics}
defines both statistics together with their computable standard errors.
Subsections~\ref{sec:multiplier} and \ref{sec:fwer} give the multiplier
calibration, simultaneous bands, and strong family-wise error control, which
are common to both procedures, and Subsection~\ref{sec:recommendation}
discusses the choice between them.

\subsection{The Proposed Test Statistics}
\label{sec:car-adjusted-statistics}
\label{sec:covariance-estimation}

Both procedures studentize a vector of treatment-control contrasts by a
plug-in estimate of its limiting covariance kernel. Throughout,
$\widehat p_s=n(s)/N$.

\paragraph{CAR-adjusted marginal statistics.}
Let $\widehat\phi_{ia}(\tau)$ be obtained from \eqref{eq:rmst-if} by replacing
$(S_a,\Lambda_a,y_a)$ with the arm-specific Kaplan--Meier, Nelson--Aalen,
and empirical at-risk estimates. For each treatment--stratum cell, let
$\widehat m_a(s,\tau)$ and $\widehat c_a(s;\tau,\tau')$ be the sample mean and
covariance of $\widehat\phi_{ia}$ over subjects with $A_i=a$ and $Z_i=s$, and
let $\widetilde m_a(s,\tau)=\widehat m_a(s,\tau)-\sum_{s'}\widehat p_{s'}\widehat
m_a(s',\tau)$. The design covariance $\bm\Omega_D$ depends only on the stratum
distribution and the prespecified randomization algorithm, so we estimate it
by reproducing the design: in each of $B_{\mathrm{MC}}$ replications, $N$
stratum profiles are drawn from $\widehat{\bm p}$, the trial's exact
randomization algorithm is applied, and the scaled imbalance vector $\bm d_N$
is recorded. The Monte Carlo covariance $\widehat{\bm\Omega}_D$ is the multi-arm
analogue of the estimator of \citet{zhao2024consistent}, and
$\widehat{\bm\Omega}_D=\bm 0$ for bounded-balance designs. Finally,
$\widehat{\bm Q}(\tau)$ replaces $m_a$ in \eqref{eq:Q-matrix} by
$\widetilde m_a$.

\begin{definition}[CAR-adjusted marginal RMST statistics]
\label{def:car-adjusted-statistics}
Estimate the covariance kernel of Theorem~\ref{thm:general-car-fclt} by
\begin{align}
\label{eq:base-kernel-estimator}
\widehat\Gamma_{\mu,ab}(\tau,\tau')
&=
\sum_{s=1}^J\widehat p_s
\widetilde m_a(s,\tau)\widetilde m_b(s,\tau')
+
\ind(a=b)\frac{1}{\pi_a}
\sum_{s=1}^J\widehat p_s
\widehat c_a(s;\tau,\tau'),\\
\label{eq:general-theta-kernel-estimator}
\widehat{\bm\Gamma}_\theta^{D}(\tau,\tau')
&=
\bm C\left\{
\widehat{\bm\Gamma}_\mu(\tau,\tau')
+
\widehat{\bm Q}(\tau)\widehat{\bm\Omega}_D\widehat{\bm Q}(\tau')^T
\right\}\bm C^T,
\end{align}
and let
\begin{equation}
\label{eq:car-adjusted-se}
\widehat\sigma_{k,\mathrm{CAR}}^2(\tau)
=
\widehat\Gamma_{\theta,kk}^{D}(\tau,\tau),
\qquad k=1,\ldots,K.
\end{equation}
For a prespecified restriction time $\tau^\star$, the fixed-time and
process-wise statistics are
\begin{equation}
\label{eq:fixed-statistic}
V_{k,\mathrm{CAR}}^{\mathrm{fix}}
=
\frac{\sqrt N\,\widehat\theta_k(\tau^\star)}
{\widehat\sigma_{k,\mathrm{CAR}}(\tau^\star)},
\qquad
V_{k,\mathrm{CAR}}^{\mathrm{proc}}
=
\sup_{\tau\in\calT}
\frac{\sqrt N\,\widehat\theta_k(\tau)}
{\widehat\sigma_{k,\mathrm{CAR}}(\tau)},
\end{equation}
and the family-level statistics are
\begin{equation}
\label{eq:car-family-statistics}
\mathcal V_{\mathrm{CAR}}^{h}
=
\max_{1\leq k\leq K}V_{k,\mathrm{CAR}}^{h},
\qquad
h\in\{\mathrm{fix},\mathrm{proc}\}.
\end{equation}
These definitions are for the one-sided alternative that a larger RMST is
beneficial. Absolute values give the corresponding two-sided statistics.
\end{definition}

The standard error \eqref{eq:car-adjusted-se} includes variation from the
sampled stratum composition, the treatment-specific outcomes, the shared
control, and, under $\sqrt N$-order imbalance, the design term. Section~\ref{sec:supp-prelim}
gives the explicit cell-level formulas, the finite-sample
degrees-of-freedom correction used by default, and the design-simulation
algorithm.

\begin{proposition}[Uniform consistency of the covariance estimator]
\label{prop:covariance-consistency}
Suppose the assumptions of Theorem~\ref{thm:general-car-fclt} hold. Assume
also that the design covariance is locally stable: for every deterministic
sequence of stratum probability vectors $\bm p_N$ satisfying
$\sqrt N(\bm p_N-\bm p)=O(1)$, the covariance of $\bm d_N$ generated under
$\bm p_N$ converges to $\bm\Omega_D$, with uniformly bounded fourth moments.
If $B_{\mathrm{MC}}\to\infty$, then
\[
\widehat{\bm\Omega}_D\xrightarrow{P}\bm\Omega_D
\]
and
\begin{equation}
\label{eq:uniform-kernel-consistency}
\sup_{\tau,\tau'\in\calT}
\left\|
\widehat{\bm\Gamma}_\theta^{D}(\tau,\tau')
-
\bm\Gamma_\theta^{D}(\tau,\tau')
\right\|_{\max}
\xrightarrow{P}0.
\end{equation}
The same conclusion holds for the arm-specific kernel
$\widehat{\bm\Gamma}_\mu^{D}$.
\end{proposition}

For two-arm Pocock--Simon minimization, \citet{zhao2024consistent} proved
this local-stability property by Le Cam's third lemma. We assume its
multi-arm analogue. The property holds trivially for strongly balancing
designs. It can also be checked numerically by recomputing
$\widehat{\bm\Omega}_D$ at perturbed stratum frequencies. The design
simulation is performed once per trial and reused for every restriction
time.

\paragraph{Stratum-standardized statistics.}
Let $\widehat\pi_a(s)=n_a(s)/n(s)$, let $\widehat\phi_{ia}^{s}$ be the plug-in
version of $\phi_{ia}^{s}$ computed within cell $(a,s)$, let
$\widehat v_a(s;\tau,\tau')$ be its within-cell sample covariance, and let
$\widehat g_a(s,\tau)=\widehat\mu_a(s,\tau)-\widehat\mu_a^{\mathrm{st}}(\tau)$.

\begin{definition}[Stratum-standardized RMST statistics]
\label{def:standardized-statistics}
Estimate the covariance kernel of Theorem~\ref{thm:standardized-fclt} by
\begin{equation}
\label{eq:standardized-theta-kernel-estimator}
\widehat\Gamma_{\mu,ab}^{\mathrm{st}}(\tau,\tau')
=
\sum_{s=1}^J\widehat p_s\widehat g_a(s,\tau)\widehat g_b(s,\tau')
+
\ind(a=b)\sum_{s=1}^J\frac{\widehat p_s}{\widehat\pi_a(s)}
\widehat v_a(s;\tau,\tau'),
\qquad
\widehat{\bm\Gamma}_\theta^{\mathrm{st}}
=
\bm C\widehat{\bm\Gamma}_\mu^{\mathrm{st}}\bm C^T,
\end{equation}
and let
\begin{equation}
\label{eq:standardized-se}
\widehat\sigma_{k,\mathrm{st}}^2(\tau)
=
\widehat\Gamma_{\theta,kk}^{\mathrm{st}}(\tau,\tau).
\end{equation}
The standardized fixed-time and process-wise statistics are
\begin{equation}
\label{eq:standardized-fixed-statistic}
V_{k,\mathrm{st}}^{\mathrm{fix}}
=
\frac{\sqrt N\,\widehat\theta_k^{\mathrm{st}}(\tau^\star)}
{\widehat\sigma_{k,\mathrm{st}}(\tau^\star)},
\qquad
V_{k,\mathrm{st}}^{\mathrm{proc}}
=
\sup_{\tau\in\calT}
\frac{\sqrt N\,\widehat\theta_k^{\mathrm{st}}(\tau)}
{\widehat\sigma_{k,\mathrm{st}}(\tau)},
\end{equation}
and the family-level statistics are
\begin{equation}
\label{eq:standardized-family-statistic}
\mathcal V_{\mathrm{st}}^{h}
=
\max_{1\leq k\leq K}V_{k,\mathrm{st}}^{h},
\qquad
h\in\{\mathrm{fix},\mathrm{proc}\}.
\end{equation}
\end{definition}

Unlike the marginal statistics, these statistics change both the numerator,
which standardizes the treatment--stratum RMST estimates to one common target
population, and its covariance, which needs no design term and no design
simulation. Their consistency is part of
Proposition~\ref{prop:standardized-inference} below.

Each procedure has two inseparable components. First, each contrast is
studentized by its CAR-valid standard error. Second, the joint critical
value is obtained from a multiplier process for the complete vector of
contrasts, so the off-diagonal dependence created by the shared control and
the randomization is retained. We deliberately do not whiten the contrast
vector by an inverse covariance square root: whitening would mix
scientifically distinct treatment-control estimands, whereas max-$T$
calibration preserves their interpretation and exploits their dependence.

\subsection{CAR-Aware Gaussian Multiplier Calibration}
\label{sec:multiplier}

Critical values are obtained from a multiplier process that mirrors the
decomposition in the proofs of Theorems~\ref{thm:bounded-balance-fclt} and
\ref{thm:general-car-fclt}. For subject $i$ in stratum $Z_i$, let
\begin{equation}
\label{eq:bootstrap-summand}
\widehat u_{ia}(\tau)
=
\frac{A_{ia}}{\pi_a}
\{\widehat\phi_{ia}(\tau)-\widehat m_a(Z_i,\tau)\}
+
\widetilde m_a(Z_i,\tau).
\end{equation}
Let $\xi_1,\ldots,\xi_N$ be independent multipliers, independent of the data
$\mathcal D_N$, with mean zero, variance one, and a finite moment of order
$2+\eta$ for some $\eta>0$, and independently draw
$\bm D^*\mid\mathcal D_N\sim\mathcal N_L(\bm 0,\widehat{\bm\Omega}_D)$. Define
\begin{equation}
\label{eq:bootstrap-mu-process}
\mathbb G_{\mu,a}^*(\tau)
=
\frac{1}{\sqrt N}
\sum_{i=1}^N\xi_i\widehat u_{ia}(\tau)
+
\widehat{\bm Q}_a(\tau)\bm D^*,
\qquad
\mathbb G_\theta^*(\tau)=\bm C\mathbb G_\mu^*(\tau),
\end{equation}
where $\widehat{\bm Q}_a(\tau)$ is row $a$ of $\widehat{\bm Q}(\tau)$. With
the finite-sample-corrected summand of Section~\ref{sec:supp-multiplier-details} and standard-normal
multipliers, $\mathbb G_\theta^*$ is, on a finite time grid, exactly
conditionally $\mathcal N(\bm 0,\widehat{\bm\Gamma}_\theta^D)$, so one may
sample directly from that distribution. Section~\ref{sec:supp-multiplier-details} also gives a
step-by-step algorithm.

The standardized multiplier process $\mathbb G_\theta^{\mathrm{st},*}$ is
defined as in \eqref{eq:bootstrap-mu-process} with summand
\begin{equation}
\label{eq:standardized-bootstrap-process}
\widehat u_{ia}^{\mathrm{st}}(\tau)
=
\widehat g_a(Z_i,\tau)
+
\frac{A_{ia}}{\widehat\pi_a(Z_i)}
\Bigl\{\widehat\phi_{ia}^{Z_i}(\tau)
-\frac{1}{n_a(Z_i)}\sum_{j:A_j=a,Z_j=Z_i}\widehat\phi_{ja}^{Z_i}(\tau)\Bigr\}
\end{equation}
and without a design draw $\bm D^*$. With the finite-sample correction
described in Section~\ref{sec:supp-standardized-details}, its empirical covariance reproduces
\eqref{eq:standardized-theta-kernel-estimator} exactly.

The design draw $\bm D^*$ is essential. Resampling subjects as if they
were independent reproduces the outcome and stratum-composition parts of the
limit but not the dependence created by the randomization. Under
minimization it would therefore omit the design term of
Theorem~\ref{thm:general-car-fclt}. The next theorem shows that adding
$\bm D^*$ restores the correct limit.

\begin{theorem}[Conditional validity of the CAR-aware multiplier process]
\label{thm:multiplier-validity}
Suppose the assumptions of
Theorem~\ref{thm:general-car-fclt} and
Proposition~\ref{prop:covariance-consistency} hold. Then, conditionally on
the data and in probability,
\begin{equation}
\label{eq:conditional-bootstrap-convergence}
\mathbb G_\theta^*(\cdot)
\rights
\mathbb G_\theta^D(\cdot)
\quad\text{in }\ell^\infty(\calT)^K.
\end{equation}
More precisely, for the bounded-Lipschitz class
$\mathrm{BL}_1\{\ell^\infty(\calT)^K\}$,
\[
\sup_{f\in\mathrm{BL}_1}
\left|
\E^*\{f(\mathbb G_\theta^*)\}
-
\E\{f(\mathbb G_\theta^D)\}
\right|
\xrightarrow{P}0,
\]
where $\E^*(\cdot)$ denotes expectation over the multipliers and
$\bm D^*$ conditional on the data.

Let the population variance function be
\begin{equation}
\label{eq:sigma-def}
\sigma_k^2(\tau)
=
\Gamma_{\theta,kk}^D(\tau,\tau).
\end{equation}
If
\[
\inf_{1\leq k\leq K}\inf_{\tau\in\calT}\sigma_k(\tau)>0,
\]
then the studentized multiplier process is also conditionally valid. In
particular, if the distribution of
\[
\mathcal M
=
\max_{1\leq k\leq K}\sup_{\tau\in\calT}
\left|
\frac{\mathbb G_{\theta,k}^D(\tau)}{\sigma_k(\tau)}
\right|
\]
is continuous at its $(1-\alpha)$ quantile, then the conditional
$(1-\alpha)$ quantile $c_{1-\alpha}^*$ of
\begin{equation}
\label{eq:bootstrap-max-abs}
\mathcal M^*
=
\max_{1\leq k\leq K}\sup_{\tau\in\calT}
\left|
\frac{\mathbb G_{\theta,k}^*(\tau)}
{\widehat\sigma_{k,\mathrm{CAR}}(\tau)}
\right|
\end{equation}
converges in probability to the corresponding quantile of $\mathcal M$.
\end{theorem}

\begin{proposition}[Validity of standardized covariance and multiplier inference]
\label{prop:standardized-inference}
Under the conditions of
Theorem~\ref{thm:standardized-fclt},
\begin{equation}
\label{eq:standardized-kernel-consistency}
\sup_{\tau,\tau'\in\calT}
\left\|
\widehat{\bm\Gamma}_\theta^{\mathrm{st}}(\tau,\tau')
-
\bm\Gamma_\theta^{\mathrm{st}}(\tau,\tau')
\right\|_{\max}
\xrightarrow{P}0.
\end{equation}
Conditionally on the data and in probability,
\[
\mathbb G_\theta^{\mathrm{st},*}(\cdot)
\rights
\mathbb G_\theta^{\mathrm{st}}(\cdot)
\quad\text{in }\ell^\infty(\calT)^K.
\]
The same conclusions hold for the arm-specific covariance kernel and
multiplier process.
\end{proposition}

\subsection{Simultaneous Bands and Strong Family-Wise Error Control}
\label{sec:fwer}

\begin{corollary}[Simultaneous RMST confidence bands]
\label{cor:simultaneous-bands}
Under the conditions of Theorem~\ref{thm:multiplier-validity}, including
the variance lower bound and the continuity condition for $\mathcal M$, the
bands
\begin{equation}
\label{eq:simultaneous-bands}
\mathcal B_k(\tau)
=
\left[
\widehat\theta_k(\tau)
-
\frac{c_{1-\alpha}^*\widehat\sigma_{k,\mathrm{CAR}}(\tau)}{\sqrt N},
\quad
\widehat\theta_k(\tau)
+
\frac{c_{1-\alpha}^*\widehat\sigma_{k,\mathrm{CAR}}(\tau)}{\sqrt N}
\right],
\quad
k=1,\ldots,K,\quad \tau\in\calT,
\end{equation}
have joint asymptotic coverage $1-\alpha$. That is,
\[
\Pp\left\{
\theta_k(\tau)\in\mathcal B_k(\tau)
\text{ for every }k=1,\ldots,K\text{ and }\tau\in\calT
\right\}
\longrightarrow 1-\alpha.
\]
\end{corollary}

In computation, the supremum over $\calT$ is replaced by the maximum over a
prespecified grid whose mesh tends to zero, which does not change the
limiting critical value. A clinically chosen finite set of restriction
times gives an ordinary finite-dimensional simultaneous region.

We distinguish a confirmatory analysis at a prespecified restriction time
from simultaneous inference over a range of restriction times.

\paragraph{Families of null hypotheses.}

For a prespecified $\tau^\star\in\calT$, consider the one-sided hypotheses
\begin{equation}
\label{eq:fixed-hypotheses}
H_{0k}^{\mathrm{fix}}:
\theta_k(\tau^\star)\leq0
\quad\text{versus}\quad
H_{1k}^{\mathrm{fix}}:
\theta_k(\tau^\star)>0,
\qquad k=1,\ldots,K.
\end{equation}
For inference on the complete RMST process, consider
\begin{equation}
\label{eq:process-hypotheses}
H_{0k}^{\mathrm{proc}}:
\theta_k(\tau)\leq0\text{ for every }\tau\in\calT
\quad\text{versus}\quad
H_{1k}^{\mathrm{proc}}:
\theta_k(\tau)>0\text{ for at least one }\tau\in\calT.
\end{equation}

The corresponding observed statistics are
$V_{k,\mathrm{CAR}}^{\mathrm{fix}}$ and
$V_{k,\mathrm{CAR}}^{\mathrm{proc}}$ from
Definition~\ref{def:car-adjusted-statistics}. For notational economy in the
multiple-testing algorithm, write
$V_k^{\mathrm{fix}}=V_{k,\mathrm{CAR}}^{\mathrm{fix}}$ and
$V_k^{\mathrm{proc}}=V_{k,\mathrm{CAR}}^{\mathrm{proc}}$.
Their centered bootstrap analogues are
\begin{align}
\label{eq:fixed-bootstrap-statistic}
V_k^{*,\mathrm{fix}}
&=
\frac{\mathbb G_{\theta,k}^*(\tau^\star)}
{\widehat\sigma_{k,\mathrm{CAR}}(\tau^\star)},\\
\label{eq:process-bootstrap-statistic}
V_k^{*,\mathrm{proc}}
&=
\sup_{\tau\in\calT}
\frac{\mathbb G_{\theta,k}^*(\tau)}
{\widehat\sigma_{k,\mathrm{CAR}}(\tau)}.
\end{align}
For either analysis type $h\in\{\mathrm{fix},\mathrm{proc}\}$ and any
nonempty subset $I\subseteq\{1,\ldots,K\}$, let
\begin{equation}
\label{eq:subset-critical-value}
c_{1-\alpha}^*(I,h)
=
\inf\left[
x:
\Pp^*\left\{
\max_{k\in I}V_k^{*,h}\leq x
\right\}
\geq1-\alpha
\right].
\end{equation}

\paragraph{Single-step and step-down procedures.}

The single-step procedure rejects $H_{0k}^h$ whenever
\begin{equation}
\label{eq:single-step-rule}
V_k^h>c_{1-\alpha}^*(\{1,\ldots,K\},h).
\end{equation}

For greater power, we also use the following step-down max-$T$ procedure.
Set $I_1=\{1,\ldots,K\}$. At step $j$:
\begin{enumerate}[label=(\roman*),leftmargin=12mm]
\item calculate $c_{1-\alpha}^*(I_j,h)$,
\item reject every $H_{0k}^h$ with
      $k\in I_j$ and $V_k^h>c_{1-\alpha}^*(I_j,h)$,
\item stop if no hypothesis is rejected, and otherwise remove the rejected
      indices to form $I_{j+1}$ and continue.
\end{enumerate}
This is a resampling-based step-down construction in the spirit of
\citet{romano2005stepdown}.

The next result applies the step-down principle of
\citet{romano2005stepdown}. Their argument requires the critical value for
every subset of true null hypotheses to be approximated consistently under
the true data-generating distribution. Because
Theorems~\ref{thm:bounded-balance-fclt}--\ref{thm:multiplier-validity} hold
under arbitrary treatment effects, this requirement is met.

\begin{corollary}[Asymptotic strong family-wise error control]
\label{thm:strong-fwer}
Suppose the conditions of Theorem~\ref{thm:multiplier-validity} hold and the
limiting subset-max distributions are continuous at their
$(1-\alpha)$ quantiles. For either
$h=\mathrm{fix}$ or $h=\mathrm{proc}$, let
\[
I_0^h
=
\{k:H_{0k}^h\text{ is true}\}.
\]
Then the single-step rule \eqref{eq:single-step-rule} and the step-down
procedure both satisfy
\begin{equation}
\label{eq:strong-fwer}
\limsup_{N\to\infty}
\Pp\{\text{at least one }H_{0k}^h,\ k\in I_0^h,
\text{ is rejected}\}
\leq\alpha
\end{equation}
under every configuration of true and false null hypotheses.
\end{corollary}

For two-sided hypotheses, the statistics and their bootstrap analogues are
replaced by absolute values, and the proof of Corollary~\ref{thm:strong-fwer}
is unchanged.

\begin{corollary}[Simultaneous inference with conditional censoring]
\label{cor:standardized-simultaneous-inference}
If the variances in \eqref{eq:standardized-se} are uniformly bounded away
from zero and the relevant
limiting maximum distributions are continuous at their target quantiles,
then replacing
$\widehat\btheta$, $\widehat\sigma_{k,\mathrm{CAR}}$, and
$\mathbb G_\theta^*$ in
Subsections~\ref{sec:multiplier}--\ref{sec:fwer} by
$\widehat\btheta^{\mathrm{st}}$,
$\widehat\sigma_{k,\mathrm{st}}$, and
$\mathbb G_\theta^{\mathrm{st},*}$ gives:
\begin{enumerate}[label=(\roman*)]
\item simultaneous confidence bands with joint asymptotic coverage
      $1-\alpha$, and
\item single-step and step-down max-$T$ tests with asymptotic strong
      family-wise error control at level $\alpha$.
\end{enumerate}
Both conclusions hold under every configuration of true and false null
hypotheses.
\end{corollary}

\subsection{Choice Between the Two Procedures}
\label{sec:recommendation}

\begin{remark}[Recommended analysis]
\label{rmk:recommended}
Theorem~\ref{thm:standardized-fclt} motivates the standardized estimator as
the default primary analysis for CAR trials with a moderate number of
prespecified strata. It is valid under a weaker censoring condition than
the marginal Kaplan--Meier analysis. Its limiting distribution is the same
under stratified blocks and minimization, and it requires neither the
design simulation of Subsection~\ref{sec:covariance-estimation} nor the
additivity check of Proposition~\ref{prop:additivity}. Its efficiency is
comparable to that of the marginal analysis under balancing designs and
better under simple randomization, where post-stratification removes the
between-stratum variance that the unstratified estimator inherits from
chance imbalance, as Section~\ref{sec:simulation} shows. The CAR-adjusted marginal procedure remains useful when an
unadjusted analysis has been prespecified, for example to match a
historical reporting convention, and as a sensitivity analysis whose
validity under minimization is governed by
Proposition~\ref{prop:additivity}.
\end{remark}

\begin{remark}[Sparse treatment--stratum cells]
\label{rmk:sparse-standardization}
The direct standardized estimator is most appropriate when the number of
randomization strata is fixed and every treatment--stratum cell contains
enough subjects and subjects at risk through $\tau_U$. If cells are very
sparse, $\tau_U$ should be shortened, clinically defensible strata may be
combined according to a prespecified rule, or an IPCW or
regression-augmented RMST estimator with a lower-dimensional censoring or
outcome model should be used \citep{tian2014predicting}. Empty cells cannot
be repaired by the asymptotic theory and should not be handled by
post-randomization deletion of inconvenient strata.
\end{remark}

\section{Simulation Studies}
\label{sec:simulation}

This section examines the finite-sample performance of the proposed
procedures. We first check their error control under four randomization
procedures. We then show the two ways in which simpler analyses go wrong:
the unstratified estimator targets the wrong quantity when censoring
depends on the strata, and omitting the minimization design term inflates
error when the factors interact. Finally, we compare the power of the standardized procedure with that of
the stratified log-rank test and the stratified Cox model.

\subsection{Design}

\paragraph{Trials and event times.}
Each trial has one control and $K=3$ experimental arms with equal
allocation $\pi_a=1/4$, and the total sample size is $N=800$ or $N=1200$.
Two independent binary factors, $Z_{i1}\sim\operatorname{Bernoulli}(0.5)$ and
$Z_{i2}\sim\operatorname{Bernoulli}(0.4)$, define four joint randomization
strata. In the control arm the hazard is
\begin{equation}
\label{eq:simulation-control-hazard}
\lambda_0(t\mid\bm Z_i)
=
0.04\exp(0.45Z_{i1}+0.65Z_{i2}),
\end{equation}
which gives a median survival between 5.8 and 17.3 months across strata.
Arm $a$ has hazard $\lambda_0(t\mid\bm Z_i)r_a(t)$, where $r_0\equiv1$ and
the multipliers of the experimental arms are listed in
Table~\ref{tab:simulation-scenarios}. Model
\eqref{eq:simulation-control-hazard} is additive in the two factors on the
log-hazard scale, which makes it the favorable case of
Proposition~\ref{prop:additivity}, although the stratum means
$m_a(\cdot,\tau)$ are then only nearly additive. To
examine the unfavorable case, scenario G0-I uses instead
\begin{equation}
\label{eq:simulation-interaction-hazard}
\lambda_0(t\mid\bm Z_i)
=
0.04\exp(1.2Z_{i1}+1.2Z_{i2}-2.4Z_{i1}Z_{i2}),
\end{equation}
under which the hazard is 0.04 in the two concordant strata and about 0.133
in the two discordant strata. Each factor is then prognostic only in the
absence of the other, a qualitative interaction that balancing the margins
cannot remove. To trace the transition between the two models, a further
null experiment uses the log-hazard
$(1-c)(0.45Z_{i1}+0.65Z_{i2})+c(1.2Z_{i1}+1.2Z_{i2}-2.4Z_{i1}Z_{i2})$ with
interaction weight $c\in\{0,0.25,0.5,0.75,1\}$, so that $c=0$ gives the
additive model and $c=1$ gives G0-I. Proposition~\ref{prop:additivity}
concerns the stratum means $m_a(s,\tau)$ rather than the log-hazard, so we
also computed the share of $m_a(\cdot,12)$ that is not additive, measured
by the $\bm p$-weighted root mean square of its non-additive remainder
relative to that of $m_a(\cdot,12)$. This share rises with $c$ from 0.09 to
0.41, 0.81, 0.95, and 0.98. Interpolating keeps every stratum hazard
moderate, whereas raising the interaction coefficient alone with the main
effects held at 1.2 would leave the stratum with both factors present
almost without patients at risk at 12 months.

Scenarios WAN and EBL keep arms 2 and 3 of P-PH and change only how the
benefit of arm 1 unfolds over time. In WAN the benefit is confined to the
first six months, after which the hazards are equal, a pattern seen when
treatment is given for a fixed period or when resistance develops. In EBL
the early benefit is followed by a modest excess hazard, as can happen
when an intensive therapy carries delayed toxicity. The early multiplier
of arm 1 is calibrated in each case so that its 12-month RMST gain equals
that under P-PH, namely 1.075 months. The three scenarios therefore share
the confirmatory effect size and differ only in its shape, which makes
their power directly comparable.

\begin{table}[t!]
\centering
\small
\caption{Event-time and censoring scenarios. The third column gives the
hazard multipliers $(r_1,r_2,r_3)$ of the three experimental arms. The
arm-1 multipliers of WAN and EBL are rounded here, and their exact values
are given in Section~\ref{sec:supp-sim-design}.}
\label{tab:simulation-scenarios}
\begin{tabular}{p{1.4cm}p{3.4cm}p{4.7cm}p{3.4cm}}
\toprule
Scenario & Treatment configuration & Hazard multipliers & Censoring model\\
\midrule
G0-M &
Global null &
$(1,1,1)$ &
Marginally independent\\
G0-C &
Global null &
$(1,1,1)$ &
Stratum dependent\\
P-PH &
Partial null under PH &
$(0.65,1,0.80)$ &
Stratum dependent\\
A-PH &
All treatments active under PH &
$(0.65,0.75,0.85)$ &
Stratum dependent\\
DEL &
Delayed effect in arm 1; arm 2 null &
$r_1(t)=1$ for $t\leq6$ and $0.50$ afterward;
$r_2(t)=1$; $r_3(t)=0.80$ &
Stratum dependent\\
CRS &
Crossing hazards in arm 1; arm 2 null &
$r_1(t)=1.25$ for $t\leq4$ and $0.50$ afterward;
$r_2(t)=1$; $r_3(t)=0.80$ &
Stratum dependent\\
WAN &
Waning benefit in arm 1; arm 2 null &
$r_1(t)=0.55$ for $t\leq6$ and $1$ afterward;
$r_2(t)=1$; $r_3(t)=0.80$ &
Stratum dependent\\
EBL &
Early benefit, late harm in arm 1; arm 2 null &
$r_1(t)=0.50$ for $t\leq6$ and $1.20$ afterward;
$r_2(t)=1$; $r_3(t)=0.80$ &
Stratum dependent\\
G0-I &
Global null; interacting prognostic factors &
$(1,1,1)$; control hazard \eqref{eq:simulation-interaction-hazard} &
Marginally independent\\
\bottomrule
\end{tabular}
\end{table}

\paragraph{Censoring and restriction times.}
All subjects are administratively censored at 24 months and may also be
lost to follow-up at an exponential rate. In the marginally independent
setting this rate is the same for every subject. In the stratum-dependent
setting it equals $\lambda_C\exp(\gamma_{a1}Z_{i1}+\gamma_{a2}Z_{i2})$, so that
loss to follow-up depends on the randomization factors in an arm-specific
way. Censoring is then independent of the event time within each stratum
but not marginally, which is precisely the situation in which the
unstratified Kaplan--Meier estimator is biased. We let the coefficients
differ across arms because otherwise this bias would largely cancel in the
treatment--control contrasts. In both settings $\lambda_C$ is chosen to give
about 30\% censoring under the global null, and Section~\ref{sec:supp-sim-design} lists the
coefficients. The confirmatory restriction time is $\tau^\star=12$ months,
and the process-wise analysis covers $\calT=[4,12]$ on a quarter-month grid.

\paragraph{Randomization procedures.}
We compare four procedures. Simple randomization (SR) serves as the
benchmark, and stratified permuted blocks (SPB) use blocks of size eight
within the joint strata. Pocock--Simon minimization (PS) balances only the
two marginal factors, assigning probability 0.80 to the arms that minimize
the marginal imbalance. The multi-arm Hu--Hu procedure (HH) of
\citet{hu2012asymptotic} and \citet{hu2023multi} combines overall,
joint-stratum, and marginal imbalances with weights
$(0.30,0.50,0.10,0.10)$ and gives probability 0.80 to the best-ranked arm.
Because SPB and HH both control the joint-stratum imbalances, their design
term is zero. Under PS the design term is estimated by simulating the same
minimization rule. Section~\ref{sec:supp-sim-design} gives both allocation rules in full.

\paragraph{Methods.}
We compare six analyses.
\begin{enumerate}[label=\textbf{M\arabic*:},leftmargin=13mm]
\item \textit{Naive marginal KM}: unstratified Kaplan--Meier RMST with the
      usual Greenwood-type covariance, which accounts for the shared control
      but not for CAR.
\item \textit{CAR-adjusted marginal RMST}: the procedure of
      Definition~\ref{def:car-adjusted-statistics}.
\item \textit{Standardized RMST}: the recommended procedure of
      Definition~\ref{def:standardized-statistics}.
\item \textit{M3 with Holm}: the fixed-time tests of M3 with Holm's
      adjustment in place of max-$T$.
\item \textit{Bounded-balance kernel}: M2 with $\widehat{\bm\Omega}_D=\bm 0$,
      used only in G0-I and the graded-interaction experiment.
\item \textit{Stratified log-rank}: pairwise treatment--control log-rank
      tests stratified by the four joint strata, with shared-control
      max-$T$ calibration and follow-up censored at $\tau^\star$.
\end{enumerate}
M2 and M3 are the proposed procedures, M1 represents current practice, and
M4 and M5 are diagnostic variants. M6 is the standard time-to-event
benchmark. Because it tests equality of survival curves rather than an RMST
contrast, its comparison with M2 and M3 is a comparison of testing
strategies, and Section~\ref{sec:supp-logrank} gives its score, variance, and correlation
estimator. In the proportional-hazards scenarios we also fit a Cox model
stratified by the joint strata, with Wald max-$T$ tests and follow-up
censored at $\tau^\star$.

\paragraph{Equal-RMST experiments.}
Apart from P-PH, WAN, and EBL, the scenarios in
Table~\ref{tab:simulation-scenarios} differ in their RMST gains, so they
cannot separate the effect of the hazard shape from the size of the effect.
Two further experiments therefore fix the gain over a wider range. Only arm 1
is active, and its 12-month RMST gain is 0.50, 0.75, 1.00, or 1.25 months.
For each gain we calibrate four hazard profiles to the same gain:
proportional hazards, a benefit that starts after month 3, a benefit that
lasts only through month 6, and an immunotherapy-type profile with a mild
early disadvantage ($r_1=1.10$ through month 3) followed by a sustained
benefit. The two experiments differ only in the censoring mechanism, which
is marginally independent in one and stratum dependent in the other. A
reverse-crossing profile with zero expected log-rank score serves only as a
stress test and is reported in Section~\ref{sec:supp-equal-rmst}, together with the calibrated
multipliers.

\paragraph{Replications and reporting.}
G0-M and G0-C use 20,000 trials and 10,000 multiplier draws per trial,
whereas G0-I, the graded-interaction experiment, the treatment-effect
scenarios, and both equal-RMST experiments use 10,000 trials and 5,000
draws. Within each configuration
all methods are applied to the same trials. The Pocock--Simon design
covariance is based on 100,000 design replicates for the global-null
calibration and on 20,000 elsewhere. All tests use one-sided family-wise
level $\alpha=0.025$, and we report the family-wise error under the global
null, the strong FWER for the true-null arms, and disjunctive power. Trials
in which some treatment--stratum risk set is empty at $\tau_U$ are flagged.
In the global-null scenarios, the graded-interaction experiment, and the
treatment-effect scenarios, error rates and power are computed among the
remaining trials, with optimistic and worst-case bounds kept in the
reproducible output. In the equal-RMST experiments, flagged M3 trials count
as nonrejections.

\subsection{Type I error control}

Table~\ref{tab:simulation-error-by-method} reports the Type I error rates.
M1 and M2 are omitted under G0-C and the treatment-effect scenarios
because, with stratum-dependent censoring, they do not identify the target
RMST.

\begin{table}[tp]
\centering
\scriptsize
\caption{Family-wise Type I error under the global-null scenarios and strong family-wise error under partial-null scenarios. Each cell reports the fixed-time/process-wise error rate.}
\label{tab:simulation-error-by-method}
\setlength{\tabcolsep}{2.8pt}
\begin{tabular}{llcccccccc}
\toprule
&& \multicolumn{4}{c}{$N=800$} & \multicolumn{4}{c}{$N=1200$}\\
\cmidrule(lr){3-6}\cmidrule(lr){7-10}
Scenario & Method & SR & SPB & PS & HH & SR & SPB & PS & HH\\
\midrule
G0-M & M1 & 0.026/0.027 & 0.022/0.024 & 0.019/0.024 & 0.021/0.024 & 0.026/0.027 & 0.020/0.023 & 0.023/0.024 & 0.022/0.024\\
 & M2 & 0.025/0.027 & 0.027/0.028 & 0.024/0.026 & 0.025/0.027 & 0.026/0.025 & 0.026/0.027 & 0.026/0.026 & 0.026/0.029\\
 & M3 & 0.025/0.028 & 0.027/0.028 & 0.023/0.026 & 0.025/0.026 & 0.025/0.026 & 0.026/0.027 & 0.026/0.026 & 0.027/0.029\\
 & M4 & 0.022/-- & 0.024/-- & 0.021/-- & 0.023/-- & 0.022/-- & 0.022/-- & 0.023/-- & 0.024/--\\
 & M6 & 0.025/-- & 0.027/-- & 0.025/-- & 0.025/-- & 0.024/-- & 0.027/-- & 0.027/-- & 0.027/--\\
\addlinespace
G0-C & M3 & 0.026/0.027 & 0.024/0.025 & 0.026/0.027 & 0.026/0.027 & 0.025/0.025 & 0.028/0.028 & 0.024/0.025 & 0.025/0.027\\
 & M4 & 0.023/-- & 0.021/-- & 0.023/-- & 0.024/-- & 0.022/-- & 0.025/-- & 0.021/-- & 0.022/--\\
 & M6 & 0.025/-- & 0.023/-- & 0.027/-- & 0.027/-- & 0.026/-- & 0.028/-- & 0.023/-- & 0.026/--\\
\addlinespace
P-PH & M3 & 0.022/0.021 & 0.022/0.020 & 0.022/0.021 & 0.020/0.020 & 0.023/0.020 & 0.021/0.020 & 0.022/0.022 & 0.021/0.021\\
 & M4 & 0.021/-- & 0.021/-- & 0.021/-- & 0.020/-- & 0.022/-- & 0.021/-- & 0.022/-- & 0.020/--\\
 & M6 & 0.022/-- & 0.022/-- & 0.023/-- & 0.021/-- & 0.022/-- & 0.023/-- & 0.021/-- & 0.022/--\\
\addlinespace
DEL & M3 & 0.014/0.013 & 0.016/0.014 & 0.017/0.015 & 0.015/0.012 & 0.015/0.014 & 0.018/0.015 & 0.018/0.016 & 0.016/0.015\\
 & M4 & 0.013/-- & 0.015/-- & 0.016/-- & 0.013/-- & 0.014/-- & 0.017/-- & 0.017/-- & 0.015/--\\
 & M6 & 0.017/-- & 0.019/-- & 0.019/-- & 0.016/-- & 0.021/-- & 0.021/-- & 0.022/-- & 0.021/--\\
\addlinespace
CRS & M3 & 0.015/0.013 & 0.014/0.013 & 0.014/0.012 & 0.011/0.012 & 0.014/0.014 & 0.011/0.012 & 0.014/0.015 & 0.013/0.013\\
 & M4 & 0.013/-- & 0.013/-- & 0.013/-- & 0.011/-- & 0.013/-- & 0.011/-- & 0.013/-- & 0.012/--\\
 & M6 & 0.017/-- & 0.016/-- & 0.017/-- & 0.014/-- & 0.020/-- & 0.018/-- & 0.017/-- & 0.018/--\\
\addlinespace
WAN & M3 & 0.020/0.020 & 0.019/0.020 & 0.020/0.018 & 0.018/0.018 & 0.021/0.021 & 0.025/0.023 & 0.022/0.021 & 0.023/0.023\\
 & M4 & 0.019/-- & 0.019/-- & 0.020/-- & 0.018/-- & 0.021/-- & 0.024/-- & 0.021/-- & 0.022/--\\
 & M6 & 0.020/-- & 0.023/-- & 0.021/-- & 0.019/-- & 0.024/-- & 0.024/-- & 0.021/-- & 0.024/--\\
\addlinespace
EBL & M3 & 0.021/0.020 & 0.019/0.017 & 0.021/0.020 & 0.020/0.019 & 0.021/0.022 & 0.022/0.023 & 0.022/0.020 & 0.022/0.024\\
 & M4 & 0.020/-- & 0.018/-- & 0.021/-- & 0.020/-- & 0.021/-- & 0.022/-- & 0.021/-- & 0.022/--\\
 & M6 & 0.019/-- & 0.017/-- & 0.021/-- & 0.021/-- & 0.021/-- & 0.020/-- & 0.023/-- & 0.021/--\\
\addlinespace
G0-I & M1 & 0.028/0.025 & 0.016/0.021 & 0.024/0.025 & 0.015/0.019 & 0.025/0.027 & 0.014/0.018 & 0.028/0.030 & 0.017/0.019\\
 & M2 & 0.026/0.024 & 0.026/0.029 & 0.023/0.023 & 0.026/0.028 & 0.026/0.026 & 0.026/0.026 & 0.027/0.030 & 0.027/0.027\\
 & M3 & 0.025/0.027 & 0.027/0.029 & 0.026/0.026 & 0.025/0.026 & 0.025/0.026 & 0.025/0.025 & 0.026/0.030 & 0.027/0.027\\
 & M4 & 0.023/-- & 0.025/-- & 0.023/-- & 0.022/-- & 0.021/-- & 0.022/-- & 0.022/-- & 0.024/--\\
 & M5 & 0.045/0.038 & 0.026/0.030 & 0.039/0.035 & 0.026/0.028 & 0.039/0.037 & 0.026/0.026 & 0.045/0.042 & 0.027/0.026\\
 & M6 & 0.028/-- & 0.027/-- & 0.026/-- & 0.026/-- & 0.024/-- & 0.025/-- & 0.026/-- & 0.027/--\\
\addlinespace
\bottomrule
\end{tabular}
\begin{minipage}{0.98\textwidth}
\footnotesize
A dash indicates that no process-wise procedure was considered for M4 or M6, and A-PH is omitted because it has no true null. M1 and M2 are not reported under G0-C, P-PH, DEL, CRS, WAN, or EBL because they do not identify the target marginal RMST under stratum-dependent censoring. G0-M and G0-C use 20,000 trials and 10,000 Gaussian draws per trial, and all other rows use 10,000 trials and 5,000 draws. The PS design covariance uses 100,000 design-only replicates for G0-M and G0-C and 20,000 otherwise. M6 tests equality of survival curves through month 12, and its true-null hypotheses coincide with the RMST nulls in the reported error-rate comparisons.
\end{minipage}
\end{table}

Both proposed procedures held the nominal level throughout. Their
family-wise error rates stayed close to 0.025 under all four designs, at
both sample sizes, and in both the fixed-time and the process-wise
analyses, and they showed no systematic dependence on the design, as
Theorems~\ref{thm:general-car-fclt} and~\ref{thm:standardized-fclt} predict.
Under the partial nulls the strong family-wise error of M3 stayed at or
below the nominal level. It was lowest under DEL and CRS, whose low power
means that the step-down procedure rarely removes the false nulls, so the
true null is usually tested against the critical value of the full family.
The simultaneous 95\% bands of both procedures covered the true contrast
processes in close to 95\% of trials in every scenario
(Section~\ref{sec:supp-simulation}).

The conventional analysis M1 behaved as the limit theory predicts. It was
close to nominal under simple randomization, where the Greenwood
covariance is correct, and conservative under the balancing designs, with
the exception of minimization under interacting factors discussed below.
It therefore never inflated the error rate under marginally independent
censoring, but it overstated the standard error and gave away part of the
precision gained by balancing. The overstatement was larger when the
factors were strongly prognostic, reaching about 9\% under SPB and HH in
G0-I, which at a nominal power of 80\% costs about 6 percentage points. M4
was slightly more conservative than M3 throughout, because Holm's
adjustment ignores the dependence induced by the shared control.

The G0-I rows isolate the minimization design term, which
Propositions~\ref{prop:additivity} and~\ref{prop:nondegenerate} tie to the
interaction between the factors. Under PS, the estimated term barely
changed the contrast variances in the additive model
\eqref{eq:simulation-control-hazard} but added 7\% to 20\% to them in the
interaction model \eqref{eq:simulation-interaction-hazard}, more at later
restriction times, and omitting it mattered. M5 then had an error rate of
about 0.04 at both sample sizes, whereas M2 restored the nominal level.
Under SPB and HH the two coincide because $\bm\Omega_D=\bm 0$. M5 was also
inflated under SR, but for a different reason. Simple randomization does
not balance the margins, so the bounded-balance kernel misses the whole
imbalance variance whatever the prognostic structure. M3 was valid under
every design without any design simulation, and it was also more efficient
than the marginal estimator under SR and PS.

\begin{figure}[t!]
\centering
\includegraphics[width=\textwidth]{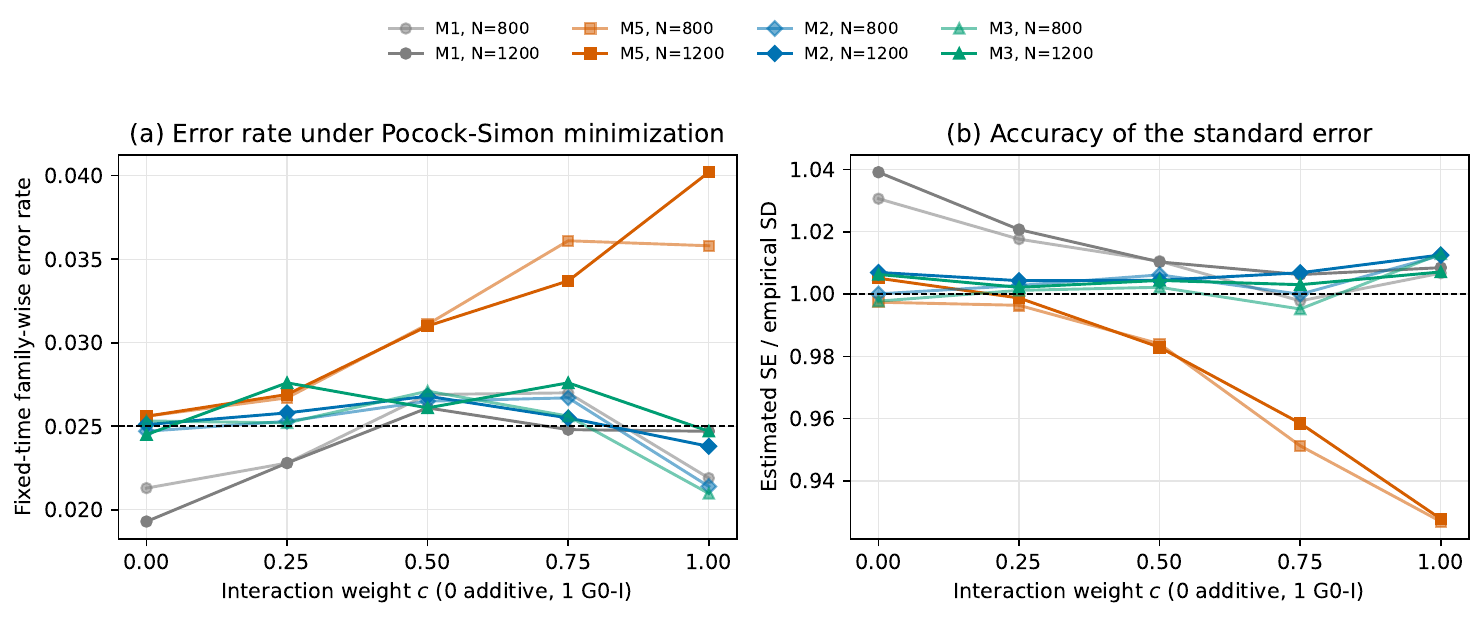}
\caption{Graded-interaction null experiment under Pocock--Simon
minimization. The interaction weight $c$ interpolates between the additive
model ($c=0$) and G0-I ($c=1$). Panel (a) gives the fixed-time family-wise
error rate, with the dashed line at the nominal level 0.025. Panel (b)
gives the ratio of the average estimated standard error of the 12-month
contrasts to their empirical standard deviation. Each point uses 10,000
trials, so the Monte Carlo standard error of each error rate is about
0.0016. The points at $c=1$ are an independent replication of G0-I.}
\label{fig:interaction-gradient}
\end{figure}

Figure~\ref{fig:interaction-gradient} shows that the error caused by
omitting the design term grows smoothly with the strength of the
interaction. As $c$ increases, the standard error used by M5 falls further
below the empirical standard deviation, and its error rate rises steadily
above the nominal level. The curves for the two sample sizes nearly
coincide, so the distortion is of first order and does not vanish as $N$
grows. M2 and M3 stay at the nominal level, with accurate standard errors,
over the whole range. M1 moves in the opposite direction, from an inflated
standard error in the additive model to a nearly exact one at $c=1$,
because minimization removes less of the prognostic variation as the
interaction grows. Under SR, M5 is inflated even at $c=0$, which again
shows that omitting the design term is safe only under designs that
balance the margins (Section~\ref{sec:supp-sim-gradient}).

Stratum-dependent censoring affects the unstratified estimator through its
target rather than its variance. Under G0-C, the unstratified contrasts
had a bias that did not shrink with $N$, a failure of identification
rather than a small-sample effect, and the error rates of M1 and M2 were
above the nominal level under every design. M2 was affected more than M1,
because its standard error is correct, whereas the conservative standard
error of M1 partly masks the bias. In an ICON6-based setting the inflation
even grows with the sample size (Section~\ref{sec:supp-icon6-censoring}).
Standardization removed essentially all of the bias, with accurate standard
errors and nominal coverage. Its positivity requirement was rarely binding,
since fewer than 1\% of the global-null trials had an empty
treatment--stratum risk set at $\tau_U$, almost all of them under G0-C
(Section~\ref{sec:supp-sim-estimation}).

The stratified log-rank benchmark M6 was equally well calibrated under
every design. Its stability has the same source as that of M3. Both
analyses compare arms within each joint stratum, so a within-stratum
imbalance changes only how much information a stratum contributes and adds
no systematic term to the statistic.

\subsection{Power}

Table~\ref{tab:simulation-power-by-method} reports disjunctive power. The M1
and M2 entries are marked because, in these scenarios, the two methods
target censoring-distorted quantities, so their apparent power is not a
valid efficiency comparison.

\begin{table}[tp]
\centering
\scriptsize
\caption{Disjunctive power across treatment-effect scenarios. Each cell reports fixed-time/process-wise power.}
\label{tab:simulation-power-by-method}
\setlength{\tabcolsep}{2.8pt}
\begin{tabular}{llcccccccc}
\toprule
&& \multicolumn{4}{c}{$N=800$} & \multicolumn{4}{c}{$N=1200$}\\
\cmidrule(lr){3-6}\cmidrule(lr){7-10}
Scenario & Method & SR & SPB & PS & HH & SR & SPB & PS & HH\\
\midrule
P-PH & M1$^\dagger$ & 0.681/0.599 & 0.691/0.613 & 0.692/0.608 & 0.681/0.600 & 0.859/0.799 & 0.859/0.799 & 0.864/0.809 & 0.854/0.798\\
 & M2$^\dagger$ & 0.675/0.593 & 0.712/0.636 & 0.712/0.629 & 0.706/0.625 & 0.855/0.796 & 0.877/0.819 & 0.875/0.825 & 0.870/0.816\\
 & M3 & 0.633/0.552 & 0.645/0.570 & 0.643/0.561 & 0.634/0.552 & 0.820/0.757 & 0.818/0.755 & 0.824/0.761 & 0.814/0.752\\
 & M4 & 0.617/-- & 0.630/-- & 0.625/-- & 0.619/-- & 0.808/-- & 0.808/-- & 0.814/-- & 0.802/--\\
 & M6 & 0.703/-- & 0.714/-- & 0.712/-- & 0.699/-- & 0.879/-- & 0.876/-- & 0.876/-- & 0.874/--\\
\addlinespace
A-PH & M1$^\dagger$ & 0.709/0.626 & 0.709/0.629 & 0.708/0.627 & 0.713/0.637 & 0.876/0.825 & 0.886/0.832 & 0.883/0.827 & 0.876/0.821\\
 & M2$^\dagger$ & 0.700/0.619 & 0.733/0.650 & 0.730/0.651 & 0.736/0.661 & 0.875/0.822 & 0.900/0.850 & 0.895/0.844 & 0.889/0.839\\
 & M3 & 0.664/0.586 & 0.671/0.592 & 0.670/0.591 & 0.675/0.604 & 0.848/0.788 & 0.855/0.793 & 0.847/0.786 & 0.845/0.784\\
 & M4 & 0.645/-- & 0.654/-- & 0.650/-- & 0.659/-- & 0.836/-- & 0.842/-- & 0.834/-- & 0.834/--\\
 & M6 & 0.735/-- & 0.744/-- & 0.739/-- & 0.742/-- & 0.900/-- & 0.904/-- & 0.899/-- & 0.897/--\\
\addlinespace
DEL & M1$^\dagger$ & 0.239/0.186 & 0.228/0.180 & 0.226/0.174 & 0.225/0.171 & 0.342/0.271 & 0.341/0.274 & 0.345/0.270 & 0.336/0.266\\
 & M2$^\dagger$ & 0.234/0.182 & 0.252/0.197 & 0.248/0.191 & 0.246/0.189 & 0.338/0.267 & 0.369/0.295 & 0.370/0.293 & 0.363/0.290\\
 & M3 & 0.209/0.166 & 0.207/0.169 & 0.204/0.160 & 0.204/0.158 & 0.299/0.240 & 0.309/0.253 & 0.311/0.247 & 0.304/0.244\\
 & M4 & 0.197/-- & 0.197/-- & 0.194/-- & 0.191/-- & 0.283/-- & 0.296/-- & 0.297/-- & 0.288/--\\
 & M6 & 0.312/-- & 0.309/-- & 0.309/-- & 0.307/-- & 0.442/-- & 0.451/-- & 0.456/-- & 0.452/--\\
\addlinespace
CRS & M1$^\dagger$ & 0.221/0.176 & 0.211/0.160 & 0.209/0.164 & 0.204/0.162 & 0.316/0.252 & 0.309/0.241 & 0.311/0.249 & 0.303/0.242\\
 & M2$^\dagger$ & 0.217/0.169 & 0.232/0.179 & 0.231/0.181 & 0.227/0.179 & 0.311/0.248 & 0.337/0.266 & 0.334/0.270 & 0.330/0.260\\
 & M3 & 0.197/0.154 & 0.194/0.152 & 0.193/0.154 & 0.192/0.154 & 0.280/0.223 & 0.283/0.223 & 0.286/0.231 & 0.278/0.223\\
 & M4 & 0.183/-- & 0.181/-- & 0.181/-- & 0.179/-- & 0.263/-- & 0.268/-- & 0.272/-- & 0.263/--\\
 & M6 & 0.283/-- & 0.276/-- & 0.282/-- & 0.276/-- & 0.403/-- & 0.401/-- & 0.407/-- & 0.402/--\\
\addlinespace
WAN & M1$^\dagger$ & 0.709/0.700 & 0.710/0.696 & 0.704/0.696 & 0.716/0.710 & 0.878/0.870 & 0.882/0.881 & 0.882/0.875 & 0.881/0.879\\
 & M2$^\dagger$ & 0.703/0.694 & 0.731/0.715 & 0.726/0.710 & 0.740/0.725 & 0.878/0.869 & 0.897/0.890 & 0.894/0.885 & 0.894/0.888\\
 & M3 & 0.653/0.665 & 0.658/0.668 & 0.651/0.661 & 0.660/0.676 & 0.841/0.851 & 0.840/0.853 & 0.840/0.851 & 0.841/0.853\\
 & M4 & 0.636/-- & 0.643/-- & 0.631/-- & 0.644/-- & 0.829/-- & 0.828/-- & 0.827/-- & 0.830/--\\
 & M6 & 0.503/-- & 0.512/-- & 0.494/-- & 0.502/-- & 0.693/-- & 0.686/-- & 0.687/-- & 0.695/--\\
\addlinespace
EBL & M1$^\dagger$ & 0.721/0.787 & 0.728/0.789 & 0.721/0.785 & 0.733/0.794 & 0.889/0.930 & 0.894/0.938 & 0.894/0.939 & 0.899/0.940\\
 & M2$^\dagger$ & 0.714/0.782 & 0.750/0.804 & 0.741/0.797 & 0.755/0.809 & 0.887/0.928 & 0.908/0.944 & 0.905/0.945 & 0.912/0.945\\
 & M3 & 0.660/0.762 & 0.675/0.769 & 0.667/0.764 & 0.678/0.770 & 0.851/0.919 & 0.851/0.924 & 0.851/0.925 & 0.856/0.926\\
 & M4 & 0.640/-- & 0.657/-- & 0.649/-- & 0.660/-- & 0.840/-- & 0.841/-- & 0.839/-- & 0.846/--\\
 & M6 & 0.402/-- & 0.403/-- & 0.404/-- & 0.413/-- & 0.573/-- & 0.582/-- & 0.573/-- & 0.577/--\\
\addlinespace
\bottomrule
\end{tabular}
\begin{minipage}{0.98\textwidth}
\footnotesize
A dash appears in the process-wise position for M4 and M6. Conjunctive power and simultaneous-band coverage are reported in Table~\ref{tab:supp-alternative-ranges}. $^\dagger$Because censoring is independent only within strata, the M1 and M2 entries are apparent rejection probabilities for censoring-distorted KM quantities, not power for the target marginal RMST. M6 tests equality of survival curves through month 12, so its power comparison with RMST procedures is at the procedure level. Every configuration uses 10,000 trials, 5,000 Gaussian draws, and, for PS, 20,000 design-only replicates.
\end{minipage}
\end{table}

The power of the recommended procedure M3 was nearly the same under the
four designs, in agreement with Theorem~\ref{thm:standardized-fclt}, under
which the standardized contrasts have the same limit under every design.
Searching over $[4,12]$ cost some power when $\tau^\star=12$ already
captured the effect well, as under the proportional-hazards, delayed, and
crossing scenarios, but it paid off under EBL. There the late excess hazard
erodes the RMST gain between months 6 and 12, so the contrast is larger
relative to its standard error at earlier restriction times. Exploiting
the shared-control dependence through max-$T$ also gave a small but
consistent gain over Holm's adjustment.

\begin{figure}[t!]
\centering
\includegraphics[width=\textwidth]
{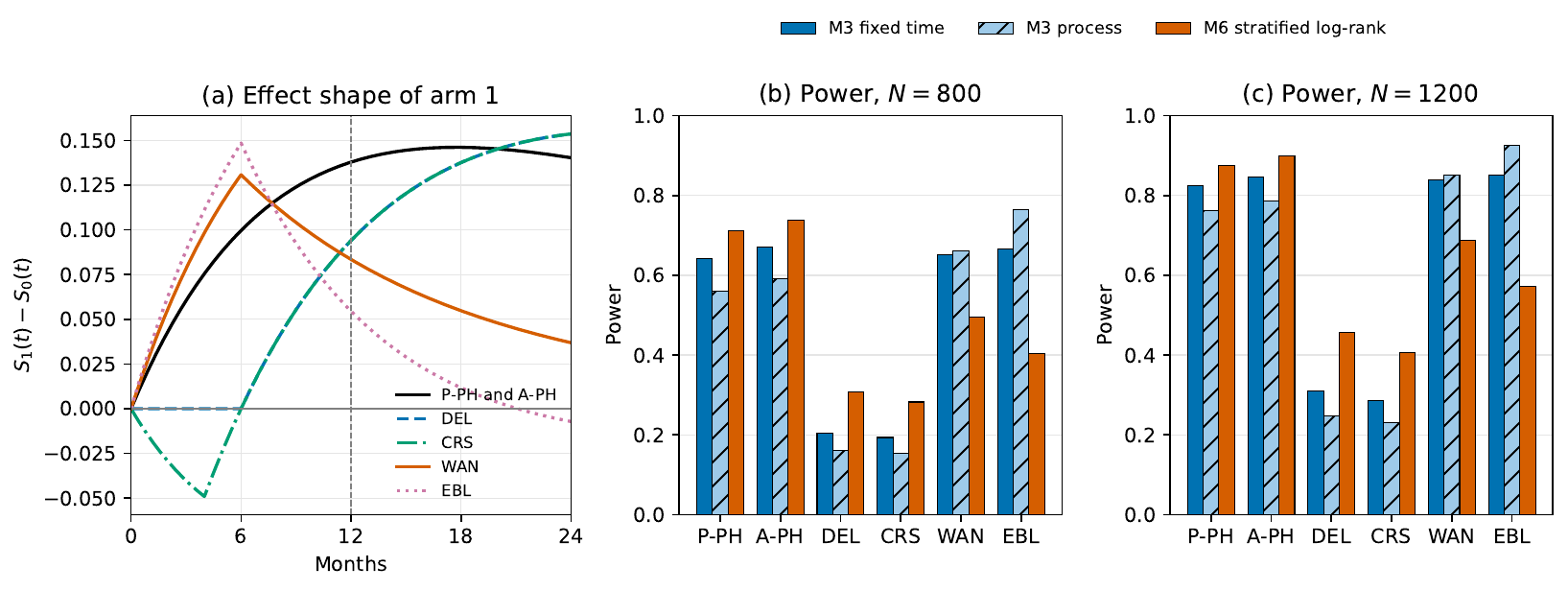}
\caption{Effect shapes and power under Pocock--Simon minimization. Panel
(a) shows the true difference between the marginal survival curves of arm
1 and the control arm, with the dashed vertical line at $\tau^\star=12$.
The DEL and CRS curves coincide after month 6 because the two arms then
have the same cumulative hazard. Panels (b) and (c) give, for each scenario,
the disjunctive power of M3 at $\tau^\star=12$, of M3 over $[4,12]$, and of
M6 with follow-up through month 12. Each bar is based on 10,000 trials, so
its Monte Carlo standard error is at most 0.005.}
\label{fig:simulation-shape-power}
\end{figure}

Figure~\ref{fig:simulation-shape-power} shows that the shape of the effect,
not its size, decides whether M3 or M6 is more powerful. M6 was more
powerful under P-PH and A-PH and under the delayed and crossing effects,
where a hazard difference persists after month 6. The ordering reversed
under WAN and EBL, although these two scenarios have the same 12-month
RMST gain as P-PH. M6 then lost a large part of its power while the power
of M3 stayed about the same, so that M3 was ahead by about 0.15 under WAN
and about 0.27 under EBL. Panel (a) explains the reversal. M3 responds to
the area between the survival curves up to month 12, which is the same in
the three scenarios, whereas the log-rank score receives no signal once the
hazards coincide and is pulled back once they cross. Extending follow-up
to 24 months helped M6 under the first four shapes but hurt it under WAN
and EBL (Section~\ref{sec:supp-logrank}). Under proportional hazards the
stratified Cox model behaved like M6 and was more powerful than M3
(Section~\ref{sec:supp-cox}). These comparisons involve different
hypotheses, since M6 tests equality of survival curves, the Cox model
targets a common hazard ratio, and M3 targets the gain in event-free time
through month 12.

The equal-RMST experiments in Tables~\ref{tab:equal-rmst-marginal-main}
and~\ref{tab:equal-rmst-conditional-main} confirm this over a range of
effect sizes. Within each hazard profile, the power of every method rises
with the RMST gain, but the ordering of M3 and M6 never changes. M3 was
more powerful for a benefit confined to the first six months, by a margin
that generally grew with the gain, whereas M6 was slightly more powerful
under proportional hazards and much more powerful under delayed benefit
and the immunotherapy-type profile. The pattern was the same under both
censoring mechanisms. The treatment-limited profile shows why. Once the
hazards become equal, the log-rank score receives no further signal, but
the survival gap built up before month 6 persists and keeps adding area to
the RMST contrast. The immunotherapy-type profile shows the converse, since
crossing survival curves do not favor RMST when the crossing is followed
by a sustained reduction in hazard. M2 and M3 had nearly identical power
under the balancing designs, M3 was more powerful under simple
randomization, where standardization acts as post-stratification, and the
strong family-wise error of all methods stayed close to the nominal level
(Section~\ref{sec:supp-equal-rmst}).

\begin{table}[p!]
\centering
\scriptsize
\caption{\textbf{Marginally independent censoring.} Power trend in the equal-RMST experiment under Pocock--Simon minimization. The columns give the arm-1 RMST gain at 12 months; arms 2 and 3 are null. Bold marks the largest displayed power within each effect shape, sample size, and RMST gain.}
\label{tab:equal-rmst-marginal-main}
\resizebox{\textwidth}{!}{%
\begin{tabular}{llrrrrrrrr}
\toprule
&& \multicolumn{4}{c}{$N=800$} & \multicolumn{4}{c}{$N=1200$}\\
\cmidrule(lr){3-6}\cmidrule(lr){7-10}
Effect shape & Method & $0.50^\dagger$ & $0.75$ & $1.00$ & $1.25$ & $0.50^\dagger$ & $0.75$ & $1.00$ & $1.25$\\
\midrule
PH & M2 & 0.122 & 0.294 & 0.546 & 0.767 & 0.191 & 0.462 & 0.739 & 0.922\\
 & M3 & 0.123 & 0.296 & 0.546 & 0.767 & 0.194 & 0.466 & 0.741 & 0.924\\
 & M6 & \textbf{0.143} & \textbf{0.349} & \textbf{0.624} & \textbf{0.841} & \textbf{0.224} & \textbf{0.538} & \textbf{0.813} & \textbf{0.961}\\
\addlinespace
Delayed benefit & M2 & 0.128 & 0.287 & 0.503 & 0.726 & 0.192 & 0.433 & 0.708 & 0.895\\
 & M3 & 0.128 & 0.289 & 0.505 & 0.727 & 0.190 & 0.434 & 0.710 & 0.894\\
 & M6 & \textbf{0.247} & \textbf{0.559} & \textbf{0.844} & \textbf{0.973} & \textbf{0.373} & \textbf{0.757} & \textbf{0.962} & \textbf{0.998}\\
\addlinespace
Treatment-limited benefit & M2 & 0.130 & 0.300 & 0.566 & 0.803 & 0.194 & 0.471 & \textbf{0.765} & 0.940\\
 & M3 & \textbf{0.133} & \textbf{0.302} & \textbf{0.568} & \textbf{0.804} & \textbf{0.195} & \textbf{0.474} & 0.764 & \textbf{0.942}\\
 & M6 & 0.096 & 0.205 & 0.394 & 0.599 & 0.142 & 0.327 & 0.570 & 0.791\\
\addlinespace
Immunotherapy-type crossing & M2 & 0.121 & 0.276 & 0.499 & 0.712 & 0.184 & 0.426 & 0.698 & 0.888\\
 & M3 & 0.120 & 0.277 & 0.503 & 0.713 & 0.187 & 0.425 & 0.703 & 0.888\\
 & M6 & \textbf{0.297} & \textbf{0.630} & \textbf{0.893} & \textbf{0.985} & \textbf{0.454} & \textbf{0.830} & \textbf{0.979} & \textbf{0.999}\\
\addlinespace
\bottomrule
\end{tabular}%
}
\begin{minipage}{0.98\textwidth}
\footnotesize
M2 and M3 target the 12-month RMST contrast; M6 tests survival-curve equality through month 12. $^\dagger$The 0.50-month gain is the prespecified moderate/local alternative. Four-design ranges and error rates are reported in the Supplementary Material.
\end{minipage}
\end{table}

\begin{table}[p!]
\centering
\scriptsize
\caption{\textbf{Conditionally independent censoring given strata.} Power trend in the equal-RMST experiment under Pocock--Simon minimization. The columns give the arm-1 RMST gain at 12 months; arms 2 and 3 are null. Bold marks the largest displayed power within each effect shape, sample size, and RMST gain.}
\label{tab:equal-rmst-conditional-main}
\resizebox{\textwidth}{!}{%
\begin{tabular}{llrrrrrrrr}
\toprule
&& \multicolumn{4}{c}{$N=800$} & \multicolumn{4}{c}{$N=1200$}\\
\cmidrule(lr){3-6}\cmidrule(lr){7-10}
Effect shape & Method & $0.50^\dagger$ & $0.75$ & $1.00$ & $1.25$ & $0.50^\dagger$ & $0.75$ & $1.00$ & $1.25$\\
\midrule
PH & M3 & 0.121 & 0.291 & 0.541 & 0.767 & 0.191 & 0.449 & 0.744 & 0.925\\
 & M6 & \textbf{0.142} & \textbf{0.342} & \textbf{0.610} & \textbf{0.830} & \textbf{0.218} & \textbf{0.513} & \textbf{0.812} & \textbf{0.959}\\
\addlinespace
Delayed benefit & M3 & 0.123 & 0.284 & 0.502 & 0.722 & 0.181 & 0.432 & 0.704 & 0.895\\
 & M6 & \textbf{0.226} & \textbf{0.527} & \textbf{0.825} & \textbf{0.969} & \textbf{0.344} & \textbf{0.735} & \textbf{0.952} & \textbf{0.997}\\
\addlinespace
Treatment-limited benefit & M3 & \textbf{0.133} & \textbf{0.304} & \textbf{0.556} & \textbf{0.792} & \textbf{0.190} & \textbf{0.471} & \textbf{0.768} & \textbf{0.937}\\
 & M6 & 0.097 & 0.213 & 0.381 & 0.596 & 0.136 & 0.331 & 0.576 & 0.789\\
\addlinespace
Immunotherapy-type crossing & M3 & 0.119 & 0.286 & 0.486 & 0.720 & 0.181 & 0.406 & 0.691 & 0.885\\
 & M6 & \textbf{0.286} & \textbf{0.615} & \textbf{0.869} & \textbf{0.979} & \textbf{0.441} & \textbf{0.799} & \textbf{0.976} & \textbf{0.999}\\
\addlinespace
\bottomrule
\end{tabular}%
}
\begin{minipage}{0.98\textwidth}
\footnotesize
M3 targets the 12-month RMST contrast; M6 tests survival-curve equality through month 12. M2 is omitted because its marginal estimator does not identify that target under stratum-dependent censoring. $^\dagger$The 0.50-month gain is the prespecified moderate/local alternative. Four-design ranges and error rates are reported in the Supplementary Material.
\end{minipage}
\end{table}

\subsection{Summary}

The simulations support five conclusions. First, M3 controlled the
family-wise error rate under all four designs and both censoring
mechanisms, and M2 did so whenever censoring was marginally independent.
Second, the conventional analysis M1 never inflated the error rate under
marginally independent censoring, but it was conservative under the
balancing designs and so gave away part of the precision gained by
balancing. Third, as Propositions~\ref{prop:additivity}
and~\ref{prop:nondegenerate} predict, ignoring the design term under
minimization was harmless under additive prognosis, and the resulting
inflation grew steadily with the strength of the interaction. Fourth,
under stratum-dependent censoring the unstratified analysis targeted the
wrong quantity and exceeded the nominal error rate, whereas M3 remained
valid. Fifth, neither test was uniformly more powerful. The stratified
log-rank test was more powerful when a hazard difference persisted through
follow-up, and the standardized RMST procedure was more powerful when a
benefit of the same 12-month size waned or turned into late harm.

\section{Application: An ICON6-Based Study}
\label{sec:application}

\subsection{The ICON6 trial}

ICON6 was a double-blind phase III trial in women with relapsed
platinum-sensitive ovarian cancer that compared two cediranib schedules
with a shared placebo control \citep{ledermann2016cediranib}. All patients
received platinum-based chemotherapy and were allocated in the ratio 2:3:3
to placebo throughout (arm A), cediranib during chemotherapy only (arm B),
or cediranib during chemotherapy followed by cediranib maintenance (arm C).
Randomization used permuted blocks within strata defined by five baseline
factors, and the primary analysis compared progression-free survival (PFS)
between arms A and C with the log-rank test. The trial combines the three
features that motivate this paper: a shared control, randomization
stratified by prognostic factors, and treatment effects of different
shapes. The benefit of arm B was similar to that of arm C during
chemotherapy but faded after the switch to placebo, whereas maintenance
sustained the benefit of arm C.

Individual patient data from ICON6 are not publicly available, so we
reconstructed them from the published PFS curves and numbers at risk with
the algorithm of \citet{guyot2012enhanced}, using the exact curve
coordinates stored in the published figure, as in \citet{zhao2025synthipd}. The reconstructed data
reproduce the reported median PFS in every arm and every published number
at risk to within one patient (Section~\ref{sec:supp-icon6-reconstruction}). Because they contain no
stratification variables, we use them in two ways. We first analyze the
reconstructed trial directly, and we then use it as the template for a
trial emulation in which treatments are reassigned under CAR.

\subsection{Analysis of the reconstructed trial}
\label{sec:application-data}

We applied the unstratified RMST procedure with the shared-control max-$T$
calibration at $\tau^\star=12$ months and over $\calT=[4,12]$, which is
valid but conservative under stratified permuted blocks by
Theorem~\ref{thm:bounded-balance-fclt}. For comparison, we applied the
log-rank test with the same calibration, once with follow-up censored at 12
months and once with all follow-up.

\begin{table}[t!]
\centering
\scriptsize
\caption{Analysis of the reconstructed ICON6 progression-free survival data.
RMST differences are in months relative to arm A, whose 12-month RMST was
8.27 months. Confidence intervals are two-sided and simultaneous across the
two comparisons. All $p$-values are one-sided and adjusted by the step-down
max-$T$ procedure.}
\label{tab:icon6-reconstructed}
\resizebox{\textwidth}{!}{%
\begin{tabular}{lcccccc}
\toprule
& \multicolumn{4}{c}{RMST at 12 months} & \multicolumn{2}{c}{Log-rank $z$ (adjusted $p$)}\\
\cmidrule(lr){2-5}\cmidrule(lr){6-7}
Comparison with arm A & Difference & 95\% simultaneous CI & Fixed-time $p$ & Process $p$ & 12 months & 24 months\\
\midrule
Concurrent (B) & 1.10 & (0.29, 1.91) & 0.002 & 0.003 & 2.39 (0.009) & 2.83 (0.002)\\
Maintenance (C) & 1.68 & (0.87, 2.48) & <0.001 & <0.001 & 4.29 (<0.001) & 4.54 (<0.001)\\
\bottomrule
\end{tabular}
}
\end{table}

\begin{figure}[t!]
\centering
\includegraphics[width=\textwidth]{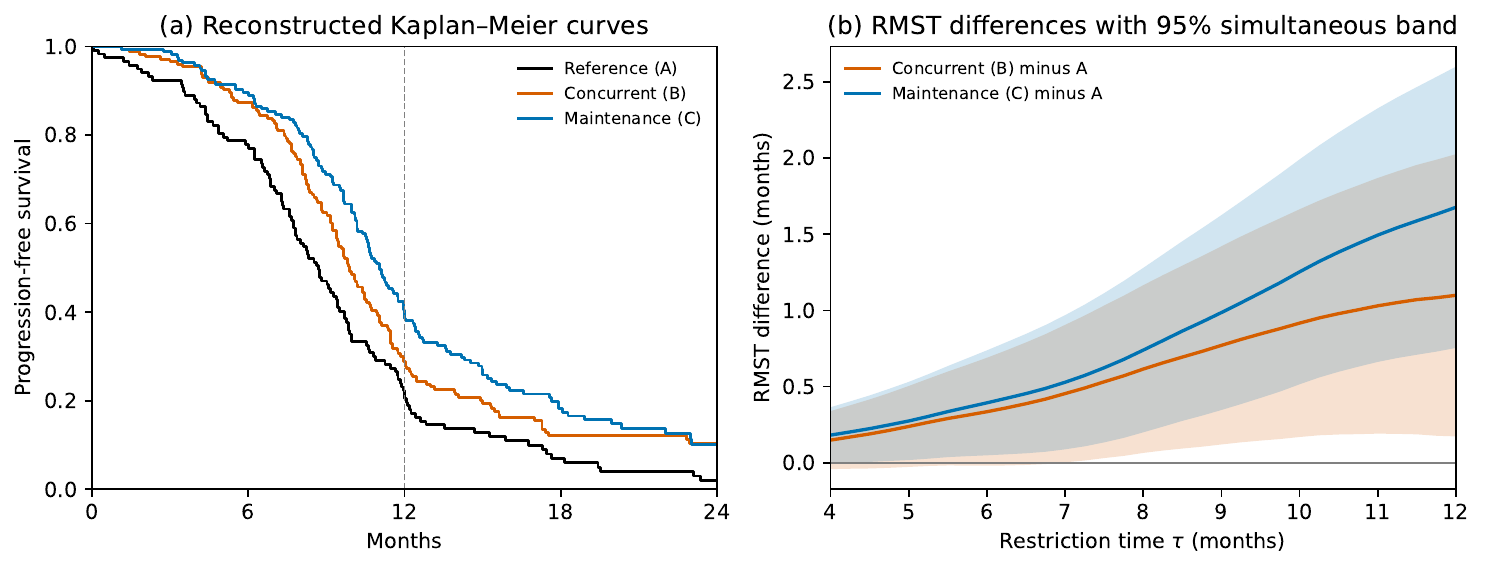}
\caption{Reconstructed ICON6 progression-free survival. Panel (a) shows the
Kaplan--Meier curves, with the dashed line at $\tau^\star=12$ months. Panel
(b) shows the estimated RMST differences relative to arm A, with two-sided
95\% confidence bands simultaneous over both comparisons and all
restriction times in $[4,12]$.}
\label{fig:icon6-application}
\end{figure}

Both cediranib schedules increased the 12-month RMST relative to placebo,
by 1.10 months for arm B and by 1.68 months for arm C, and both
simultaneous confidence intervals excluded zero
(Table~\ref{tab:icon6-reconstructed} and
Figure~\ref{fig:icon6-application}). The two tests agreed for arm C. For
arm B, the RMST statistic was 2.98, compared with 2.39 for the log-rank test
over the same 12 months and 2.83 with all follow-up. The difference follows
from the shape of the effect. The hazard ratio of arm B relative to arm A
was about 0.5 during the first six months and close to one between months 9
and 12, so the log-rank score gained little after chemotherapy ended,
whereas the survival gap built up during chemotherapy kept adding area to
the RMST contrast.

\subsection{Trial emulation}
\label{sec:application-emulation}

A single trial cannot establish the operating characteristics of a
procedure, so we used ICON6 as the template for a trial emulation. Each
emulated trial draws two of the ICON6 stratification factors, the
relapse-free interval and the planned chemotherapy regimen, from their
reported distribution. It then assigns treatments in the ratio 2:3:3 by
simple randomization, by stratified permuted blocks as in ICON6, or by
Pocock--Simon minimization, and generates PFS times so that each arm
reproduces its reconstructed ICON6 curve. Within this constraint the two
factors are prognostic, with hazard ratios of 1.6 and 1.15 that we assumed
because stratum-specific curves were not published. Censoring follows the
reconstructed censoring distribution, and under the global null all three
arms follow the curve of arm A. We used the realized sample size $N=456$
and a smaller $N=300$, with 10,000 trials per configuration, and
Section~\ref{sec:supp-icon6-model} gives the generating model.

\begin{table}[t!]
\centering
\scriptsize
\caption{Family-wise error rate and power in the ICON6-based trial emulation.}
\label{tab:icon6-emulation}
\begin{tabular}{llccc}
\toprule
Quantity & Method & SR & SPB & PS\\
\midrule
FWER, global null & M3, fixed time & \shortstack{0.022\\0.025} & \shortstack{0.025\\0.025} & \shortstack{0.025\\0.024}\\
 & M3, process & \shortstack{0.019\\0.021} & \shortstack{0.019\\0.021} & \shortstack{0.021\\0.019}\\
 & M6, 12 months & \shortstack{0.026\\0.030} & \shortstack{0.030\\0.026} & \shortstack{0.027\\0.030}\\
 & M6, 24 months & \shortstack{0.029\\0.029} & \shortstack{0.031\\0.027} & \shortstack{0.028\\0.026}\\
\addlinespace
Power, arm B (waning) & M3, fixed time & \shortstack{0.661\\0.848} & \shortstack{0.677\\0.854} & \shortstack{0.680\\0.860}\\
 & M3, process & \shortstack{0.609\\0.822} & \shortstack{0.635\\0.832} & \shortstack{0.635\\0.837}\\
 & M6, 12 months & \shortstack{0.501\\0.679} & \shortstack{0.512\\0.683} & \shortstack{0.503\\0.690}\\
 & M6, 24 months & \shortstack{0.605\\0.791} & \shortstack{0.617\\0.802} & \shortstack{0.615\\0.811}\\
\addlinespace
Power, arm C (sustained) & M3, fixed time & \shortstack{0.937\\0.992} & \shortstack{0.949\\0.994} & \shortstack{0.944\\0.994}\\
 & M3, process & \shortstack{0.896\\0.985} & \shortstack{0.914\\0.988} & \shortstack{0.907\\0.988}\\
 & M6, 12 months & \shortstack{0.905\\0.985} & \shortstack{0.915\\0.984} & \shortstack{0.909\\0.986}\\
 & M6, 24 months & \shortstack{0.931\\0.991} & \shortstack{0.937\\0.992} & \shortstack{0.936\\0.992}\\
\addlinespace
\bottomrule
\end{tabular}
\begin{minipage}{0.92\textwidth}
\footnotesize
Each cell gives the rate for $N=300$ (first line) and for the realized ICON6 sample size $N=456$ (second line), based on 10,000 emulated trials with 2:3:3 allocation. Power is the probability that the step-down procedure rejects the null hypothesis for the indicated arm. M6 with 24 months uses all available follow-up.
\end{minipage}
\end{table}

Table~\ref{tab:icon6-emulation} shows that the standardized procedure M3
controlled the family-wise error rate under all three designs, with rates
of 0.022 to 0.025, whereas the stratified log-rank test M6 was slightly
liberal, with rates of 0.026 to 0.031. For the sustained benefit of arm C,
the two procedures had similar power. For the waning benefit of arm B, M3
was markedly more powerful. At the realized sample size it rejected in 85\%
to 86\% of trials, compared with 68\% to 69\% for M6 over the same 12 months
and 79\% to 81\% for M6 with all 24 months of follow-up, and the ordering
was the same at $N=300$. M3 therefore identified the schedule whose benefit
faded more reliably than the conventional analysis, even when the latter
used twice the follow-up. Its power varied by at most 0.02 across the three
designs, as Theorem~\ref{thm:standardized-fclt} predicts. For arm B, the
fixed-time test of the marginal RMST procedure M2 remained more powerful
than the log-rank test with the same follow-up for restriction times from 9
to 24 months, whereas the standardized estimator was increasingly often
undefined beyond 12 months because some treatment--stratum risk sets became
empty (Section~\ref{sec:supp-icon6-tau}). When the two
factors were made to interact (Section~\ref{sec:supp-icon6-interaction}), omitting the design term under
minimization raised the error rate of the marginal analysis to about 0.03,
whereas M2 and M3 remained at the nominal level.

Because almost no patient in ICON6 was censored before 12 months, the
unstratified and standardized analyses performed alike in this emulation.
Section~\ref{sec:supp-icon6-censoring} reports a hypothetical variant in which early dropout depends on
the relapse-free interval and differs by arm. There the unstratified RMST
estimator was biased and the family-wise error rate of M1 and M2 rose to
about 0.05 at $N=912$, while M3 remained unbiased and kept the error rate
near the nominal level.

\section{Discussion}
\label{sec:discussion}

This article develops joint, design-aware inference for several RMST
contrasts in multi-arm trials under CAR, separating three issues that are
often conflated: dependence created by a shared control, covariance induced
by the randomization algorithm, and identification under censoring.
Theorem~\ref{thm:general-car-fclt}
gives the joint limit of the unadjusted Kaplan--Meier RMST processes under
designs with $\sqrt N$-order imbalance, including Pocock--Simon minimization,
and shows that the randomization adds a design term to the usual
covariance. This term is driven by the interaction between the
randomization factors. It vanishes when they act additively
(Proposition~\ref{prop:additivity}) and is otherwise of the order of the
squared interaction (Proposition~\ref{prop:nondegenerate}).
Theorem~\ref{thm:standardized-fclt} shows that the stratum-standardized
estimator has the same limit under all these designs, with no design term,
and remains valid when censoring depends on the strata. Because these
limits hold under arbitrary treatment effects, the multiplier procedures
built on them retain the dependence across contrasts created by the shared
control and give simultaneous bands and strong family-wise error control.

These results lead to a simple recommendation. We recommend prespecifying
the stratum-standardized procedure as the primary analysis whenever every
treatment-by-stratum cell retains subjects at risk through the upper
restriction time. It identifies the marginal RMST when censoring is
independent within strata, has the same limit under stratified blocks and
minimization, and needs no design simulation. A fixed $\tau^\star$ chosen
before unblinding is the natural confirmatory target
\citep{ich2019estimands,fda2022multiple}, and the process-wise bands are
best viewed as supportive, complementing data-adaptive choices of a single
restriction time \citep{sun2026beyond}. Because neither the RMST procedures
nor the stratified log-rank test is uniformly more powerful
(Section~\ref{sec:simulation}), the primary method should be chosen before
unblinding from the clinically relevant estimand and the plausible effect
shape.

When the unadjusted Kaplan--Meier analysis is retained, the CAR-adjusted
covariance removes the conservativeness of the usual standard error under
stratified blocks. Under minimization, the degree of interaction between
the factors is rarely known in advance, so such an analysis should include
the estimated design term. Randomization-based tests are exact under the
strong null of no effect for any subject, but they do not deliver RMST
estimates, simultaneous bands, or strong family-wise error control, so we
view them as supporting analyses \citep{bugni2018inference}.

The main limitation is that standardization requires finitely many strata
with adequate follow-up in every cell, whereas many minimization trials
balance on site and several other factors. Such trials may require a
lower-dimensional censoring model, inverse-probability-of-censoring
weighting, regression augmentation, or prespecified pooling of strata, and
extending the present multiplicity theory to those estimators is an
important next step. Extensions to competing risks and recurrent events,
and sample-size procedures for the proposed tests, would also be useful.

\section*{Supplementary material}

\begin{description}
\item[Supplementary Material:] Section~\ref{sec:supp-prelim} reviews the
counting-process background and gives the cell-level covariance formulas,
the design simulation, the multiplier construction, and step-by-step
algorithms for both procedures.
Sections~\ref{sec:supp-linearity}--\ref{sec:supp-standardized-inference}
contain the proofs of all theorems, propositions, and corollaries,
Section~\ref{sec:supp-simulation} gives simulation design details and
additional results, and Section~\ref{sec:supp-application} gives the
details of the ICON6 reconstruction and emulation together with a
reanalysis of the four-arm ACTG 175 trial.
\end{description}

\bibliographystyle{chicago}
\bibliography{references}

\clearpage
%%%%%%%%%%%%%%%%%%%%%%%%%%%% Supplementary Material %%%%%%%%%%%%%%%%%%%%%%%%%%%%
\setcounter{section}{0}
\renewcommand{\thesection}{S\arabic{section}}
\renewcommand{\theHsection}{S\arabic{section}}
\renewcommand{\theequation}{S\arabic{section}.\arabic{equation}}
\setcounter{figure}{0}
\renewcommand{\thefigure}{S\arabic{figure}}
\renewcommand{\theHfigure}{S\arabic{figure}}
\setcounter{table}{0}
\renewcommand{\thetable}{S\arabic{table}}
\renewcommand{\theHtable}{S\arabic{table}}
\setcounter{remark}{0}
\renewcommand{\theremark}{S\arabic{remark}}
\renewcommand{\theHremark}{S\arabic{remark}}
\begin{center}
{\Large\bfseries Supplementary Material}
\end{center}
\medskip

\spacingset{1.5}

\noindent
This supplement contains the implementation details, proofs, and additional
numerical results for the main paper. Sections~\ref{sec:supp-prelim}
and~\ref{sec:supp-remarks} give the counting-process background, the
implementation of both procedures, and additional remarks.
Sections~\ref{sec:supp-linearity} and~\ref{sec:supp-car-clt} establish the
uniform asymptotic linear representation of the Kaplan--Meier RMST
estimators and their joint limit under CAR, which the later proofs use
throughout, and Sections~\ref{sec:supp-main-proof}
to~\ref{sec:supp-standardized-inference} contain the proofs, in the order of
the main paper. Sections~\ref{sec:supp-simulation}
and~\ref{sec:supp-application} give the simulation design details, additional
simulation results, and the details of the ICON6-based application together
with a reanalysis of the ACTG 175 trial. Notation and assumptions are as in
the main paper. Unless prefixed by S, references to theorems, propositions,
definitions, assumptions, and equations are to the main paper.

\addtocontents{toc}{\protect\setcounter{tocdepth}{2}}
{\spacingset{1}\renewcommand{\contentsname}{Contents of the Supplementary Material}\tableofcontents}

\setcounter{section}{0}
\setcounter{equation}{0}

\supppart{Implementation details and remarks}

\section{Background and implementation details}
\label{sec:supp-prelim}

This section collects standard material and implementation details that the
main paper uses but does not need to display.

\subsection{Counting processes and Kaplan--Meier estimators}
\label{sec:supp-km}

For each arm $a$, define the observed counting and at-risk processes
\[
N_{ia}(t)
=
A_{ia}\ind(X_i\leq t,\Delta_i=1),
\qquad
Y_{ia}(t)
=
A_{ia}\ind(X_i\geq t),
\]
and their aggregates
\[
N_a(t)=\sum_{i=1}^N N_{ia}(t),
\qquad
Y_a(t)=\sum_{i=1}^N Y_{ia}(t).
\]
The arm-specific Nelson--Aalen and Kaplan--Meier estimators are
\begin{align}
\label{eq:na-estimator}
\widehat\Lambda_a(t)
&=
\int_{(0,t]}
\frac{\ind\{Y_a(u)>0\}}{Y_a(u)}\,dN_a(u),\\
\label{eq:km-estimator}
\widehat S_a(t)
&=
\prod_{(0,t]}
\{1-d\widehat\Lambda_a(u)\}.
\end{align}
The estimated RMST and treatment-control contrast are
\begin{align}
\label{eq:rmst-estimator}
\widehat\mu_a(\tau)
&=
\int_0^\tau\widehat S_a(t)\,dt,\\
\label{supp-eq:contrast-estimator}
\widehat\btheta(\tau)
&=
\bm C\widehat\bmu(\tau),
\qquad
\widehat\bmu(\tau)
=
\{\widehat\mu_0(\tau),\ldots,\widehat\mu_K(\tau)\}^{T}.
\end{align}

\subsection{Population quantities and influence functions}
\label{sec:supp-if}

For arm $a\in\{0,\ldots,K\}$, recall the potential counting and at-risk
processes
\[
N_i^{(a)}(t)
=
\ind\{X_i(a)\leq t,\Delta_i(a)=1\},
\qquad
Y_i^{(a)}(t)
=
\ind\{X_i(a)\geq t\}.
\]
Let
\[
y_a(t)=\E\{Y_i^{(a)}(t)\},
\qquad
F_{1a}(t)=\E\{N_i^{(a)}(t)\},
\]
and
\[
\Lambda_a(t)
=
\int_{(0,t]}\frac{dF_{1a}(u)}{y_a(u)},
\qquad
S_a(t)
=
\prod_{(0,t]}\{1-d\Lambda_a(u)\}.
\]
The mean-zero residual process is
\[
M_i^{(a)}(t)
=
N_i^{(a)}(t)
-
\int_{(0,t]}Y_i^{(a)}(u)\,d\Lambda_a(u).
\]
For $\tau\in\calT$, define
\begin{align}
\label{eq:s-survival-if}
\psi_{ia}(t)
&=
-S_a(t)
\int_{(0,t]}
\frac{dM_i^{(a)}(u)}
{y_a(u)\{1-\Delta\Lambda_a(u)\}},\\
\label{eq:s-rmst-if}
\phi_{ia}(\tau)
&=
\int_0^\tau\psi_{ia}(t)\,dt.
\end{align}
The conditional mean and covariance kernel are
\begin{align}
\label{eq:s-m}
m_a(s,\tau)
&=
\E\{\phi_{ia}(\tau)\mid Z_i=s\},\\
\label{eq:s-c}
c_a(s;\tau,\tau')
&=
\Cov\{\phi_{ia}(\tau),\phi_{ia}(\tau')\mid Z_i=s\}.
\end{align}
Because $\E\{\phi_{ia}(\tau)\}=0$,
\begin{equation}
\label{eq:s-m-centering}
\sum_{s=1}^Jp_s m_a(s,\tau)=0
\quad\text{for every }a\text{ and }\tau\in\calT.
\end{equation}

The empirical-process arguments below use the following process-regularity
condition: the processes $\phi_{ia}(\cdot)$ have square-integrable envelopes
and are asymptotically equicontinuous on $\calT$, a sufficient condition
being that $|\phi_{ia}(\tau)-\phi_{ia}(\tau')|\leq L_{ia}|\tau-\tau'|$ for
random variables $L_{ia}$ with $\E(L_{ia}^{2+\eta})<\infty$ for some
$\eta>0$. The analogous condition is used for the stratum-specific
processes $\phi_{ia}^{s}(\cdot)$. The first lemma shows that the positivity
conditions of the main paper already imply this condition, so it imposes no
additional restriction.

\begin{lemma}[Bounded Lipschitz influence functions]
\label{lem:if-bounds}
Under Assumption~\ref{asu:censoring} of the main paper, there is a finite nonrandom constant
$B$, depending only on $\epsilon$ and $\tau_U$, such that for every arm $a$,
\[
\sup_{0\leq t\leq\tau_U}|\psi_{ia}(t)|\leq B
\qquad\text{almost surely}.
\]
Consequently,
$\sup_{\tau\in\calT}|\phi_{ia}(\tau)|\leq B\tau_U$ and
\[
|\phi_{ia}(\tau)-\phi_{ia}(\tau')|
\leq B|\tau-\tau'|
\qquad\text{for all }\tau,\tau'\in\calT,
\]
so the process-regularity condition holds with the nonrandom envelope
$L_{ia}=B$. Under Assumption~\ref{asu:conditional-censoring} of the main paper, the same conclusions hold
for the stratum-specific processes $\phi_{ia}^{s}(\cdot)$ with a constant
depending only on $\epsilon$ and $\tau_U$, so the stratum-specific
regularity condition also holds automatically.
\end{lemma}

\begin{proof}
Fix an arm $a$ and $t\in[0,\tau_U]$. By Assumption~\ref{asu:censoring},
$y_a(u)\geq y_a(\tau_U)\geq\epsilon$ and
$1-\Delta\Lambda_a(u)\geq\epsilon$ for all $u\in(0,\tau_U]$, and
$\Lambda_a(\tau_U)\leq F_{1a}(\tau_U)/y_a(\tau_U)\leq1/\epsilon$.
Since $0\leq S_a(t)\leq1$ and
$N_i^{(a)}$ is a $\{0,1\}$-valued nondecreasing process,
\begin{align*}
|\psi_{ia}(t)|
&\leq
\int_{(0,t]}
\frac{dN_i^{(a)}(u)}{y_a(u)\{1-\Delta\Lambda_a(u)\}}
+
\int_{(0,t]}
\frac{Y_i^{(a)}(u)\,d\Lambda_a(u)}
{y_a(u)\{1-\Delta\Lambda_a(u)\}}\\
&\leq
\frac{1}{\epsilon^2}
+
\frac{\Lambda_a(\tau_U)}{\epsilon^2}
\leq
\frac{1}{\epsilon^2}\left(1+\frac{1}{\epsilon}\right)
=:B.
\end{align*}
The bound and the Lipschitz property of
$\phi_{ia}(\tau)=\int_0^\tau\psi_{ia}(t)\,dt$ follow immediately. The
stratum-specific statement is identical with
$(y_{a,s},\Lambda_{a,s},S_a(\cdot\mid s))$ in place of
$(y_a,\Lambda_a,S_a)$, using the stratum-level positivity bounds in
Assumption~\ref{asu:conditional-censoring}.
\end{proof}

\subsection{Cell-level covariance estimators for the marginal analysis}
\label{sec:supp-cov-details}

Let $n_a=\sum_{i=1}^N A_{ia}$ and define
\[
\widehat y_a(u)=\frac{Y_a(u)}{n_a}.
\]
For each subject, define the fitted counting-process residual
\[
\widehat M_{ia}(t)
=
N_{ia}(t)
-
\int_{(0,t]}Y_{ia}(u)\,d\widehat\Lambda_a(u).
\]
The estimated RMST influence function is
\begin{equation}
\label{eq:estimated-rmst-if}
\widehat\phi_{ia}(\tau)
=
-\int_0^\tau
\widehat S_a(t)
\left[
\int_{(0,t]}
\frac{d\widehat M_{ia}(u)}
{\widehat y_a(u)\{1-\Delta\widehat\Lambda_a(u)\}}
\right]dt.
\end{equation}
It is zero for subjects not assigned to arm $a$ because both
$N_{ia}$ and $Y_{ia}$ are then zero.

Let $\widehat p_s=n(s)/N$. For $n_a(s)>0$, estimate the conditional
influence-function mean by
\begin{equation}
\label{eq:mhat}
\widehat m_a(s,\tau)
=
\frac{1}{n_a(s)}
\sum_{i=1}^N
A_{ia}\ind(Z_i=s)\widehat\phi_{ia}(\tau).
\end{equation}
For finite-sample stabilization, center these estimates using
\begin{equation}
\label{eq:mtilde}
\overline m_a(\tau)
=
\sum_{s=1}^J\widehat p_s\widehat m_a(s,\tau),
\qquad
\widetilde m_a(s,\tau)
=
\widehat m_a(s,\tau)-\overline m_a(\tau).
\end{equation}
The conditional covariance kernel is estimated by
\begin{align}
\label{eq:chat}
\widehat c_a(s;\tau,\tau')
={}&
\frac{1}{n_a(s)-1}
\sum_{i=1}^N A_{ia}\ind(Z_i=s)
\{\widehat\phi_{ia}(\tau)-\widehat m_a(s,\tau)\}\notag\\
&\hspace{35mm}\times
\{\widehat\phi_{ia}(\tau')-\widehat m_a(s,\tau')\}.
\end{align}
Because $p_s\pi_a>0$, every $n_a(s)$ diverges in probability. Zero or
singleton cells are flagged as operationally unstable and their current
contribution may be set to zero without affecting the asymptotic result.
The denominator $n_a(s)-1$ is a finite-sample degrees-of-freedom correction.
Replacing it by $n_a(s)$ gives the asymptotically equivalent uncorrected
plug-in estimator, but the corrected version is used by default.

The bounded-balance component of the arm-specific covariance kernel is
\begin{align}
\label{supp-eq:base-kernel-estimator}
\widehat\Gamma_{\mu,ab}(\tau,\tau')
={}&
\sum_{s=1}^J\widehat p_s
\widetilde m_a(s,\tau)\widetilde m_b(s,\tau')\notag\\
&+
\ind(a=b)\frac{1}{\pi_a}
\sum_{s=1}^J\widehat p_s
\widehat c_a(s;\tau,\tau').
\end{align}

\subsection{Design simulation for the imbalance covariance}
\label{sec:supp-design-sim}

The matrix $\bm\Omega_D$ depends only on the stratum distribution and the
prespecified randomization algorithm, not on the survival outcomes. We
estimate it by reproducing the design. Set
\[
\widehat{\bm p}=(\widehat p_1,\ldots,\widehat p_J)^T
\]
and independently repeat the following procedure
$B_{\mathrm{MC}}$ times.

\begin{enumerate}[label=\textbf{Step \arabic*:},leftmargin=17mm]
\item Generate $N$ stratum profiles independently from
      $\widehat{\bm p}$. When Pocock--Simon minimization is based on several
      factors, generate the complete joint profile so that the same marginal
      imbalance scores can be reconstructed.
\item Apply exactly the randomization algorithm used by the trial, including
      its factor weights, allocation probabilities, target ratios, and
      tie-breaking rule.
\item Calculate all simulated within-stratum imbalances and stack
      \[
      \bm d_N^{(b)}
      =
      N^{-1/2}\operatorname{vec}
      \{D_{N,a}^{(b)}(s):a=0,\ldots,K;\ s=1,\ldots,J\}.
      \]
\end{enumerate}

Let
\[
\overline{\bm d}_N
=
\frac{1}{B_{\mathrm{MC}}}
\sum_{b=1}^{B_{\mathrm{MC}}}\bm d_N^{(b)}.
\]
The design covariance estimator is
\begin{equation}
\label{eq:omega-hat}
\widehat{\bm\Omega}_D
=
\frac{1}{B_{\mathrm{MC}}-1}
\sum_{b=1}^{B_{\mathrm{MC}}}
\{\bm d_N^{(b)}-\overline{\bm d}_N\}
\{\bm d_N^{(b)}-\overline{\bm d}_N\}^{T}.
\end{equation}
This is the multi-arm analogue of the minimization covariance estimator of
\citet{zhao2024consistent}.

Define $\widehat{\bm Q}(\tau)$ by replacing
$m_a(s,\tau)$ in the matrix $\bm Q(\tau)$ of Theorem~\ref{thm:general-car-fclt} with
$\widetilde m_a(s,\tau)$. The proposed covariance estimators are
\begin{align}
\label{eq:general-mu-kernel-estimator}
\widehat{\bm\Gamma}_\mu^{D}(\tau,\tau')
&=
\widehat{\bm\Gamma}_\mu(\tau,\tau')
+
\widehat{\bm Q}(\tau)
\widehat{\bm\Omega}_D
\widehat{\bm Q}(\tau')^T,\\
\label{supp-eq:general-theta-kernel-estimator}
\widehat{\bm\Gamma}_\theta^{D}(\tau,\tau')
&=
\bm C
\widehat{\bm\Gamma}_\mu^{D}(\tau,\tau')
\bm C^T.
\end{align}

\subsection{Multiplier construction and implementation}
\label{sec:supp-multiplier-details}

The covariance estimator in Section~\ref{sec:supp-cov-details} permits
direct Gaussian simulation on a time grid. For theoretical development and
subject-level implementation, it is useful to construct an equivalent
multiplier process that mirrors the decomposition in the proofs of
Theorems~\ref{thm:bounded-balance-fclt} and
\ref{thm:general-car-fclt}.

For a subject in stratum $s$, define
\begin{equation}
\label{supp-eq:bootstrap-summand}
\widehat u_{ia}(\tau)
=
\frac{A_{ia}}{\pi_a}
\{\widehat\phi_{ia}(\tau)-\widehat m_a(s,\tau)\}
+
\widetilde m_a(s,\tau),
\qquad s=Z_i.
\end{equation}
To reproduce the finite-sample-corrected covariance in
\eqref{eq:chat}, the implemented multiplier summand replaces the centered
residual in \eqref{supp-eq:bootstrap-summand} by
\[
\widehat\kappa_a(s)
\{\widehat\phi_{ia}(\tau)-\widehat m_a(s,\tau)\},
\qquad
\widehat\kappa_a(s)
=
\left\{
\frac{\pi_a n(s)}{n_a(s)-1}
\right\}^{1/2}.
\]
Because $n_a(s)/n(s)\to\pi_a$, this multiplier converges to one and does
not change the limiting process.
Let $\xi_1,\ldots,\xi_N$ be independent random multipliers, independent of
the data, with
\[
\E(\xi_i)=0,\qquad
\E(\xi_i^2)=1,\qquad
\E|\xi_i|^{2+\eta}<\infty
\]
for some $\eta>0$. Standard normal or Rademacher multipliers may be used.
Independently generate
\[
\bm D^*\mid\mathcal D_N
\sim
\mathcal N_L(\bm 0,\widehat{\bm\Omega}_D),
\]
where $\mathcal D_N$ denotes the observed data. The singularity of
$\widehat{\bm\Omega}_D$ causes no problem because $\bm D^*$ can be generated
from its nonnegative eigendecomposition.

The arm-specific and treatment-control bootstrap processes are
\begin{align}
\label{supp-eq:bootstrap-mu-process}
\mathbb G_{\mu,a}^*(\tau)
&=
\frac{1}{\sqrt N}
\sum_{i=1}^N\xi_i\widehat u_{ia}(\tau)
+
\widehat{\bm Q}_a(\tau)\bm D^*,\\
\label{eq:bootstrap-theta-process}
\mathbb G_\theta^*(\tau)
&=
\bm C\mathbb G_\mu^*(\tau),
\end{align}
where $\widehat{\bm Q}_a(\tau)$ is row $a$ of
$\widehat{\bm Q}(\tau)$. Under bounded-balance randomization,
$\widehat{\bm\Omega}_D=\bm 0$ and the second term in
\eqref{supp-eq:bootstrap-mu-process} is omitted.

On a fixed time grid, standard-normal subject multipliers make the first
term in \eqref{supp-eq:bootstrap-mu-process} conditionally Gaussian with its
empirical influence-function covariance. Since the independent design draw
is also Gaussian, sampling directly from
$\mathcal N\{\bm 0,\widehat{\bm\Gamma}_\theta^D\}$ is conditionally
distributionally identical and is the computational default. The explicit
subject-level normal implementation is retained as a verification option,
and Rademacher multipliers provide an asymptotically equivalent sensitivity
analysis.

\begin{remark}[Grid implementation]
\label{rmk:grid}
In computation, evaluate the process on a deterministic grid
$\calT_N\subset\calT$ containing $\tau_L$ and $\tau_U$. If the mesh width of
$\calT_N$ converges to zero, asymptotic equicontinuity implies that replacing
the continuum supremum by the grid maximum does not change the limiting
critical value. A clinically chosen finite set of restriction times may also
be used, in which case the result is an ordinary finite-dimensional
simultaneous confidence region.
\end{remark}

\paragraph{Algorithm for the CAR-adjusted marginal analysis.}

For a selected time grid $\calT_N$ and
$B_{\mathrm{boot}}$ bootstrap repetitions, the procedure is:
\begin{enumerate}[label=\textbf{Step \arabic*:},leftmargin=17mm]
\item Prespecify whether the confirmatory target is a fixed restriction time
      $\tau^\star$, the process over $\calT$, or both, and select the
      evaluation grid $\calT_N$ before examining treatment effects.
\item Fit the arm-specific Kaplan--Meier estimators and calculate
      $\widehat\theta_k(\tau)$,
      $\widehat\phi_{ia}(\tau)$,
      $\widehat m_a(s,\tau)$, and
      $\widehat c_a(s;\tau,\tau')$ on $\calT_N$.
\item For Pocock--Simon minimization, estimate
      $\widehat{\bm\Omega}_D$ using the design simulation in
      Section~\ref{sec:supp-design-sim}. For bounded-balance designs,
      set it equal to zero.
\item Calculate $\widehat{\bm\Gamma}_\theta^D$,
      $\widehat\sigma_{k,\mathrm{CAR}}$, and the observed statistics
      $V_{k,\mathrm{CAR}}^{\mathrm{fix}}$ or
      $V_{k,\mathrm{CAR}}^{\mathrm{proc}}$ from
      Definition~\ref{def:car-adjusted-statistics}.
\item For bootstrap repetition $b$, generate
      $\xi_1^{(b)},\ldots,\xi_N^{(b)}$ and an independent
      $\bm D^{*(b)}$, then evaluate
      $\mathbb G_\theta^{*(b)}(\tau)$ for every arm and grid point.
\item Studentize the stored multiplier processes, calculate their fixed-time
      or process-wise maxima, and obtain the simultaneous-band and
      single-step critical values.
\item Compare the observed statistics with the joint critical values. For
      step-down inference, recompute subset maxima from the same stored
      multiplier draws. No additional Kaplan--Meier fitting or design
      simulation is required.
\end{enumerate}

\subsection{Stratum-standardized estimator}
\label{sec:supp-standardized-details}

For a subject in stratum $s$, let
\[
y_{a,s}(t)
=
\E\{Y_i^{(a)}(t)\mid Z_i=s\},
\qquad
F_{1a,s}(t)
=
\E\{N_i^{(a)}(t)\mid Z_i=s\},
\]
and define
\[
\Lambda_{a,s}(t)
=
\int_{(0,t]}
\frac{dF_{1a,s}(u)}{y_{a,s}(u)}.
\]
Under Assumption~\ref{asu:conditional-censoring},
\[
S_a(t\mid s)
=
\prod_{(0,t]}\{1-d\Lambda_{a,s}(u)\}.
\]
For a subject with $Z_i=s$, define
\begin{align}
\label{eq:stratum-martingale-residual}
M_{ia}^{s}(t)
&=
N_i^{(a)}(t)
-
\int_{(0,t]}Y_i^{(a)}(u)\,d\Lambda_{a,s}(u),\\
\label{eq:stratum-rmst-if}
\phi_{ia}^{s}(\tau)
&=
-\int_0^\tau S_a(t\mid s)
\left[
\int_{(0,t]}
\frac{dM_{ia}^{s}(u)}
{y_{a,s}(u)\{1-\Delta\Lambda_{a,s}(u)\}}
\right]dt.
\end{align}
The defining property is
\[
\E\{\phi_{ia}^{s}(\tau)\mid Z_i=s\}=0.
\]
Write
\begin{equation}
\label{supp-eq:stratum-if-covariance}
v_a(s;\tau,\tau')
=
\E\{\phi_{ia}^{s}(\tau)\phi_{ia}^{s}(\tau')\mid Z_i=s\}.
\end{equation}

Define the observed stratum-specific counting and at-risk processes
\begin{align*}
N_{a,s}(t)
&=
\sum_{i=1}^N
A_{ia}\ind(Z_i=s)\ind(X_i\leq t,\Delta_i=1),\\
Y_{a,s}(t)
&=
\sum_{i=1}^N
A_{ia}\ind(Z_i=s)\ind(X_i\geq t).
\end{align*}
The corresponding Nelson--Aalen, Kaplan--Meier, and RMST estimators are
\begin{align}
\label{eq:stratum-km-estimators}
\widehat\Lambda_{a,s}(t)
&=
\int_{(0,t]}
\frac{\ind\{Y_{a,s}(u)>0\}}{Y_{a,s}(u)}\,dN_{a,s}(u),\notag\\
\widehat S_a(t\mid s)
&=
\prod_{(0,t]}\{1-d\widehat\Lambda_{a,s}(u)\},\notag\\
\widehat\mu_a(s,\tau)
&=
\int_0^\tau\widehat S_a(t\mid s)\,dt.
\end{align}

Let
\[
\widehat\pi_a(s)=\frac{n_a(s)}{n(s)}
\]
denote the realized allocation fraction in stratum $s$. For a subject with
$A_i=a$ and $Z_i=s$, define
\begin{align}
\label{eq:estimated-stratum-residual}
\widehat M_{ia}^{s}(t)
&=
\ind(X_i\leq t,\Delta_i=1)
-
\int_{(0,t]}\ind(X_i\geq u)\,d\widehat\Lambda_{a,s}(u),\\
\label{eq:estimated-stratum-if}
\widehat\phi_{ia}^{s}(\tau)
&=
-\int_0^\tau\widehat S_a(t\mid s)
\left[
\int_{(0,t]}
\frac{d\widehat M_{ia}^{s}(u)}
{\widehat y_{a,s}(u)
\{1-\Delta\widehat\Lambda_{a,s}(u)\}}
\right]dt,
\qquad
\widehat y_{a,s}(u)=\frac{Y_{a,s}(u)}{n_a(s)}.
\end{align}
Center the estimated influence functions within each treatment--stratum
cell:
\begin{align}
\label{eq:stratum-if-centering}
\overline\phi_a(s,\tau)
&=
\frac{1}{n_a(s)}
\sum_{i=1}^N
A_{ia}\ind(Z_i=s)\widehat\phi_{ia}^{s}(\tau),\\
\label{eq:stratum-vhat}
\widehat v_a(s;\tau,\tau')
&=
\frac{1}{n_a(s)-1}
\sum_{i=1}^N A_{ia}\ind(Z_i=s)
\{\widehat\phi_{ia}^{s}(\tau)-\overline\phi_a(s,\tau)\}\notag\\
&\hspace{35mm}\times
\{\widehat\phi_{ia}^{s}(\tau')-\overline\phi_a(s,\tau')\}.
\end{align}
The $n_a(s)-1$ denominator is the same finite-sample
degrees-of-freedom correction used in \eqref{eq:chat}, and it is
asymptotically immaterial because every treatment--stratum cell diverges.
Also set
\begin{equation}
\label{eq:stratum-ghat}
\widehat g_a(s,\tau)
=
\widehat\mu_a(s,\tau)
-
\widehat\mu_a^{\mathrm{st}}(\tau).
\end{equation}
The arm-specific and treatment-control covariance estimators are
\begin{align}
\label{eq:standardized-mu-kernel-estimator}
\widehat\Gamma_{\mu,ab}^{\mathrm{st}}(\tau,\tau')
={}&
\sum_{s=1}^J
\widehat p_s
\widehat g_a(s,\tau)\widehat g_b(s,\tau')\notag\\
&+
\ind(a=b)
\sum_{s=1}^J
\frac{\widehat p_s}{\widehat\pi_a(s)}
\widehat v_a(s;\tau,\tau'),\\
\label{supp-eq:standardized-theta-kernel-estimator}
\widehat{\bm\Gamma}_\theta^{\mathrm{st}}(\tau,\tau')
&=
\bm C
\widehat{\bm\Gamma}_\mu^{\mathrm{st}}(\tau,\tau')
\bm C^T.
\end{align}

For the multiplier approximation, define
\begin{equation}
\label{eq:standardized-multiplier-summand}
\widehat u_{ia}^{\mathrm{st}}(\tau)
=
\widehat g_a(Z_i,\tau)
+
\frac{A_{ia}}{\widehat\pi_a(Z_i)}
\{\widehat\phi_{ia}^{Z_i}(\tau)
-\overline\phi_a(Z_i,\tau)\},
\end{equation}
where the second term is evaluated only when $A_i=a$. With independent
mean-zero, variance-one multipliers $\xi_1,\ldots,\xi_N$, set
\begin{align}
\label{supp-eq:standardized-bootstrap-process}
\mathbb G_{\mu,a}^{\mathrm{st},*}(\tau)
&=
\frac{1}{\sqrt N}
\sum_{i=1}^N
\xi_i\widehat u_{ia}^{\mathrm{st}}(\tau),&
\mathbb G_\theta^{\mathrm{st},*}(\tau)
&=
\bm C\mathbb G_\mu^{\mathrm{st},*}(\tau).
\end{align}
No independent design draw $\bm D^*$ is included.
In the finite-sample-corrected implementation, the centered residual in
\eqref{eq:standardized-multiplier-summand} is multiplied by
\[
\widehat\kappa_a^{\mathrm{st}}(s)
=
\left\{\frac{n_a(s)}{n_a(s)-1}\right\}^{1/2},
\]
which exactly matches \eqref{eq:stratum-vhat} and converges to one.

\paragraph{Algorithm for the standardized analysis.}

The standardized analysis can be implemented as follows.
\begin{enumerate}[label=\textbf{Step \arabic*:},leftmargin=17mm]
\item Use the joint strata prespecified in the randomization algorithm and
      choose $\tau^\star$ or $\calT$ before examining treatment effects.
\item Fit one Kaplan--Meier curve in each treatment--stratum cell and
      calculate $\widehat\mu_a(s,\tau)$.
\item Average the stratum-specific RMST estimates using the same
      $\widehat p_s=n(s)/N$ for every arm to obtain
      $\widehat\btheta^{\mathrm{st}}(\tau)$.
\item Calculate \eqref{supp-eq:standardized-theta-kernel-estimator},
      $\widehat\sigma_{k,\mathrm{st}}$, and the observed statistics in
      Definition~\ref{def:standardized-statistics}.
\item Generate the multiplier process
      \eqref{supp-eq:standardized-bootstrap-process}, studentize each draw using
      $\widehat\sigma_{k,\mathrm{st}}$, and store the fixed-time or
      process-wise maxima.
\item Apply the simultaneous-band, single-step, or step-down rule using the
      stored multiplier draws. No simulation of the original randomization
      algorithm is required.
\end{enumerate}

\section{Additional remarks}
\label{sec:supp-remarks}

\begin{remark}[Why a matrix rather than one variance parameter]
\label{rmk:full-omega}
The RMST process uses different stratum-specific functions
$m_a(s,\tau)$ for different arms and restriction times. Therefore, a single
scalar imbalance variance cannot generally recover
$\bm R(\tau)\bm\Omega_D\bm R(\tau')^T$. Retaining the full matrix also covers
unequal stratum prevalences and interactions among the randomization factors.
No inverse of $\bm\Omega_D$ is required.
\end{remark}

\begin{remark}[Two-sided hypotheses]
\label{rmk:two-sided-fwer}
For two-sided fixed-time hypotheses
$H_{0k}:\theta_k(\tau^\star)=0$, replace $V_k^{\mathrm{fix}}$ and its
bootstrap analogue by their absolute values. For the functional equality
$H_{0k}:\theta_k(\tau)=0$ for every $\tau\in\calT$, use the supremum of the
absolute studentized process. The proof of
Corollary~\ref{thm:strong-fwer} is unchanged.
\end{remark}

\begin{remark}[Nonprognostic CAR strata]
If
\[
m_a(s,\tau)=0
\quad\text{for every }a,s,\tau,
\]
then random variation in the stratum composition contributes no first-order
variance. The treatment-control covariance reduces to
\[
\Gamma_{\theta,k\ell}(\tau,\tau')
=
\ind(k=\ell)\frac{1}{\pi_k}
\E\{c_k(Z_i;\tau,\tau')\}
+
\frac{1}{\pi_0}
\E\{c_0(Z_i;\tau,\tau')\}.
\]
The off-diagonal covariance is then entirely due to the shared control.
\end{remark}

\begin{remark}[Continuous event times]
When $\Lambda_a$ is continuous,
$1-\Delta\Lambda_a(u)=1$, and the RMST influence function can be written as
\[
\phi_{ia}(\tau)
=
-\int_{(0,\tau]}
\frac{\int_u^\tau S_a(t)\,dt}{y_a(u)}
\,dM_i^{(a)}(u).
\]
This equivalent form is often more convenient for implementation.
\end{remark}

\begin{remark}[Why the global null is unnecessary]
Every step of the proofs of Theorems~\ref{thm:bounded-balance-fclt} to~\ref{thm:standardized-fclt} is arm-specific and centered at
the true survival
function $S_a$. No equality among $S_0,\ldots,S_K$ is used. Hence the same
joint Gaussian approximation remains valid when some treatment-control null
hypotheses are true and others are false. This property is the key input for
the proof of strong family-wise error control in
Section~\ref{sec:supp-fwer}.
\end{remark}

\supppart{Proofs}

\section{Uniform asymptotic linearity}
\label{sec:supp-linearity}

The linearization arguments below repeatedly use empirical processes formed
within treatment--stratum cells, whose sizes are random under CAR. The next
lemma makes the conditioning argument explicit. For an arm $a$, a stratum
$s$, and a class $\mathcal F$ of functions of $\{X_i(a),\Delta_i(a)\}$, define
the cell empirical process
\[
\mathbb H_{N,as}(f)
=
\frac{1}{\sqrt{n_a(s)}}
\sum_{i=1}^N
A_{ia}\ind(Z_i=s)
\left[f\{X_i(a),\Delta_i(a)\}-P_sf\right],
\qquad f\in\mathcal F,
\]
where $P_s$ is the conditional law of $\{X_i(a),\Delta_i(a)\}$ given
$Z_i=s$, and set $\mathbb H_{N,as}=0$ when $n_a(s)=0$.

\begin{lemma}[Cell empirical processes with random cell sizes]
\label{lem:cell-processes}
Suppose Assumptions~\ref{asu:sampling} and~\ref{asu:randomization} of the main paper hold and that
$n_a(s)\to\infty$ in probability for every arm $a$ and stratum $s$. Let
$\mathcal F$ be a $P_s$-Donsker class with a bounded envelope, for every
$s$. Then, jointly over the finitely many cells, the processes
$\mathbb H_{N,as}$ converge in distribution in $\ell^\infty(\mathcal F)$ to
independent tight $P_s$-Brownian bridges $\mathbb H_{as}$. Moreover, the
convergence is stable with respect to $(\mathcal Z_N,\mathcal A_N)$: for
every bounded Lipschitz function $h$ on the product space and every
bounded random variable $\xi_N$ that is measurable with respect to
$(\mathcal Z_N,\mathcal A_N)$ and converges in distribution jointly with
other such variables,
\[
\E^*\left[h\{(\mathbb H_{N,as})_{a,s}\}\,\xi_N\right]
-
\E\left[h\{(\mathbb H_{as})_{a,s}\}\right]\E(\xi_N)
\longrightarrow0.
\]
\end{lemma}

\begin{proof}
Fix the complete stratum and assignment sequences. By Assumption~\ref{asu:sampling}, the
full-data vectors are independent and identically distributed. By
Assumption~\ref{asu:randomization}, the allocation sequence is independent of the potential
outcomes given $(Z_1,\ldots,Z_N)$. Hence, conditionally on
$(\mathcal Z_N,\mathcal A_N)$, the outcomes of the subjects in cell $(a,s)$
are independent draws from $P_s$, and outcomes in different cells are
independent. Given $n_a(s)=m_{as}$ for every cell, the joint conditional law
of $(\mathbb H_{N,as})_{a,s}$ is therefore the product of the laws of
ordinary empirical processes $\mathbb H_{m_{as}}$ based on $m_{as}$
independent draws from $P_s$.

For a single cell, let
\[
\delta_s(m)
=
\sup_{h\in\mathrm{BL}_1}
\left|\E^*h(\mathbb H_m)-\E h(\mathbb H_s)\right|.
\]
Because $\mathcal F$ is $P_s$-Donsker, $\delta_s(m)\to0$ as $m\to\infty$
\citep[Section~1.12]{vaart1996weak}, and $0\leq\delta_s(m)\leq2$. The
product structure extends this to all cells: the bounded-Lipschitz distance
between the joint conditional law and the law of independent limits is at
most $\delta(\bm m)=\sum_{a,s}\delta_s(m_{as})$. Therefore
\[
\left|
\E^*\left[h\{(\mathbb H_{N,as})_{a,s}\}\,\xi_N\right]
-
\E\left[h\{(\mathbb H_{as})_{a,s}\}\right]\E(\xi_N)
\right|
\leq
\|\xi_N\|_\infty\,
\E\left[\delta\{(n_a(s))_{a,s}\}\right]
\]
for $h\in\mathrm{BL}_1$. Since every $n_a(s)\to\infty$ in probability and
$\delta$ is bounded, the right-hand side tends to zero by dominated
convergence. Taking $\xi_N\equiv1$ gives the joint convergence in
distribution, and general $\xi_N$ gives stability.
\end{proof}

\begin{lemma}[Arm-specific RMST expansion]
\label{lem:ual}
Under Assumptions~\ref{asu:sampling}--\ref{asu:censoring} of the main paper, for every
$a=0,\ldots,K$,
\begin{equation}
\label{eq:s-ual}
\sup_{\tau\in\calT}
\left|
\sqrt N\{\widehat\mu_a(\tau)-\mu_a(\tau)\}
-
\frac{1}{\sqrt N}
\sum_{i=1}^N\frac{A_{ia}}{\pi_a}\phi_{ia}(\tau)
\right|
=o_p(1).
\end{equation}
The expansions hold jointly over the $K+1$ arms.
\end{lemma}

\begin{proof}
For a fixed arm $a$, introduce the scaled empirical at-risk and failure
processes
\begin{align*}
\mathbb Y_{N,a}(t)
&=
\frac{1}{N}\sum_{i=1}^N A_{ia}Y_i^{(a)}(t),\\
\mathbb F_{N,1a}(t)
&=
\frac{1}{N}\sum_{i=1}^N A_{ia}N_i^{(a)}(t).
\end{align*}
Their population limits are $\pi_a y_a(t)$ and $\pi_a F_{1a}(t)$,
respectively. The no-ties condition in Assumption~\ref{asu:censoring} ensures that the
product-limit functional identifies the marginal event-time survival
function even when either marginal distribution has jumps.

To see the first assertion, consider a generic bounded process
$f_t(\mathcal O_i^\star)$ and let
\[
\bar f_t(s)=\E\{f_t(\mathcal O_i^\star)\mid Z_i=s\}.
\]
The following decomposition is exact:
\begin{align}
\label{eq:s-generic-decomposition}
\frac{1}{N}\sum_{i=1}^N A_{ia}f_t(\mathcal O_i^\star)
-\pi_a\E\{f_t(\mathcal O_i^\star)\}
={}&
\frac{1}{N}\sum_{i=1}^N
A_{ia}\{f_t(\mathcal O_i^\star)-\bar f_t(Z_i)\}
\notag\\
&+
\pi_a
\left[
\frac{1}{N}\sum_{i=1}^N\bar f_t(Z_i)
-\E\{f_t(\mathcal O_i^\star)\}
\right]\notag\\
&+
\frac{1}{N}
\sum_{s=1}^J D_{N,a}(s)\bar f_t(s).
\end{align}
Conditional on the strata and assignments, the first term on the
right-hand side has mean zero. The second term converges uniformly to zero by
the ordinary uniform law of large numbers. Since $J$ is fixed and
$D_{N,a}(s)=O_p(1)$, the final term is $o_p(1)$ uniformly in $t$.
Applying \eqref{eq:s-generic-decomposition} to the indicator classes defining
$Y_i^{(a)}(t)$ and $N_i^{(a)}(t)$ gives
\begin{align}
\label{eq:s-uniform-risk}
\sup_{0\leq t\leq\tau_U}
|\mathbb Y_{N,a}(t)-\pi_a y_a(t)|
&=o_p(1),\\
\label{eq:s-uniform-event}
\sup_{0\leq t\leq\tau_U}
|\mathbb F_{N,1a}(t)-\pi_a F_{1a}(t)|
&=o_p(1).
\end{align}

Uniform consistency alone does not justify the functional delta method, because the
input processes must also be asymptotically tight at the $\sqrt N$ scale.
This follows from a treatment--stratum cell decomposition. Write
$F_{1a,s}(t)=\E\{N_i^{(a)}(t)\mid Z_i=s\}$ and decompose
\begin{align}
\label{eq:s-cell-decomposition}
\sqrt N\{\mathbb F_{N,1a}(t)-\pi_aF_{1a}(t)\}
={}&
\sum_{s=1}^J
\sqrt{\frac{n_a(s)}{N}}
\left[
\frac{1}{\sqrt{n_a(s)}}
\sum_{i=1}^N
A_{ia}\ind(Z_i=s)
\{N_i^{(a)}(t)-F_{1a,s}(t)\}
\right]\notag\\
&+
\sum_{s=1}^J
F_{1a,s}(t)\,
\frac{n_a(s)-N\pi_ap_s}{\sqrt N}.
\end{align}
By Assumption~\ref{asu:randomization}, conditional on the complete stratum and assignment
sequences, the subjects in cell $(a,s)$ carry independent copies of the
conditional law of $\{X_i(a),\Delta_i(a)\}$ given $Z_i=s$, and
$n_a(s)/N\to\pi_ap_s>0$ in probability. The classes of indicators
$\{\ind(x\leq t,\delta=1):t\in[0,\tau_U]\}$ and
$\{\ind(x\geq t):t\in[0,\tau_U]\}$ are Vapnik--Chervonenkis classes, hence
Donsker for the cell-specific sampling distribution
\citep[Section~2.5]{vaart1996weak}. Each bracketed cell process in
\eqref{eq:s-cell-decomposition} is the process $\mathbb H_{N,as}$ of
Lemma~\ref{lem:cell-processes} indexed by this class. The lemma applies
because $n_a(s)/N\to\pi_ap_s>0$ in probability, so these processes converge
jointly and stably to independent Brownian bridges, and
$\sqrt{n_a(s)/N}\to\sqrt{\pi_ap_s}$ in probability. For the
second term in \eqref{eq:s-cell-decomposition},
\[
\frac{n_a(s)-N\pi_ap_s}{\sqrt N}
=
\frac{D_{N,a}(s)}{\sqrt N}
+
\pi_a\frac{n(s)-Np_s}{\sqrt N}
=O_p(1)
\]
by Assumption~\ref{asu:bounded-balance} and the multinomial central limit theorem, and
$F_{1a,s}(\cdot)$ is a fixed monotone bounded function. Since $J$ is fixed,
$\sqrt N(\mathbb F_{N,1a}-\pi_aF_{1a})$ is asymptotically tight in
$\ell^\infty[0,\tau_U]$, and the identical argument applies to
$\sqrt N(\mathbb Y_{N,a}-\pi_ay_a)$.

The Nelson--Aalen estimator can be written as
\[
\widehat\Lambda_a(t)
=
\int_{(0,t]}
\frac{d\mathbb F_{N,1a}(u)}
{\mathbb Y_{N,a}(u)},
\]
on the event $\{\inf_{t\leq\tau_U}\mathbb Y_{N,a}(t)\geq\pi_a\epsilon/2\}$,
whose probability tends to one by \eqref{eq:s-uniform-risk} and
Assumption~\ref{asu:censoring}. The map
\[
(H,F)\longmapsto
\left[t\longmapsto\int_{(0,t]}\frac{dF(u)}{H(u)}\right]
\]
is Hadamard differentiable at
$(\pi_a y_a,\pi_a F_{1a})$ on the domain of pairs $(H,F)$ with
$\inf_{t\leq\tau_U}H(t)\geq\pi_a\epsilon/2$ and $F$ of total variation at
most one, tangentially to $D[0,\tau_U]^2$, where the integral against a
perturbation of $F$ that is not of bounded variation is defined by
integration by parts \citep[Section~3.9]{vaart1996weak}. Both the
population pair and, on an event with probability tending to one, the
empirical pair $(\mathbb Y_{N,a},\mathbb F_{N,1a})$ lie in this domain by
Assumption~\ref{asu:censoring} and \eqref{eq:s-uniform-risk}. The decomposition
\eqref{eq:s-cell-decomposition} in fact gives weak convergence, not only
tightness, of the $\sqrt N$-scaled input processes. The cell processes
converge stably by Lemma~\ref{lem:cell-processes}. The composition term is
measurable with respect to $\mathcal Z_N$ and converges by the multinomial
central limit theorem, and the imbalance term is $o_p(1)$ under
Assumption~\ref{asu:bounded-balance}. Stability lets the limits be combined: the cell limits are
independent of the composition limit, so the sum converges in
distribution. The functional delta method
\citep[Theorem~3.9.4]{vaart1996weak} therefore yields the uniform
expansion
\begin{equation}
\label{eq:s-na-expansion}
\sqrt N\{\widehat\Lambda_a(t)-\Lambda_a(t)\}
=
\frac{1}{\sqrt N}
\sum_{i=1}^N
\frac{A_{ia}}{\pi_a}
\int_{(0,t]}
\frac{dM_i^{(a)}(u)}{y_a(u)}
+o_p(1),
\end{equation}
where the remainder is uniform in $t\in[0,\tau_U]$.

For completeness, the summand in \eqref{eq:s-na-expansion} follows by
combining the two derivative terms:
\begin{align*}
&\frac{A_{ia}}{\pi_a}
\left\{
\int_{(0,t]}\frac{dN_i^{(a)}(u)}{y_a(u)}
-
\int_{(0,t]}
\frac{Y_i^{(a)}(u)}{y_a(u)}\,d\Lambda_a(u)
\right\}\\
&\hspace{35mm}
=
\frac{A_{ia}}{\pi_a}
\int_{(0,t]}\frac{dM_i^{(a)}(u)}{y_a(u)}.
\end{align*}

Next consider the product-integral map
\[
\Lambda\longmapsto
\left[t\longmapsto\prod_{(0,t]}\{1-d\Lambda(u)\}\right].
\]
At $\Lambda_a$, its derivative in direction $h$ is
\begin{equation}
\label{eq:s-product-derivative}
\dot{\mathcal P}_{\Lambda_a}(h)(t)
=
-S_a(t)
\int_{(0,t]}
\frac{dh(u)}{1-\Delta\Lambda_a(u)}.
\end{equation}
Hadamard differentiability of the product-integral map on cumulative
hazards with uniformly bounded variation and jumps bounded away from one is
established by \citet{gill1990survey}, with related results in
\citet[Section~3.9]{vaart1996weak} and
\citet[Chapter~II]{andersen1993statistical}. The jump-positivity condition
in Assumption~\ref{asu:censoring} places $\Lambda_a$ in this domain, so the derivative is
well-defined and continuous. Combining
\eqref{eq:s-na-expansion} and \eqref{eq:s-product-derivative} gives
\begin{equation}
\label{eq:s-survival-expansion}
\sup_{0\leq t\leq\tau_U}
\left|
\sqrt N\{\widehat S_a(t)-S_a(t)\}
-
\frac{1}{\sqrt N}\sum_{i=1}^N
\frac{A_{ia}}{\pi_a}\psi_{ia}(t)
\right|
=o_p(1).
\end{equation}

Finally, the integration map
\[
f\longmapsto
\left[\tau\longmapsto\int_0^\tau f(t)\,dt\right]
\]
is continuous from $\ell^\infty[0,\tau_U]$ into
$\ell^\infty(\calT)$. Integrating \eqref{eq:s-survival-expansion} and using
\eqref{eq:s-rmst-if} proves \eqref{eq:s-ual}. Since $K$ is fixed, the
expansions hold jointly over all arms.
\end{proof}

\section{CAR decomposition and functional central limit theorem}
\label{sec:supp-car-clt}

\begin{lemma}[Joint limit of the CAR-weighted influence processes]
\label{lem:car-clt}
Define
\[
\mathbb U_{N,a}(\tau)
=
\frac{1}{\sqrt N}
\sum_{i=1}^N\frac{A_{ia}}{\pi_a}\phi_{ia}(\tau),
\qquad
a=0,\ldots,K.
\]
Under Assumptions~\ref{asu:sampling}--\ref{asu:censoring} of the main paper,
\[
\{\mathbb U_{N,0},\ldots,\mathbb U_{N,K}\}
\rights
\{\mathbb G_{\mu,0},\ldots,\mathbb G_{\mu,K}\}
\]
in $\ell^\infty(\calT)^{K+1}$, where the limit is a mean-zero Gaussian
process with covariance kernel
\begin{equation}
\label{eq:s-mu-covariance}
\Cov\{\mathbb G_{\mu,a}(\tau),\mathbb G_{\mu,b}(\tau')\}
=
\E\{m_a(Z_i,\tau)m_b(Z_i,\tau')\}
+
\ind(a=b)\frac{1}{\pi_a}
\E\{c_a(Z_i;\tau,\tau')\}.
\end{equation}
\end{lemma}

\begin{proof}
The argument uses the decomposition into within-stratum residual, stratum
composition, and imbalance terms that is standard in the CAR literature
\citep{ma2015testing,ye2022inference,wang2023model}. For every arm $a$, add
and subtract the conditional mean $m_a(Z_i,\tau)$. This gives the exact
decomposition
\begin{equation}
\label{eq:s-three-part}
\mathbb U_{N,a}(\tau)
=
R_{N,a}(\tau)+B_{N,a}(\tau)+I_{N,a}(\tau),
\end{equation}
where
\begin{align}
\label{eq:s-residual-part}
R_{N,a}(\tau)
&=
\frac{1}{\sqrt N}
\sum_{i=1}^N
\frac{A_{ia}}{\pi_a}
\{\phi_{ia}(\tau)-m_a(Z_i,\tau)\},\\
\label{eq:s-stratum-part}
B_{N,a}(\tau)
&=
\frac{1}{\sqrt N}
\sum_{i=1}^N m_a(Z_i,\tau),\\
\label{eq:s-imbalance-part}
I_{N,a}(\tau)
&=
\frac{1}{\pi_a\sqrt N}
\sum_{s=1}^J D_{N,a}(s)m_a(s,\tau).
\end{align}
Indeed,
\[
\frac{A_{ia}}{\pi_a}m_a(Z_i,\tau)
=
m_a(Z_i,\tau)
+
\frac{A_{ia}-\pi_a}{\pi_a}m_a(Z_i,\tau),
\]
and summing the second term over subjects within stratum $s$ produces
\eqref{eq:s-imbalance-part}.

We first show that the imbalance component is negligible. Because $J$ is
fixed, the functions $m_a(s,\cdot)$ are bounded on $\calT$, and
$D_{N,a}(s)=O_p(1)$,
\begin{align}
\label{eq:s-imbalance-negligible}
\sup_{\tau\in\calT}|I_{N,a}(\tau)|
&\leq
\frac{1}{\pi_a\sqrt N}
\sum_{s=1}^J
|D_{N,a}(s)|
\sup_{\tau\in\calT}|m_a(s,\tau)|
=o_p(1).
\end{align}

Next consider the residual component. Conditional on
\[
\mathcal Z_N=\sigma(Z_1,\ldots,Z_N)
\quad\text{and}\quad
\mathcal A_N=\sigma(A_1,\ldots,A_N),
\]
the summands in \eqref{eq:s-residual-part} are independent across subjects,
have conditional mean zero, and only one arm-specific summand can be nonzero
for each subject. For any $a,b$ and $\tau,\tau'\in\calT$, their conditional
covariance is zero when $a\neq b$. When $a=b$, it is
\begin{align}
\label{eq:s-conditional-residual-cov}
&\Cov\{R_{N,a}(\tau),R_{N,a}(\tau')
\mid\mathcal Z_N,\mathcal A_N\}\notag\\
&\quad=
\frac{1}{N}
\sum_{i=1}^N
\frac{A_{ia}}{\pi_a^2}
c_a(Z_i;\tau,\tau')\notag\\
&\quad=
\frac{1}{N\pi_a^2}
\sum_{s=1}^J
\{\pi_a n(s)+D_{N,a}(s)\}
c_a(s;\tau,\tau')\notag\\
&\quad\xrightarrow{P}
\frac{1}{\pi_a}
\E\{c_a(Z_i;\tau,\tau')\}.
\end{align}
By Lemma~\ref{lem:if-bounds}, every summand is bounded by
$2B\tau_U/\pi_a$ uniformly in $\tau$, so the conditional Lindeberg condition for
this triangular array holds trivially. The Lindeberg--Feller central limit
theorem therefore gives finite-dimensional convergence of
$(R_{N,0},\ldots,R_{N,K})$, conditionally on
$(\mathcal Z_N,\mathcal A_N)$, to a mean-zero Gaussian vector with the
nonrandom covariance
\[
\ind(a=b)\frac{1}{\pi_a}
\E\{c_a(Z_i;\tau,\tau')\}.
\]
Asymptotic equicontinuity follows from the bracketing central limit
theorem for independent, non-identically distributed processes
\citep[Theorem~2.11.9]{vaart1996weak}, of which the Jain--Marcus theorem
for Lipschitz processes is a special case. By Lemma~\ref{lem:if-bounds},
each summand is Lipschitz on $\calT$ with the nonrandom constant
$L=2B/\pi_a$. For a grid $\tau_1<\cdots<\tau_m$ of mesh $\delta$, the
brackets
$[\,Z_{Ni}(\tau_j)-N^{-1/2}L\delta,\ Z_{Ni}(\tau_j)+N^{-1/2}L\delta\,]$,
with $Z_{Ni}$ the $i$th scaled summand, have
$\sum_i\E(\text{width}^2)=4L^2\delta^2$, so the bracketing number at
level $\varepsilon$ is of order $1/\varepsilon$ and the bracketing entropy
integral is finite. The Lindeberg condition holds because the summands are
uniformly bounded. We refer to this verification below as the
Lipschitz-bracketing argument. Hence the conditional convergence holds in
$\ell^\infty(\calT)^{K+1}$.

The stratum-composition term has a particularly simple form. By
\eqref{eq:s-m-centering},
\begin{align}
\label{eq:s-stratum-count-form}
B_{N,a}(\tau)
&=
\sum_{s=1}^J
\frac{n(s)-Np_s}{\sqrt N}\,m_a(s,\tau).
\end{align}
The vector of centered stratum counts obeys the multinomial central limit
theorem. Because $J$ is finite, \eqref{eq:s-stratum-count-form} converges in
$\ell^\infty(\calT)^{K+1}$ to a mean-zero Gaussian process with covariance
\begin{align}
\label{eq:s-stratum-cov}
\lim_{N\to\infty}
\Cov\{B_{N,a}(\tau),B_{N,b}(\tau')\}
&=
\sum_{s=1}^Jp_s m_a(s,\tau)m_b(s,\tau')\notag\\
&=
\E\{m_a(Z_i,\tau)m_b(Z_i,\tau')\}.
\end{align}
The usual subtraction of the product of the means vanishes because of
\eqref{eq:s-m-centering}.

It remains to combine the two non-negligible processes. The term
$B_N=(B_{N,0},\ldots,B_{N,K})$ is measurable with respect to
$\mathcal Z_N$. Conditional on $(\mathcal Z_N,\mathcal A_N)$, the
characteristic function of every finite-dimensional projection of $R_N$
converges in probability to that of a Gaussian vector whose covariance is
the nonrandom limit in \eqref{eq:s-conditional-residual-cov}. Consequently,
$R_N$ converges stably and is asymptotically independent of $B_N$. Equivalently,
for bounded continuous functions $f$ and $g$ of finite-dimensional
projections,
\begin{align*}
\E\{f(R_N)g(B_N)\}
&=
\E\!\left[
g(B_N)
\E\{f(R_N)\mid\mathcal Z_N,\mathcal A_N\}
\right]\\
&\longrightarrow
\E\{f(\mathbb G_R)\}\E\{g(\mathbb G_B)\}.
\end{align*}
Adding the covariance kernels in
\eqref{eq:s-conditional-residual-cov} and \eqref{eq:s-stratum-cov}, and using
\eqref{eq:s-imbalance-negligible}, proves
\eqref{eq:s-mu-covariance}.
\end{proof}

\section[\texorpdfstring{Proof of Theorem~\ref*{thm:bounded-balance-fclt}}{Proof of Theorem 1}]{Proof of \texorpdfstring{Theorem~\ref{thm:bounded-balance-fclt}}{Theorem 1}}
\label{sec:supp-main-proof}

\begin{proof}[Proof of Theorem~\ref{thm:bounded-balance-fclt} in the main paper]
Let
\[
\widehat\bmu(\tau)
=
\{\widehat\mu_0(\tau),\ldots,\widehat\mu_K(\tau)\}^{T}
\]
and define $\bmu(\tau)$ analogously. By Lemma~\ref{lem:ual},
\begin{equation}
\label{eq:s-vector-ual}
\sqrt N\{\widehat\bmu(\cdot)-\bmu(\cdot)\}
=
\begin{pmatrix}
\mathbb U_{N,0}(\cdot)\\
\vdots\\
\mathbb U_{N,K}(\cdot)
\end{pmatrix}
+o_p(1)
\end{equation}
in $\ell^\infty(\calT)^{K+1}$. Lemma~\ref{lem:car-clt} and Slutsky's
theorem therefore imply
\[
\sqrt N\{\widehat\bmu(\cdot)-\bmu(\cdot)\}
\rights
\mathbb G_\mu(\cdot),
\]
where
\[
\Cov\{\mathbb G_{\mu,a}(\tau),\mathbb G_{\mu,b}(\tau')\}
=
\E\{m_a(Z_i,\tau)m_b(Z_i,\tau')\}
+
\ind(a=b)\frac{1}{\pi_a}
\E\{c_a(Z_i;\tau,\tau')\}.
\]
This proves the first assertion.

For the treatment-control contrasts, let
\[
\bm C=[-\bone_K,\bm I_K].
\]
Because multiplication by the fixed matrix $\bm C$ is a continuous linear
map from $\ell^\infty(\calT)^{K+1}$ to
$\ell^\infty(\calT)^K$, the continuous mapping theorem gives
\[
\sqrt N\{\widehat\btheta(\cdot)-\btheta(\cdot)\}
=
\bm C\sqrt N\{\widehat\bmu(\cdot)-\bmu(\cdot)\}
\rights
\bm C\mathbb G_\mu(\cdot)
=
\mathbb G_\theta(\cdot).
\]

For $k,\ell\in\{1,\ldots,K\}$, direct expansion gives
\begin{align*}
&\Cov\{\mathbb G_{\theta,k}(\tau),
\mathbb G_{\theta,\ell}(\tau')\}\\
&\quad=
\Cov\{\mathbb G_{\mu,k}(\tau)-\mathbb G_{\mu,0}(\tau),
\mathbb G_{\mu,\ell}(\tau')-\mathbb G_{\mu,0}(\tau')\}\\
&\quad=
\E\!\left[
\{m_k(Z_i,\tau)-m_0(Z_i,\tau)\}
\{m_\ell(Z_i,\tau')-m_0(Z_i,\tau')\}
\right]\\
&\qquad+
\ind(k=\ell)\frac{1}{\pi_k}
\E\{c_k(Z_i;\tau,\tau')\}
+
\frac{1}{\pi_0}
\E\{c_0(Z_i;\tau,\tau')\}.
\end{align*}
The final term is present for every $(k,\ell)$ because every contrast
contains the same estimated control-arm RMST.

Evaluation at a fixed $\tau\in\calT$ is a continuous map on
$\ell^\infty(\calT)^K$. A final application of the continuous mapping
theorem yields
\[
\sqrt N\{\widehat\btheta(\tau)-\btheta(\tau)\}
\xrightarrow{D}
\mathcal N_K\{\bm 0,\bm\Sigma_\theta(\tau)\},
\]
where the $(k,\ell)$ entry of $\bm\Sigma_\theta(\tau)$ is the preceding
kernel evaluated at $(\tau,\tau)$. This completes the proof.
\end{proof}

\section[\texorpdfstring{Proof of Theorem~\ref*{thm:general-car-fclt}}{Proof of Theorem 2}]{Proof of \texorpdfstring{Theorem~\ref{thm:general-car-fclt}}{Theorem 2}}
\label{sec:supp-general-car}

Let
\[
\bm d_N
=
\frac{1}{\sqrt N}
\operatorname{vec}
\{D_{N,a}(s):a=0,\ldots,K;\ s=1,\ldots,J\}
\]
and let $\bm Q(\tau)$ be the $(K+1)\times (K+1)J$ matrix whose entries are
\[
Q_{a,(b,s)}(\tau)
=
\frac{\ind(a=b)}{\pi_a}m_a(s,\tau).
\]

\begin{lemma}[RMST expansion under $\sqrt N$-order imbalance]
\label{lem:ual-general}
Suppose Assumptions~\ref{asu:sampling}, \ref{asu:randomization}, \ref{asu:censoring}, and~\ref{asu:general-imbalance} of the main paper hold. Then
\begin{equation}
\label{eq:s-ual-general}
\sup_{\tau\in\calT}
\left|
\sqrt N\{\widehat\mu_a(\tau)-\mu_a(\tau)\}
-
\frac{1}{\sqrt N}
\sum_{i=1}^N\frac{A_{ia}}{\pi_a}\phi_{ia}(\tau)
\right|
=o_p(1)
\end{equation}
for every $a=0,\ldots,K$, jointly over arms.
\end{lemma}

\begin{proof}
The proof follows the differentiability argument in
Lemma~\ref{lem:ual}. We only need to verify that replacing bounded
within-stratum imbalance by Assumption~\ref{asu:general-imbalance} does not alter the first-order
linearization.

For a generic bounded process $f_t$ and
$\bar f_t(s)=\E(f_t\mid Z_i=s)$, decomposition
\eqref{eq:s-generic-decomposition} remains exact. Its imbalance component
now satisfies
\[
\sup_t
\left|
\frac{1}{N}
\sum_{s=1}^JD_{N,a}(s)\bar f_t(s)
\right|
=
O_p(N^{-1/2})
=o_p(1),
\]
because $J$ is fixed and
$\max_{a,s}|D_{N,a}(s)|=O_p(\sqrt N)$. Hence the arm-specific empirical
at-risk and failure processes still converge uniformly to
$\pi_a y_a$ and $\pi_a F_{1a}$.

At the $\sqrt N$ scale, the weighted empirical process admits the exact
conditional-mean decomposition
\begin{align*}
&\frac{1}{\sqrt N}
\sum_{i=1}^N
\frac{A_{ia}}{\pi_a}
\{f_t(\mathcal O_i^\star)-\E f_t\}\\
&\quad=
\frac{1}{\sqrt N}
\sum_{i=1}^N
\frac{A_{ia}}{\pi_a}
\{f_t(\mathcal O_i^\star)-\bar f_t(Z_i)\}\\
&\qquad+
\frac{1}{\sqrt N}
\sum_{i=1}^N
\{\bar f_t(Z_i)-\E f_t\}
+
\frac{1}{\pi_a}
\sum_{s=1}^J
\frac{D_{N,a}(s)}{\sqrt N}\{\bar f_t(s)-\E f_t\}.
\end{align*}
The first term converges by the treatment--stratum cell decomposition
\eqref{eq:s-cell-decomposition} and Lemma~\ref{lem:cell-processes}, which
requires only $n_a(s)/n(s)\to\pi_a$ and therefore holds under Assumption~\ref{asu:general-imbalance}.
Because $\sum_sp_s\bar f_t(s)=\E f_t$, the second term equals
$\sum_s\bar f_t(s)W_{N,s}$, a linear function of $\bm W_N$, and the third
is a linear function of $\bm d_N$. By Assumption~\ref{asu:general-imbalance} these two terms
converge jointly, and because they are measurable with respect to
$(\mathcal Z_N,\mathcal A_N)$, the stability in
Lemma~\ref{lem:cell-processes} gives joint weak convergence of all three
terms. Thus the input empirical processes converge weakly at the root-$N$
scale, as the functional delta method requires.
Applying the Hadamard derivatives of the Nelson--Aalen, product-integral,
and integration maps exactly as in Lemma~\ref{lem:ual} yields
\eqref{eq:s-ual-general}.
\end{proof}

\begin{proof}[Proof of Theorem~\ref{thm:general-car-fclt} in the main paper]
By Lemma~\ref{lem:ual-general},
\[
\sqrt N\{\widehat\bmu(\cdot)-\bmu(\cdot)\}
=
\mathbb U_N(\cdot)+o_p(1),
\]
where the $a$th component of $\mathbb U_N$ is
\[
\mathbb U_{N,a}(\tau)
=
\frac{1}{\sqrt N}
\sum_{i=1}^N\frac{A_{ia}}{\pi_a}\phi_{ia}(\tau).
\]
The exact decomposition \eqref{eq:s-three-part} continues to hold. In
matrix form, it is
\begin{equation}
\label{eq:s-general-decomposition}
\mathbb U_N(\tau)
=
\bm R_N(\tau)+\bm B_N(\tau)+\bm Q(\tau)\bm d_N,
\end{equation}
where the $a$th components of $\bm R_N$ and $\bm B_N$ are given by
\eqref{eq:s-residual-part} and \eqref{eq:s-stratum-part}.

Because $n_a(s)/n(s)\to\pi_a$, the conditional covariance calculation
\eqref{eq:s-conditional-residual-cov} remains valid. Hence
$\bm R_N$ converges stably to a mean-zero Gaussian process
$\mathbb G_R$ with covariance
\begin{equation}
\label{eq:s-general-residual-kernel}
\Cov\{\mathbb G_{R,a}(\tau),\mathbb G_{R,b}(\tau')\}
=
\ind(a=b)\frac{1}{\pi_a}
\E\{c_a(Z_i;\tau,\tau')\}.
\end{equation}
Stable convergence implies that $\mathbb G_R$ is independent of every
weak limit measurable with respect to the strata and treatment assignments.

The stratum-composition process is the continuous linear transformation
\[
B_{N,a}(\tau)
=
\sum_{s=1}^JW_{N,s}m_a(s,\tau)
\]
of the vector $\bm W_N$. Assumption~\ref{asu:general-imbalance} therefore gives the joint convergence
\[
\{\bm B_N(\cdot),\bm d_N\}
\rights
\{\mathbb G_B(\cdot),\bm D\},
\]
where
\[
\bm D\sim\mathcal N_L(\bm 0,\bm\Omega_D),
\]
$\mathbb G_B$ and $\bm D$ are independent, and
\begin{equation}
\label{eq:s-general-stratum-kernel}
\Cov\{\mathbb G_{B,a}(\tau),\mathbb G_{B,b}(\tau')\}
=
\E\{m_a(Z_i,\tau)m_b(Z_i,\tau')\}.
\end{equation}
The independence follows from the block-diagonal covariance in Assumption~\ref{asu:general-imbalance}
and joint Gaussianity.

The functions $m_a(s,\cdot)$ are continuous on $\calT$, so
$\bm Q(\cdot)\bm d_N\rights\bm Q(\cdot)\bm D$ in
$\ell^\infty(\calT)^{K+1}$. Combining the preceding limits in
\eqref{eq:s-general-decomposition} gives
\[
\mathbb G_\mu^D(\tau)
=
\mathbb G_R(\tau)+\mathbb G_B(\tau)+\bm Q(\tau)\bm D.
\]
Its covariance kernel is
\[
\bm\Gamma_\mu^D(\tau,\tau')
=
\bm\Gamma_\mu(\tau,\tau')
+
\bm Q(\tau)\bm\Omega_D\bm Q(\tau')^T.
\]

Finally, multiplication by
$\bm C=[-\bone_K,\bm I_K]$ is a continuous linear map. Therefore,
\[
\sqrt N\{\widehat\btheta(\cdot)-\btheta(\cdot)\}
\rights
\bm C\mathbb G_\mu^D(\cdot)
\]
with covariance
\[
\bm C\bm\Gamma_\mu(\tau,\tau')\bm C^T
+
\bm C\bm Q(\tau)\bm\Omega_D
\bm Q(\tau')^T\bm C^T.
\]
Writing $\bm R(\tau)=\bm C\bm Q(\tau)$ gives the stated result. Its
$k$th row is
\[
R_{k,(a,s)}(\tau)
=
\frac{\ind(a=k)}{\pi_k}m_k(s,\tau)
-
\frac{\ind(a=0)}{\pi_0}m_0(s,\tau).
\]
No equality among the treatment-specific survival functions was used, so the
result holds under arbitrary treatment effects.
\end{proof}

\section[\texorpdfstring{Proof of Proposition~\ref*{prop:additivity}}{Proof of Proposition 1}]{Proof of \texorpdfstring{Proposition~\ref{prop:additivity}}{Proposition 1}}
\label{sec:supp-additivity}

\begin{proof}[Proof of Proposition~\ref{prop:additivity} in the main paper]
Fix $\tau\in\calT$ and write
$\bm Q(\tau)=\bm Q^{\mathrm{add}}(\tau)+\bm Q^{\eta}(\tau)$, where
$\bm Q^{\mathrm{add}}(\tau)$ is defined by \eqref{eq:s-Q-def-add} below. For
arm $a$, let
\begin{equation}
\label{eq:s-Q-def-add}
Q^{\mathrm{add}}_{a,(b,s)}(\tau)
=
\frac{\ind(a=b)}{\pi_a}m_a^{\mathrm{add}}(s,\tau),
\qquad
m_a^{\mathrm{add}}(s,\tau)=\sum_{r=1}^R f_{a,r}(s_r,\tau).
\end{equation}
Then the $a$th coordinate of $\bm Q^{\mathrm{add}}(\tau)\bm d_N$ is
\[
\frac{1}{\pi_a\sqrt N}
\sum_{s\in\mathcal S}D_{N,a}(s)\sum_{r=1}^Rf_{a,r}(s_r,\tau)
=
\frac{1}{\pi_a\sqrt N}
\sum_{r=1}^R\sum_{l}f_{a,r}(l,\tau)D_{N,a}^{(r)}(l),
\]
because grouping the strata by the level of factor $r$ turns the sum over
$s$ into a sum of marginal imbalances. The numbers of arms, factors, and
levels are finite and every $D_{N,a}^{(r)}(l)=O_p(1)$, so
$\bm Q^{\mathrm{add}}(\tau)\bm d_N=O_p(N^{-1/2})\to0$ in probability.

By Assumption~\ref{asu:general-imbalance}, $\bm d_N\xrightarrow{D}\bm D\sim\mathcal N_L(\bm0,\bm\Omega_D)$,
so the continuous mapping theorem gives
$\bm Q^{\mathrm{add}}(\tau)\bm d_N\xrightarrow{D}\bm Q^{\mathrm{add}}(\tau)\bm D$.
Limits in distribution are unique, so $\bm Q^{\mathrm{add}}(\tau)\bm D=\bm 0$
almost surely, and therefore
$\bm Q^{\mathrm{add}}(\tau)\bm\Omega_D\bm Q^{\mathrm{add}}(\tau)^T=\bm0$.
Because $\bm\Omega_D$ is positive semidefinite, write
$\bm\Omega_D=\bm\Omega_D^{1/2}\bm\Omega_D^{1/2}$. The previous display
then implies $\bm Q^{\mathrm{add}}(\tau)\bm\Omega_D^{1/2}=\bm0$ and hence
$\bm Q^{\mathrm{add}}(\tau)\bm\Omega_D=\bm0$. As this holds for every
$\tau\in\calT$, expanding
$\bm Q(\tau)\bm\Omega_D\bm Q(\tau')^T$ with
$\bm Q=\bm Q^{\mathrm{add}}+\bm Q^{\eta}$ leaves only
$\bm Q^{\eta}(\tau)\bm\Omega_D\bm Q^{\eta}(\tau')^T$, which is the
identity stated in Proposition~\ref{prop:additivity}. If every $m_a(\cdot,\tau)$ is additive, one may
take $\eta_a\equiv0$, so the design term vanishes and the kernel of
Theorem~\ref{thm:general-car-fclt} reduces to $\bm\Gamma_\mu$.
\end{proof}

\begin{remark}
The decomposition in Proposition~\ref{prop:additivity} is not unique, but the conclusion does
not depend on the choice: any two admissible remainders differ by an
additive function, whose loading annihilates $\bm\Omega_D$. A convenient
choice is the $\bm p$-weighted least-squares projection of
$m_a(\cdot,\tau)$ onto additive functions, so that $\eta_a$ is
$\bm p$-weighted orthogonal to all additive functions. With two binary factors, $\eta_a(\cdot,\tau)$ is then
proportional to the interaction contrast of $m_a(\cdot,\tau)$ over the four
joint strata.
\end{remark}

\section[\texorpdfstring{Proof of Proposition~\ref*{prop:nondegenerate}}{Proof of Proposition 2}]{Proof of \texorpdfstring{Proposition~\ref{prop:nondegenerate}}{Proposition 2}}
\label{sec:supp-nondegenerate}

Throughout, $\bm D_n\in\R^L$ stacks the within-stratum imbalances
$\{D_{n,a}(s)\}$ of the first $n$ subjects in the arm-major order of
$\bm d_N$, so that $\bm D_0=\bm 0$ and $\bm M_0=\bm 0$. A subject in
stratum $s$ assigned to arm $a$ changes $\bm D_n$ by $\bm g(s,a)$, where
$g(s,a)_{(b,s')}=\ind(s'=s)\{\ind(b=a)-\pi_b\}$. The vector $\bm M_n$ is a
fixed linear function of $\bm D_n$, obtained by summing over the strata
that share a level of each factor, and $\bm M_n=\bm 0$ is the balanced
state. At time $\varrho$ every marginal count satisfies
$n^{(r)}_{\varrho,a}(l)=\pi_an^{(r)}_{\varrho}(l)$, so condition~(ii)
forces every $\pi_a$ to be rational.

\begin{proof}[Proof of Proposition~\ref{prop:nondegenerate} in the main paper]
\textit{Step 1: regeneration.} By condition~(i) and
Assumption~\ref{asu:sampling}, $(Z_n,A_n)$ is drawn from a distribution
that depends on the past only through $\bm M_{n-1}$, and $\bm M_n$ is
determined by $\bm M_{n-1}$ and $(Z_n,A_n)$. Hence $\{\bm M_n\}$ is a
time-homogeneous Markov chain on a countable set, and
$\{(\bm M_n,Z_n,A_n)\}$ is a Markov chain whose transitions depend only on
$\bm M_n$. Let $0=\vartheta_0<\vartheta_1<\cdots$ be the successive times
at which $\bm M_n=\bm 0$, so that $\vartheta_1=\varrho$. By condition~(ii)
every $\vartheta_k$ is finite almost surely, and by the strong Markov
property the blocks $\{(Z_n,A_n):\vartheta_{k-1}<n\leq\vartheta_k\}$,
$k\geq1$, are independent and identically distributed. Let
$\varrho_k=\vartheta_k-\vartheta_{k-1}$ and
$\bm\zeta_k=\bm D_{\vartheta_k}-\bm D_{\vartheta_{k-1}}$, so that
$(\varrho_1,\bm\zeta_1)=(\varrho,\bm\zeta)$. Each $\bm\zeta_k$ lies in
$\mathcal V$: its arm-wise sums vanish because those of every $\bm D_n$
do, and its marginal sums vanish because
$\bm M_{\vartheta_k}=\bm M_{\vartheta_{k-1}}=\bm 0$.

\textit{Step 2: the covariance formula.} Fix $\bm v\in\R^L$. Then
$\bm v^T\bm D_N=\sum_{n=1}^N\bm v^T\bm g(Z_n,A_n)$ with
$|\bm v^T\bm g(s,a)|\leq2\|\bm v\|_\infty$, so
$|\bm v^T\bm\zeta_k|\leq2\|\bm v\|_\infty\varrho_k$ and
$\E\{(\bm v^T\bm\zeta)^2\}<\infty$ by condition~(ii). Let
$k_N=\max\{k:\vartheta_k\leq N\}$, which is at most $N$. Then
\[
\bm v^T\bm D_N=\sum_{k=1}^{k_N}\bm v^T\bm\zeta_k+\mathcal R_N,
\qquad
|\mathcal R_N|\leq2\|\bm v\|_\infty\,\varrho_{k_N+1}
\leq2\|\bm v\|_\infty\max_{k\leq N+1}\varrho_k .
\]
By the renewal theorem, $k_N/N\to1/\E(\varrho)$ almost surely. Because the
$\varrho_k$ are independent and integrable, $\varrho_k/k\to0$ almost
surely, so $\mathcal R_N/N\to0$ almost surely, and the strong law of large
numbers gives $N^{-1}\bm v^T\bm D_N\to\E(\bm v^T\bm\zeta)/\E(\varrho)$
almost surely. Assumption~\ref{asu:general-imbalance} implies
$N^{-1}\bm D_N\to\bm 0$ in probability, so $\E(\bm\zeta)=\bm 0$. The
central limit theorem with a random number of summands
\citep{anscombe1952large} then gives
$N^{-1/2}\sum_{k\leq k_N}\bm v^T\bm\zeta_k\xrightarrow{D}
\mathcal N[0,\E\{(\bm v^T\bm\zeta)^2\}/\E(\varrho)]$. For every $x>0$,
\[
\Pp\Bigl(\max_{k\leq N+1}\varrho_k>x\sqrt N\Bigr)
\leq(N+1)\Pp(\varrho>x\sqrt N)
\leq\frac{N+1}{x^2N}\E\{\varrho^2\ind(\varrho>x\sqrt N)\}\to0,
\]
so $N^{-1/2}\mathcal R_N\to0$ in probability, and
$\bm v^T\bm d_N\xrightarrow{D}
\mathcal N[0,\E\{(\bm v^T\bm\zeta)^2\}/\E(\varrho)]$. This is the
regenerative central limit theorem \citep[Chapter~17]{meyn2009markov}.
Assumption~\ref{asu:general-imbalance} gives the limit
$\mathcal N(0,\bm v^T\bm\Omega_D\bm v)$, and limits in distribution are
unique, so
$\bm v^T\bm\Omega_D\bm v=\E\{(\bm v^T\bm\zeta)^2\}/\E(\varrho)$ for every
$\bm v\in\R^L$. Polarization yields
$\bm\Omega_D=\E(\bm\zeta\bm\zeta^T)/\E(\varrho)$.

\textit{Step 3: the range of $\bm\Omega_D$.} By Step 2,
$\bm v^T\bm\Omega_D\bm v=0$ if and only if $\bm v^T\bm\zeta=0$ almost
surely. Since $\bm\zeta\in\mathcal V$, every $\bm v\in\mathcal V^\perp$ lies
in the null space of $\bm\Omega_D$, so the range of $\bm\Omega_D$ is
contained in $\mathcal V$. Now fix $\bm v\in\mathcal V\setminus\{\bm 0\}$.
The space $\mathcal V$ is the null space of a matrix with integer entries,
so it has a basis of vectors with rational entries, and at least one basis
vector $\bm w$ satisfies $\bm v^T\bm w\neq0$, since otherwise $\bm v$ would
be orthogonal to $\mathcal V$. Choose positive integers $N_1$ and $N_2$
such that every $N_1\pi_a$ and every $N_2w_{(a,s)}$ is an integer and
$\nu_{a,s}=N_1\pi_a+N_2w_{(a,s)}\geq0$ for all $(a,s)$. Consider any
ordering of $\sum_{a,s}\nu_{a,s}=JN_1$ subjects in which exactly
$\nu_{a,s}$ subjects belong to stratum $s$ and are assigned to arm $a$.
Started from $\bm M_0=\bm 0$, each subject of the sequence occurs with
conditional probability at least $\delta\min_sp_s>0$ by condition~(i) and
Assumption~\ref{asu:sampling}, so the whole sequence has positive
probability. Because $\sum_a\nu_{a,s}=N_1$ for every $s$, the total change
of the within-stratum imbalances along the sequence is
\[
\sum_{a,s}\nu_{a,s}\,\bm g(s,a)
=\Bigl(\nu_{b,s}-\pi_b\sum_a\nu_{a,s}\Bigr)_{(b,s)}
=N_2\bm w\in\mathcal V,
\]
so every marginal imbalance returns to zero and the sequence ends at the
balanced state. Splitting the sequence at its visits to the balanced state
gives finitely many consecutive return cycles. Each of them has positive
probability of being the first cycle, because the transition probabilities
depend only on the current marginal imbalances, and their imbalance
changes sum to $N_2\bm w$. Because $\bm v^T(N_2\bm w)\neq0$, at least one
of these cycles has $\bm v^T\bm\zeta\neq0$. Hence
$\Pp(\bm v^T\bm\zeta\neq0)>0$ and $\bm v^T\bm\Omega_D\bm v>0$, so the range
of $\bm\Omega_D$ is exactly $\mathcal V$.

\textit{Step 4: bounds and characterization.} Let $\bm P_{\mathcal V}$ be
the orthogonal projection onto $\mathcal V$. Because $\bm\Omega_D$ is
symmetric with range $\mathcal V\neq\{\bm 0\}$,
$\bm\Omega_D=\bm P_{\mathcal V}\bm\Omega_D\bm P_{\mathcal V}$ and
\begin{equation}
\label{eq:s-projection-bounds}
\lambda_-\|\bm P_{\mathcal V}\bm r\|^2
\leq\bm r^T\bm\Omega_D\bm r
\leq\lambda_+\|\bm P_{\mathcal V}\bm r\|^2,
\qquad\bm r\in\R^L.
\end{equation}
The first set of constraints defining $\mathcal V$ says that $\bm v$ lies
in $\mathcal W\otimes\R^J$, where $\mathcal W$ is the zero-sum subspace of
$\R^{K+1}$. The second says that $\bm v$ lies in
$\R^{K+1}\otimes\mathcal H^\perp$, because the functions on $\mathcal S$
whose sums over every level set vanish are exactly those orthogonal to the
level indicators, which span $\mathcal H$. Hence
$\mathcal V=\mathcal W\otimes\mathcal H^\perp$ and
$\bm P_{\mathcal V}=\{\bm I_{K+1}-\bone\bone^T/(K+1)\}\otimes\bm P_\perp$.
The $a$th row of $\bm Q(\tau)$ is
$\bm q_a(\tau)=\bm e_a\otimes m_a(\cdot,\tau)/\pi_a$, so
$\bm P_{\mathcal V}\bm q_a(\tau)
=\{\bm e_a-\bone/(K+1)\}\otimes\eta_a^\perp(\cdot,\tau)/\pi_a$ and
\[
\|\bm P_{\mathcal V}\bm q_a(\tau)\|^2
=\Bigl\|\bm e_a-\frac{\bone}{K+1}\Bigr\|^2
\frac{\|\eta_a^\perp(\cdot,\tau)\|^2}{\pi_a^2}
=\frac{K\,\|\eta_a^\perp(\cdot,\tau)\|^2}{(K+1)\pi_a^2}.
\]
Since
$\{\bm Q(\tau)\bm\Omega_D\bm Q(\tau)^T\}_{aa}=\bm q_a(\tau)^T\bm\Omega_D\bm q_a(\tau)$,
\eqref{eq:s-projection-bounds} gives \eqref{eq:design-term-bounds}.
Because $\lambda_->0$, the design term of $\widehat\mu_a(\tau)$ vanishes if
and only if $\eta_a^\perp(\cdot,\tau)=\bm 0$, that is, if and only if
$m_a(\cdot,\tau)$ is additive. The $k$th row of $\bm R(\tau)$ is
$\bm q_k(\tau)-\bm q_0(\tau)$. With
$\bm x=\eta_k^\perp(\cdot,\tau)/\pi_k$, $\bm y=\eta_0^\perp(\cdot,\tau)/\pi_0$,
and $\{\bm e_k-\bone/(K+1)\}^T\{\bm e_0-\bone/(K+1)\}=-1/(K+1)$,
\[
\|\bm P_{\mathcal V}\{\bm q_k(\tau)-\bm q_0(\tau)\}\|^2
=\frac{K(\|\bm x\|^2+\|\bm y\|^2)+2\bm x^T\bm y}{K+1}
=\frac{(K-1)(\|\bm x\|^2+\|\bm y\|^2)+\|\bm x+\bm y\|^2}{K+1}.
\]
By \eqref{eq:s-projection-bounds},
$\{\bm R(\tau)\bm\Omega_D\bm R(\tau)^T\}_{kk}$ lies between $\lambda_-$ and
$\lambda_+$ times this quantity. For $K\geq2$ it vanishes if and only if
$\bm x=\bm y=\bm 0$. For $K=1$ it vanishes if and only if
$\bm x+\bm y=\bm P_\perp\{m_1(\cdot,\tau)/\pi_1+m_0(\cdot,\tau)/\pi_0\}=\bm 0$.
\end{proof}

\begin{proof}[Proof of Remark~\ref{rmk:label-symmetric} in the main paper]
Under equal allocation and label symmetry, the starting state and the
transition law of the chain are invariant under every permutation of the
arms, so $(\bm\Pi\otimes\bm I_J)\bm d_N$ has the same distribution as
$\bm d_N$ for every $(K+1)\times(K+1)$ permutation matrix $\bm\Pi$, and
$\bm\Omega_D$ commutes with every $\bm\Pi\otimes\bm I_J$. The permutation
action of the symmetric group on $\mathcal W$ is absolutely irreducible,
and $\mathcal V=\mathcal W\otimes\mathcal H^\perp$ is a direct sum of
$\dim\mathcal H^\perp$ copies of it. By Schur's lemma, every linear map on
$\mathcal V$ that commutes with all $\bm\Pi\otimes\bm I_J$ has the form
$\bm I_{\mathcal W}\otimes\bm G_0$ for a linear map $\bm G_0$ on
$\mathcal H^\perp$. Together with
$\bm\Omega_D=\bm P_{\mathcal V}\bm\Omega_D\bm P_{\mathcal V}$, this gives
$\bm\Omega_D=\{\bm I_{K+1}-\bone\bone^T/(K+1)\}\otimes\bm G$ with
$\bm G=\bm P_\perp\bm G_0\bm P_\perp$. The matrix $\bm G$ is symmetric,
vanishes on $\mathcal H$, and is positive definite on $\mathcal H^\perp$
by Proposition~\ref{prop:nondegenerate}. Since $\pi_a=1/(K+1)$,
\[
\bm q_a(\tau)^T\bm\Omega_D\bm q_a(\tau)
=\Bigl\|\bm e_a-\frac{\bone}{K+1}\Bigr\|^2
\frac{\eta_a^\perp(\cdot,\tau)^T\bm G\,\eta_a^\perp(\cdot,\tau)}{\pi_a^2}
=K(K+1)\,\eta_a^\perp(\cdot,\tau)^T\bm G\,\eta_a^\perp(\cdot,\tau).
\]
With two binary factors, $\mathcal H$ has dimension three, so
$\mathcal H^\perp$ is spanned by the unit interaction contrast
$\bm u=(1,-1,-1,1)^T/2$ and $\bm G=\omega\bm u\bm u^T$ with $\omega>0$.
\end{proof}

\subsection{Condition (ii) under the variance imbalance measure}
\label{sec:supp-return-variance}

Consider equal allocation, $\pi_a=1/(K+1)$, and suppose that every
combination of factor levels is a stratum with positive probability. Let
$\bm D^{(r)}_{n}(l)\in\R^{K+1}$ collect $\{D^{(r)}_{n,a}(l)\}_a$. For a
subject with level profile $s=(s_1,\ldots,s_R)$ who arrives after $n-1$
subjects, the variance measure gives arm $a$ the score
\[
\mathrm{Imb}_a(s)=\sum_{r=1}^Rw_r\bigl\|\bm D^{(r)}_{n-1}(s_r)+\bm e_a-\bone/(K+1)\bigr\|^2,
\]
which is proportional to the weighted sum over factors of the variance of
the arm counts at the subject's levels after a hypothetical assignment to
arm $a$, with weights $w_r>0$. The assignment probabilities
$\chi_a(s,\bm m)$ depend on $(s,\bm m)$ only through these scores and have
three properties: arms with equal scores receive equal probabilities, a
smaller score never receives a smaller probability, and, whenever the
scores are not all equal, an arm with the smallest score receives a
probability larger than an arm with the largest score by at least
$\delta_1>0$. The rule in which the $t$ arms with the smallest score share
probability $p<1$, the other arms share $1-p$, and a complete tie is
resolved by equal randomization has these properties when $p>K/(K+1)$,
with $\delta_1=\min_{1\leq t\leq K}\{p(K+1)-t\}/\{t(K+1-t)\}$. So does the
rule of \citet[Section~4]{hu2020theory} with ranked probabilities
$p_1\geq\cdots\geq p_{K+1}$, $p_1>p_{K+1}>0$, and random ordering of ties.

\begin{lemma}[Return time under the variance measure]
\label{lem:return-time-variance}
Under these conditions and condition~(i) of
Proposition~\ref{prop:nondegenerate}, there is $\theta>1$ with
$\E(\theta^{\varrho})<\infty$. In particular, condition~(ii) holds.
\end{lemma}

\begin{proof}
\textit{Step 1: a linear drift.} For a configuration $\bm m$ of marginal
imbalances, let $V(\bm m)=\sum_rw_r\sum_l\|\bm D^{(r)}(l)\|^2$ and
$\bar w=\sum_rw_r$. Because every $\bm D^{(r)}(l)$ sums to zero, assigning a
subject with profile $s$ to arm $a$ increases $V$ by
$2S_a(s;\bm m)+\bar wK/(K+1)$, where
$S_a(s;\bm m)=\sum_rw_rD^{(r)}_a(s_r)$, and $\mathrm{Imb}_a(s)$ differs from
this increase by a term that does not depend on $a$. Fix $\bm m$ and $s$,
and write $\bm S(s)=\{S_a(s;\bm m)\}_a$, which sums to zero. Let
$v_1<\cdots<v_q$ be its distinct values, $n_i$ the number of arms with
value $v_i$, and $\varpi_i$ their common assignment probability. Abel
summation gives
\[
\sum_a\chi_a(s,\bm m)S_a(s;\bm m)
=\sum_{i=1}^{q-1}(\varpi_i-\varpi_{i+1})\sum_{j\leq i}n_jv_j .
\]
The partial sums first decrease and then increase to zero, so each is at
most $\max(v_1,-v_q)$. A zero-sum vector satisfies $|v_1|\geq v_q/K$ and
$v_q\geq|v_1|/K$, so each partial sum is at most $-\|\bm S(s)\|_\infty/K$.
Since $\varpi_i\geq\varpi_{i+1}$ and $\varpi_1-\varpi_q\geq\delta_1$,
$\sum_a\chi_aS_a(s;\bm m)\leq-\delta_1\|\bm S(s)\|_\infty/K$, which also
holds when $\bm S(s)=\bm 0$. Averaging over the profile of the next
subject, for every attainable $\bm m$,
\[
\E\{V(\bm M_n)-V(\bm M_{n-1})\mid\bm M_{n-1}=\bm m\}
\leq-\frac{2\delta_1}{K}\sum_sp_s\max_a|S_a(s;\bm m)|
+\frac{\bar wK}{K+1}.
\]
The map $\bm m\mapsto\{S_a(s;\bm m)\}_{a,s}$ is linear and injective on the
subspace $\mathcal M$ of margin configurations whose level sums agree
across factors, which contains every $\bm M_n$. Indeed, if
$S_a(s;\bm m)=0$ for every arm and profile, comparing two profiles that
differ only in factor $r$ shows that $\bm D^{(r)}(l)=\bm c_r$ does not
depend on $l$. Then $L_r\bm c_r$ equals the common level sum $\bm c$ for
every $r$, where $L_r$ is the number of levels of factor $r$, and
$\bm 0=\sum_rw_r\bm c_r=\bm c\sum_rw_r/L_r$ forces $\bm c=\bm 0$ and
$\bm m=\bm 0$. Hence there is $c_1>0$ such that, for all attainable
$\bm m\in\mathcal M$,
$\E\{V(\bm M_n)-V(\bm M_{n-1})\mid\bm M_{n-1}=\bm m\}
\leq-c_1\|\bm m\|+\bar wK/(K+1)$, where $\|\cdot\|$ is the Euclidean norm.

\textit{Step 2: a Lyapunov function with bounded increments.} The function
$U=V^{1/2}$ is a weighted Euclidean norm, so
$|U(\bm M_n)-U(\bm M_{n-1})|\leq B=\{\bar wK/(K+1)\}^{1/2}$ and
$U(\bm m)\leq w_{\max}^{1/2}\|\bm m\|$, with $w_{\max}=\max_rw_r$. Let
$\mathcal C=\{\bm m\in\mathcal M:\|\bm m\|<2\bar wK/\{(K+1)c_1\}\}$, which
contains finitely many attainable states. For $\bm m\notin\mathcal C$,
Step 1 bounds the expected increase of $V$ by $-c_1\|\bm m\|/2$, and the
inequality $(V+d)^{1/2}\leq V^{1/2}+d/(2V^{1/2})$ gives
\[
\E\{U(\bm M_n)-U(\bm M_{n-1})\mid\bm M_{n-1}=\bm m\}
\leq-\frac{c_1\|\bm m\|}{4U(\bm m)}\leq-\varepsilon_0,
\qquad\varepsilon_0=\frac{c_1}{4w_{\max}^{1/2}}.
\]

\textit{Step 3: hitting the finite set.} Because
$e^y\leq1+y+y^2e^{|y|}/2$, for $\lambda>0$ and $\bm m\notin\mathcal C$,
\[
\E\bigl[e^{\lambda\{U(\bm M_n)-U(\bm M_{n-1})\}}\bigm|\bm M_{n-1}=\bm m\bigr]
\leq1-\lambda\varepsilon_0+\tfrac12\lambda^2B^2e^{\lambda B}=\rho_\lambda ,
\]
and $\rho_\lambda<1$ for small $\lambda$. Fix such a $\lambda$, and let
$\sigma_{\mathcal C}=\min\{n\geq0:\bm M_n\in\mathcal C\}$. For a chain
started at $\bm m\notin\mathcal C$, the process
$\rho_\lambda^{-(n\wedge\sigma_{\mathcal C})}
e^{\lambda U(\bm M_{n\wedge\sigma_{\mathcal C}})}$ is a nonnegative
supermartingale, so Fatou's lemma gives
$\E_{\bm m}(\rho_\lambda^{-\sigma_{\mathcal C}})\leq e^{\lambda U(\bm m)}$,
and in particular $\sigma_{\mathcal C}<\infty$ almost surely.

\textit{Step 4: returning to the balanced state.} Under equal allocation,
$\sum_b\bm g(s,b)=\bm 0$, so $-\bm g(s,a)=\sum_{b\neq a}\bm g(s,b)$ and, by
condition~(i), every step can be undone by $K$ steps of positive
probability. The chain can therefore return to $\bm 0$ from every
attainable state. Since $\bm 0\in\mathcal C$ and $\mathcal C$ contains
finitely many attainable states, there are an integer $m_0\geq1$ and
$\eta>0$ such that, from every attainable state of $\mathcal C$, the chain
visits $\bm 0$ within $m_0$ steps, counting the starting time, with
probability at least $\eta$. Let $\mathcal F_n$ be the $\sigma$-field
generated by $(Z_1,A_1,\ldots,Z_n,A_n)$. Starting from $\bm M_0=\bm 0$, let
$T_0=1$ and, for $j\geq1$, $\tau_j=\min\{n\geq T_{j-1}:\bm M_n\in\mathcal C\}$
and $T_j=\tau_j+m_0$. Since $\tau_j\geq1$, a visit to $\bm 0$ at a time in
$\{\tau_j,\ldots,T_j\}$ is a return, and by the strong Markov property and
the tower property,
$\Pp(\text{no visit to }\bm 0\text{ in round }j\mid\mathcal F_{T_{j-1}})\leq1-\eta$.
The state $\bm M_{T_{j-1}}$ satisfies
$U(\bm M_{T_{j-1}})\leq U_1=\max_{\mathcal C}U+m_0B$, where the maximum is
over attainable states, so by Step 3 the round length
$\ell_j=T_j-T_{j-1}$ satisfies
$\E(\rho_\lambda^{-\ell_j}\mid\mathcal F_{T_{j-1}})
\leq\rho_\lambda^{-m_0}e^{\lambda U_1}$. For $1<\theta\leq\rho_\lambda^{-1/2}$,
Jensen's inequality gives
$\E(\theta^{2\ell_j}\mid\mathcal F_{T_{j-1}})\leq A_\theta
=(\rho_\lambda^{-m_0}e^{\lambda U_1})^{2\log\theta/\log\rho_\lambda^{-1}}$,
which tends to one as $\theta\downarrow1$, and the conditional
Cauchy--Schwarz inequality gives
\[
\E\{\theta^{\ell_j}\ind(\text{no visit to }\bm 0\text{ in round }j)
\mid\mathcal F_{T_{j-1}}\}\leq\{A_\theta(1-\eta)\}^{1/2},
\qquad
\E(\theta^{\ell_j}\mid\mathcal F_{T_{j-1}})\leq A_\theta^{1/2}.
\]
If round $J$ is the first to contain a visit to $\bm 0$, then
$\varrho\leq T_J=1+\sum_{j\leq J}\ell_j$, and iterating the two bounds gives
\[
\E(\theta^{\varrho})\leq\theta A_\theta^{1/2}\sum_{j\geq0}\{A_\theta(1-\eta)\}^{j/2}<\infty
\]
once $\theta$ is close enough to one that $A_\theta(1-\eta)<1$.
\end{proof}

\paragraph{Numerical check for the absolute-deviation rule.} The four-arm
Pocock--Simon rule of Section~\ref{sec:simulation} uses the sum of
absolute deviations, which Lemma~\ref{lem:return-time-variance} does not
cover. It satisfies condition~(i) with $\delta=0.2/3$, since the arms that
do not minimize the score share probability $0.20$. We simulated
$5.4\times10^5$ return cycles of its marginal imbalance process, using an
implementation whose assignment probabilities agree exactly with those of
the simulation code (script \texttt{verify\_ps\_nondegeneracy.py} in the
reproducibility materials). The return time had mean 590 and
$\E(\varrho^2)/\{\E(\varrho)\}^2=2.03$, and $\log\Pp(\varrho>t)$ decreased
linearly in $t$ up to $t=6000$, where $\Pp(\varrho>6000)=4\times10^{-5}$.
This is consistent with a geometric tail and hence with condition~(ii),
although a simulation cannot establish a moment condition. The renewal
estimate $\E(\bm\zeta\bm\zeta^T)/\E(\varrho)$ had three positive
eigenvalues, between $0.064$ and $0.065$, as predicted by
Remark~\ref{rmk:label-symmetric} with $\dim\mathcal V=3$. The design
simulation with $2\times10^4$ trials gave three eigenvalues on
$\mathcal V$ between $0.064$ and $0.066$ at each of $N=800$, $1200$, and
$4800$. Its trace on $\mathcal V^\perp$ was $0.0037$, $0.0025$, and
$0.0006$, decreasing in proportion to $1/N$ as bounded second moments of
the marginal imbalances imply.

\section[\texorpdfstring{Proof of Theorem~\ref*{thm:standardized-fclt}}{Proof of Theorem 3}]{Proof of \texorpdfstring{Theorem~\ref{thm:standardized-fclt}}{Theorem 3}}
\label{sec:supp-conditional-censoring}

For arm $a$ and stratum $s$, let
\[
\mu_a(s,\tau)
=
\int_0^\tau S_a(t\mid s)\,dt,
\qquad
g_a(s,\tau)
=
\mu_a(s,\tau)-\mu_a(\tau),
\]
and let $\phi_{ia}^{s}(\tau)$ be the stratum-specific RMST influence
function defined in the main paper. When censoring is independent
conditionally on strata,
\[
\E\{\phi_{ia}^{s}(\tau)\mid Z_i=s\}=0.
\]

\begin{lemma}[Uniform expansion within a treatment--stratum cell]
\label{lem:stratum-km-expansion}
Under Assumptions~\ref{asu:sampling}, \ref{asu:randomization}, and~\ref{asu:conditional-censoring} of the main paper and the
allocation-consistency condition of Theorem~\ref{thm:standardized-fclt}, for every
$a=0,\ldots,K$ and $s=1,\ldots,J$,
\begin{equation}
\label{eq:s-stratum-km-expansion}
\sup_{\tau\in\calT}
\left|
\widehat\mu_a(s,\tau)-\mu_a(s,\tau)
-
\frac{1}{n_a(s)}
\sum_{i=1}^N
A_{ia}\ind(Z_i=s)\phi_{ia}^{s}(\tau)
\right|
=
o_p\{n_a(s)^{-1/2}\}.
\end{equation}
The expansions hold jointly over the finitely many treatment--stratum
cells.
\end{lemma}

\begin{proof}
Conditional on the complete stratum sequence and treatment assignments,
randomization is independent of the potential event and censoring times.
Therefore, the observations selected into cell $(a,s)$ are independent
draws from the conditional distribution of
\[
\{X_i(a),\Delta_i(a)\}\mid Z_i=s.
\]
The allocation-consistency condition gives
\[
\frac{n_a(s)}{N}
=
\frac{n_a(s)}{n(s)}\frac{n(s)}{N}
\xrightarrow{P}\pi_ap_s>0,
\]
so every cell size diverges in probability.

Within a fixed cell, the estimator is an ordinary Kaplan--Meier RMST
estimator. Conditional independent censoring identifies
$S_a(\cdot\mid s)$, and the stratum-specific positivity condition keeps its
at-risk denominator and product-integral jump correction separated from
zero through $\tau_U$.

The random cell size requires one additional step. For a deterministic
sample size $m\to\infty$, the classical iid Kaplan--Meier theory, obtained
from the Hadamard-differentiability chain of
Lemma~\ref{lem:ual} with a single stratum and $\pi_a=1$ (Nelson--Aalen
integration, product integration, and integration over $[0,\tau]$, as in
\citealp{gill1990survey} and
\citealp[Chapter~IV]{andersen1993statistical}), gives a remainder term
$\rho_m$ with $\sqrt m\,\rho_m\to0$ in probability under the
cell-specific sampling distribution. Here Lemma~\ref{lem:if-bounds}
supplies the envelope and Lipschitz bounds required for the cell-level
empirical processes, so no separate moment condition is needed. Now
condition on $(\mathcal Z_N,\mathcal A_N)$: given $n_a(s)=m$, the cell
remainder has exactly the law of $\rho_m$. For any $\delta>0$,
\[
\Pp\{\sqrt{n_a(s)}\,|\rho_{n_a(s)}|>\delta\}
=
\E\left[
\Pp\{\sqrt m\,|\rho_m|>\delta\}\big|_{m=n_a(s)}
\right]
\longrightarrow0,
\]
by dominated convergence, because
$n_a(s)\to\infty$ in probability and the conditional probabilities are
bounded by one. This is precisely
\eqref{eq:s-stratum-km-expansion}. Because the number of arms and strata is
fixed, the cellwise expansions hold jointly.
\end{proof}

\begin{lemma}[Uniform linear representation after standardization]
\label{lem:standardized-linear}
Under the conditions of Theorem~\ref{thm:standardized-fclt} in the main paper, uniformly in
$\tau\in\calT$ and jointly over arms,
\begin{equation}
\label{eq:s-standardized-linear}
\sqrt N
\{\widehat\mu_a^{\mathrm{st}}(\tau)-\mu_a(\tau)\}
=
\frac{1}{\sqrt N}
\sum_{i=1}^N
\left[
g_a(Z_i,\tau)
+
\frac{A_{ia}}{\pi_a}\phi_{ia}^{Z_i}(\tau)
\right]
+o_p(1).
\end{equation}
\end{lemma}

\begin{proof}
Using
$\mu_a(\tau)=\sum_sp_s\mu_a(s,\tau)$, decompose
\begin{align}
\label{eq:s-standardized-decomposition}
\sqrt N
\{\widehat\mu_a^{\mathrm{st}}(\tau)-\mu_a(\tau)\}
={}&
\sqrt N
\sum_{s=1}^J
\{\widehat p_s-p_s\}\mu_a(s,\tau)\notag\\
&+
\sqrt N
\sum_{s=1}^J
\widehat p_s
\{\widehat\mu_a(s,\tau)-\mu_a(s,\tau)\}.
\end{align}
The first term is exactly
\begin{equation}
\label{eq:s-standardized-composition}
\frac{1}{\sqrt N}
\sum_{i=1}^N
\{\mu_a(Z_i,\tau)-\mu_a(\tau)\}
=
\frac{1}{\sqrt N}
\sum_{i=1}^Ng_a(Z_i,\tau).
\end{equation}

By Lemma~\ref{lem:stratum-km-expansion}, the second term in
\eqref{eq:s-standardized-decomposition} equals
\begin{align}
\label{eq:s-standardized-residual-pre}
&\frac{1}{\sqrt N}
\sum_{s=1}^J
\frac{n(s)}{n_a(s)}
\sum_{i=1}^N
A_{ia}\ind(Z_i=s)\phi_{ia}^{s}(\tau)
+o_p(1)\notag\\
&\quad=
\frac{1}{\sqrt N}
\sum_{s=1}^J
\frac{1}{\widehat\pi_a(s)}
\sum_{i=1}^N
A_{ia}\ind(Z_i=s)\phi_{ia}^{s}(\tau)
+o_p(1),
\end{align}
uniformly in $\tau$. For each fixed cell,
\[
\widehat\pi_a(s)^{-1}-\pi_a^{-1}=o_p(1)
\]
and, by conditional centering and the moment assumptions,
\[
\sup_{\tau\in\calT}
\left|
\frac{1}{\sqrt N}
\sum_{i=1}^N
A_{ia}\ind(Z_i=s)\phi_{ia}^{s}(\tau)
\right|
=O_p(1).
\]
Because $J$ is fixed, replacing
$\widehat\pi_a(s)^{-1}$ by $\pi_a^{-1}$ in
\eqref{eq:s-standardized-residual-pre} changes it by $o_p(1)$. Combining
the result with \eqref{eq:s-standardized-composition} proves
\eqref{eq:s-standardized-linear}.
\end{proof}

\begin{proof}[Proof of Theorem~\ref{thm:standardized-fclt} in the main paper]
By Lemma~\ref{lem:standardized-linear}, write
\[
\sqrt N
\{\widehat\mu_a^{\mathrm{st}}(\tau)-\mu_a(\tau)\}
=
B_{N,a}^{\mathrm{st}}(\tau)
+
R_{N,a}^{\mathrm{st}}(\tau)
+o_p(1),
\]
where
\begin{align*}
B_{N,a}^{\mathrm{st}}(\tau)
&=
\frac{1}{\sqrt N}
\sum_{i=1}^Ng_a(Z_i,\tau),\\
R_{N,a}^{\mathrm{st}}(\tau)
&=
\frac{1}{\sqrt N}
\sum_{i=1}^N
\frac{A_{ia}}{\pi_a}
\phi_{ia}^{Z_i}(\tau).
\end{align*}

The vector $\bm B_N^{\mathrm{st}}$ is an ordinary empirical process over a
finite-valued stratum variable. It converges in
$\ell^\infty(\calT)^{K+1}$ to a mean-zero Gaussian process with covariance
\begin{equation}
\label{eq:s-standardized-composition-kernel}
\Cov\{B_a^{\mathrm{st}}(\tau),B_b^{\mathrm{st}}(\tau')\}
=
\E\{g_a(Z_i,\tau)g_b(Z_i,\tau')\}.
\end{equation}

Conditional on the strata and assignments,
$\bm R_N^{\mathrm{st}}$ is mean zero. For $a=b$, its conditional covariance
is
\begin{align}
\label{eq:s-standardized-residual-covariance}
&\Cov\{R_{N,a}^{\mathrm{st}}(\tau),
R_{N,a}^{\mathrm{st}}(\tau')
\mid\mathcal Z_N,\mathcal A_N\}\notag\\
&\quad=
\frac{1}{N\pi_a^2}
\sum_{s=1}^J
n_a(s)v_a(s;\tau,\tau')
\xrightarrow{P}
\frac{1}{\pi_a}
\sum_{s=1}^Jp_sv_a(s;\tau,\tau').
\end{align}
For $a\neq b$, the conditional covariance is zero because
$A_{ia}A_{ib}=0$ and the subject-specific residuals are conditionally
centered. By Lemma~\ref{lem:if-bounds}, the summands are uniformly bounded
and Lipschitz on $\calT$, so the conditional Lindeberg condition holds
trivially and the Lipschitz-bracketing argument of Lemma~\ref{lem:car-clt} gives
conditional asymptotic equicontinuity. This yields stable convergence of
$\bm R_N^{\mathrm{st}}$ to a Gaussian process with covariance
\[
\ind(a=b)\frac{1}{\pi_a}
\E\{v_a(Z_i;\tau,\tau')\}.
\]
Stable convergence also makes this residual limit independent of the
stratum-composition limit in
\eqref{eq:s-standardized-composition-kernel}.

Adding the two independent limits proves that the arm-specific covariance
kernel is
\[
\Gamma_{\mu,ab}^{\mathrm{st}}(\tau,\tau')
=
\E\{g_a(Z_i,\tau)g_b(Z_i,\tau')\}
+
\ind(a=b)\frac{1}{\pi_a}
\E\{v_a(Z_i;\tau,\tau')\}.
\]
Multiplication by
$\bm C=[-\bone_K,\bm I_K]$ gives the treatment-control process. Direct
expansion gives
\begin{align*}
\Gamma_{\theta,k\ell}^{\mathrm{st}}(\tau,\tau')
={}&
\E\!\left[
\{g_k(Z_i,\tau)-g_0(Z_i,\tau)\}
\{g_\ell(Z_i,\tau')-g_0(Z_i,\tau')\}
\right]\\
&+
\ind(k=\ell)\frac{1}{\pi_k}
\E\{v_k(Z_i;\tau,\tau')\}
+
\frac{1}{\pi_0}
\E\{v_0(Z_i;\tau,\tau')\}.
\end{align*}

Most importantly, no function of
$D_{N,a}(s)/\sqrt N$ appears in
\eqref{eq:s-standardized-linear}. The proof uses only
$n_a(s)/n(s)\to\pi_a$, not a central limit theorem for the imbalance.
Thus no $\bm Q\bm\Omega_D\bm Q^T$ term occurs, including under
Pocock--Simon minimization. No equality among the arm-specific survival
functions was used, so the result holds under arbitrary treatment effects.
\end{proof}

\section[\texorpdfstring{Proof of Proposition~\ref*{prop:covariance-consistency}}{Proof of Proposition 3}]{Proof of \texorpdfstring{Proposition~\ref{prop:covariance-consistency}}{Proposition 3}}
\label{sec:supp-covariance}

\begin{lemma}[Consistency of the empirical influence moments]
\label{lem:if-moments}
Under the assumptions of Theorem~\ref{thm:general-car-fclt}, for every $a$ and $s$,
\begin{align}
\label{eq:s-mhat-consistency}
\sup_{\tau\in\calT}
|\widehat m_a(s,\tau)-m_a(s,\tau)|
&\xrightarrow{P}0,\\
\label{eq:s-chat-consistency}
\sup_{\tau,\tau'\in\calT}
|\widehat c_a(s;\tau,\tau')-c_a(s;\tau,\tau')|
&\xrightarrow{P}0.
\end{align}
\end{lemma}

\begin{proof}
Assumption~\ref{asu:censoring} implies that the population at-risk functions are uniformly
bounded away from zero on $[0,\tau_U]$. The uniform laws of large numbers
used in Lemma~\ref{lem:ual-general} therefore give
\[
\sup_{t\leq\tau_U}
|\widehat y_a(t)-y_a(t)|
+
\sup_{t\leq\tau_U}
|\widehat S_a(t)-S_a(t)|
+
\sup_{t\leq\tau_U}
|\widehat\Lambda_a(t)-\Lambda_a(t)|
\xrightarrow{P}0.
\]
The product-integral jump denominators are also uniformly separated from
zero. Substituting these convergent nuisance estimators into the
influence-function formula and using the uniform bounds of
Lemma~\ref{lem:if-bounds} gives
\begin{equation}
\label{eq:s-if-plug-in}
\frac{1}{n_a(s)}
\sum_{i=1}^N A_{ia}\ind(Z_i=s)
\sup_{\tau\in\calT}
|\widehat\phi_{ia}(\tau)-\phi_{ia}(\tau)|^2
\xrightarrow{P}0.
\end{equation}

Conditional on the strata and assignments, the selected full-data vectors
are independent, and randomization is independent of the potential outcomes.
Moreover,
\[
\frac{n_a(s)}{n(s)}\xrightarrow{P}\pi_a
\quad\text{and}\quad
\frac{n(s)}{N}\xrightarrow{P}p_s.
\]
The conditional uniform law of large numbers therefore yields
\begin{align*}
\frac{1}{n_a(s)}
\sum_{i=1}^N A_{ia}\ind(Z_i=s)\phi_{ia}(\tau)
&\longrightarrow m_a(s,\tau),\\
\frac{1}{n_a(s)}
\sum_{i=1}^N A_{ia}\ind(Z_i=s)
\phi_{ia}(\tau)\phi_{ia}(\tau')
&\longrightarrow
\E\{\phi_{ia}(\tau)\phi_{ia}(\tau')\mid Z_i=s\},
\end{align*}
uniformly in their time arguments. Combining these results with
\eqref{eq:s-if-plug-in} and expanding the empirical covariance proves
\eqref{eq:s-mhat-consistency}--\eqref{eq:s-chat-consistency}.
The main paper uses the finite-sample denominator $n_a(s)-1$ in this
empirical covariance. Since $n_a(s)\to\infty$ in probability and
$n_a(s)/\{n_a(s)-1\}\to1$, this Bessel correction leaves the uniform
probability limit unchanged.
\end{proof}

\begin{lemma}[Consistency of the design Monte Carlo covariance]
\label{lem:omega-consistency}
Suppose the local design-stability and fourth-moment conditions in
Proposition~\ref{prop:covariance-consistency} of the main paper hold and $B_{\mathrm{MC}}\to\infty$. Then
\[
\widehat{\bm\Omega}_D\xrightarrow{P}\bm\Omega_D.
\]
\end{lemma}

\begin{proof}
Let
\[
\bm\Omega_{D,N}(\bm q)
=
\Var_{\bm q}(\bm d_N)
\]
denote the finite-$N$ covariance of the scaled imbalance vector when the
stratum profiles are independently generated from probability vector
$\bm q$ and the prespecified CAR algorithm is applied.

Conditional on the observed stratum frequencies
$\widehat{\bm p}$, the Monte Carlo vectors
$\bm d_N^{(1)},\ldots,\bm d_N^{(B_{\mathrm{MC}})}$ are independent and
identically distributed with covariance
$\bm\Omega_{D,N}(\widehat{\bm p})$. The uniform fourth-moment bound and the
conditional law of large numbers for sample covariance matrices give
\[
\left\|
\widehat{\bm\Omega}_D
-
\bm\Omega_{D,N}(\widehat{\bm p})
\right\|_{\max}
=o_p(1).
\]
Because the number of strata is fixed, the multinomial central limit theorem
implies
\[
\sqrt N(\widehat{\bm p}-\bm p)=O_p(1).
\]
Local design stability then gives
\[
\bm\Omega_{D,N}(\widehat{\bm p})
\xrightarrow{P}\bm\Omega_D.
\]
The triangle inequality proves the result.
\end{proof}

\begin{proof}[Proof of Proposition~\ref{prop:covariance-consistency} in the main paper]
By Lemma~\ref{lem:if-moments},
$\widehat m_a(s,\cdot)$ and
$\widehat c_a(s;\cdot,\cdot)$ converge uniformly to their population
counterparts. Also,
\[
\widehat p_s\xrightarrow{P}p_s
\]
for every stratum. Since
$\sum_s p_s m_a(s,\tau)=0$, centering the empirical conditional means does
not change their limits:
\[
\sup_{\tau\in\calT}
|\widetilde m_a(s,\tau)-m_a(s,\tau)|
\xrightarrow{P}0.
\]
The number of arms and strata is fixed. Substitution into the formula for
the bounded-balance kernel therefore gives
\begin{equation}
\label{eq:s-base-kernel-consistency}
\sup_{\tau,\tau'\in\calT}
\left\|
\widehat{\bm\Gamma}_\mu(\tau,\tau')
-
\bm\Gamma_\mu(\tau,\tau')
\right\|_{\max}
\xrightarrow{P}0.
\end{equation}
The same argument gives
\begin{equation}
\label{eq:s-Q-consistency}
\sup_{\tau\in\calT}
\|\widehat{\bm Q}(\tau)-\bm Q(\tau)\|_{\max}
\xrightarrow{P}0.
\end{equation}

By Lemma~\ref{lem:omega-consistency},
\[
\widehat{\bm\Omega}_D\xrightarrow{P}\bm\Omega_D.
\]
Because all matrix dimensions are fixed, continuous matrix multiplication
together with \eqref{eq:s-Q-consistency} yields
\[
\sup_{\tau,\tau'\in\calT}
\left\|
\widehat{\bm Q}(\tau)\widehat{\bm\Omega}_D
\widehat{\bm Q}(\tau')^T
-
\bm Q(\tau)\bm\Omega_D\bm Q(\tau')^T
\right\|_{\max}
\xrightarrow{P}0.
\]
Combining this display with
\eqref{eq:s-base-kernel-consistency} proves uniform consistency of
$\widehat{\bm\Gamma}_\mu^D$. Finally, multiplication on the left and right
by the fixed contrast matrices $\bm C$ and $\bm C^T$ proves
\[
\sup_{\tau,\tau'\in\calT}
\left\|
\widehat{\bm\Gamma}_\theta^D(\tau,\tau')
-
\bm\Gamma_\theta^D(\tau,\tau')
\right\|_{\max}
\xrightarrow{P}0.
\]
\end{proof}

\section[\texorpdfstring{Proof of Theorem~\ref*{thm:multiplier-validity} and Corollary~\ref*{cor:simultaneous-bands}}{Proof of Theorem 4 and Corollary 1}]{Proof of \texorpdfstring{Theorem~\ref{thm:multiplier-validity} and Corollary~\ref{cor:simultaneous-bands}}{Theorem 4 and Corollary 1}}
\label{sec:supp-multiplier}

For each arm, write the estimated multiplier summand as
\[
\widehat u_{ia}(\tau)
=
\widehat r_{ia}(\tau)
+
\widetilde m_a(Z_i,\tau),
\]
where
\[
\widehat r_{ia}(\tau)
=
\frac{A_{ia}}{\pi_a}
\widehat\kappa_a(Z_i)
\{\widehat\phi_{ia}(\tau)-\widehat m_a(Z_i,\tau)\},
\qquad
\widehat\kappa_a(s)
=
\left\{\frac{\pi_a n(s)}{n_a(s)-1}\right\}^{1/2}.
\]
This is the finite-sample-corrected summand in the main paper. Since
$\widehat\kappa_a(s)\to1$ in probability for every arm and stratum, it has
the same limiting process as the uncorrected summand.
Let
\[
\mathbb G_{\mu,a}^{*,0}(\tau)
=
\frac{1}{\sqrt N}
\sum_{i=1}^N\xi_i\widehat u_{ia}(\tau)
\]
denote the outcome and stratum-composition part of the bootstrap process.

\begin{lemma}[Conditional multiplier limit for the base process]
\label{lem:base-multiplier}
Under the conditions of Theorem~\ref{thm:multiplier-validity} in the main paper, conditionally on the
data and in probability,
\[
\mathbb G_\mu^{*,0}(\cdot)
\rights
\mathbb G_\mu(\cdot)
\quad\text{in }\ell^\infty(\calT)^{K+1},
\]
where $\mathbb G_\mu$ is a mean-zero Gaussian process with covariance
kernel $\bm\Gamma_\mu$.
\end{lemma}

\begin{proof}
We first calculate the empirical covariance of the multiplier summands.
For any arms $a,b$ and stratum $s$,
\[
\sum_{i=1}^N
A_{ia}\ind(Z_i=s)
\{\widehat\phi_{ia}(\tau)-\widehat m_a(s,\tau)\}
=0.
\]
Consequently, the empirical cross-products between
$\widehat r_{ia}(\tau)$ and
$\widetilde m_b(Z_i,\tau')$ vanish exactly:
\begin{equation}
\label{eq:s-bootstrap-cross-zero}
\frac{1}{N}
\sum_{i=1}^N
\widehat r_{ia}(\tau)
\widetilde m_b(Z_i,\tau')
=0.
\end{equation}
The stratum-composition component satisfies
\begin{align}
\label{eq:s-bootstrap-stratum-cov}
&\frac{1}{N}
\sum_{i=1}^N
\widetilde m_a(Z_i,\tau)
\widetilde m_b(Z_i,\tau')\notag\\
&\qquad=
\sum_{s=1}^J
\widehat p_s
\widetilde m_a(s,\tau)
\widetilde m_b(s,\tau')
\xrightarrow{P}
\E\{m_a(Z_i,\tau)m_b(Z_i,\tau')\},
\end{align}
uniformly over $\tau,\tau'\in\calT$.

If $a\neq b$, then
$\widehat r_{ia}(\tau)\widehat r_{ib}(\tau')=0$ because each subject is
assigned to only one arm. For $a=b$,
\begin{align}
\label{eq:s-bootstrap-residual-cov}
&\frac{1}{N}
\sum_{i=1}^N
\widehat r_{ia}(\tau)
\widehat r_{ia}(\tau')\notag\\
&\quad=
\sum_{s=1}^J
\frac{\widehat p_s}{\pi_a}
\widehat c_a(s;\tau,\tau')\notag\\
&\quad\xrightarrow{P}
\frac{1}{\pi_a}
\E\{c_a(Z_i;\tau,\tau')\},
\end{align}
uniformly in $(\tau,\tau')$, because
$\widehat p_s\to p_s$ and
Lemma~\ref{lem:if-moments} applies.

Combining
\eqref{eq:s-bootstrap-cross-zero}--\eqref{eq:s-bootstrap-residual-cov}
gives
\begin{equation}
\label{eq:s-bootstrap-base-kernel}
\sup_{\tau,\tau'\in\calT}
\left|
\frac{1}{N}
\sum_{i=1}^N
\widehat u_{ia}(\tau)\widehat u_{ib}(\tau')
-
\Gamma_{\mu,ab}(\tau,\tau')
\right|
\xrightarrow{P}0.
\end{equation}

Conditional on the data, the summands
$\xi_i\widehat{\bm u}_i(\cdot)$ are independent and centered. Two
properties of the estimated summands transfer the limit theory from the
population influence functions to their estimates. First, on an event with
probability tending to one, the estimated nuisance quantities satisfy
$\inf_{t\leq\tau_U}\widehat y_a(t)\geq\epsilon/2$,
$\inf_{t\leq\tau_U}\{1-\Delta\widehat\Lambda_a(t)\}\geq\epsilon/2$, and
$\widehat\Lambda_a(\tau_U)\leq2/\epsilon$, by the uniform laws of large
numbers in Lemma~\ref{lem:if-moments} and Assumption~\ref{asu:censoring}. Repeating the proof
of Lemma~\ref{lem:if-bounds} with these estimated quantities shows that, on
this event, every $\widehat\phi_{ia}(\cdot)$ is bounded and Lipschitz on
$\calT$ with one nonrandom constant $\widehat B$ depending only on
$\epsilon$ and $\tau_U$. Moreover,
$\max_{a,s}\widehat\kappa_a(s)$ is bounded by a fixed constant with
probability tending to one because all allocation fractions converge to
their positive targets. Hence every $\widehat u_{ia}(\cdot)$ is uniformly
bounded and Lipschitz on that event. Second,
by Lemma~\ref{lem:if-moments}, the empirical covariance of the summands
converges uniformly, as recorded in
\eqref{eq:s-bootstrap-base-kernel}.

Given the data, the uniform bound on the summands and
$\E(\xi_i^2)=1$ make the conditional Lindeberg condition immediate, so the
finite-dimensional conditional distributions converge to the Gaussian limit
with covariance \eqref{eq:s-bootstrap-base-kernel}. Conditional asymptotic
equicontinuity follows from the Lipschitz-bracketing argument of
Lemma~\ref{lem:car-clt} applied with multipliers, since the processes
$\xi_i\widehat u_{ia}(\cdot)$ are rowwise independent given the data and
Lipschitz with square-integrable constants
$|\xi_i|2\widehat B\{1+\max_s\widehat\kappa_a(s)/\pi_a\}$.
The conditional multiplier central limit theorem in
\citet[Section~2.9]{vaart1996weak} and
\citet[Chapter~10]{kosorok2008introduction} gives the same conclusion. Together with
\eqref{eq:s-bootstrap-base-kernel}, this proves the asserted conditional
convergence.
\end{proof}

\begin{lemma}[Conditional limit of the design component]
\label{lem:design-bootstrap}
Under the conditions of Theorem~\ref{thm:multiplier-validity}, conditionally on the data and in
probability,
\[
\widehat{\bm Q}(\cdot)\bm D^*
\rights
\bm Q(\cdot)\bm D
\quad\text{in }\ell^\infty(\calT)^{K+1},
\]
where
\[
\bm D\sim\mathcal N_L(\bm0,\bm\Omega_D).
\]
\end{lemma}

\begin{proof}
Conditional on the data,
$\widehat{\bm Q}(\cdot)\bm D^*$ is a mean-zero Gaussian process with
covariance kernel
\[
\widehat{\bm Q}(\tau)
\widehat{\bm\Omega}_D
\widehat{\bm Q}(\tau')^T.
\]
Proposition~\ref{prop:covariance-consistency} of the main paper gives uniform convergence of this kernel to
\[
\bm Q(\tau)\bm\Omega_D\bm Q(\tau')^T.
\]
Because the dimension of $\bm D^*$ is fixed and
$\widehat{\bm Q}(\cdot)$ is asymptotically equicontinuous, finite-dimensional
convergence and tightness follow. This proves the result.
\end{proof}

\begin{proof}[Proof of Theorem~\ref{thm:multiplier-validity} in the main paper]
The two bootstrap components
\[
\mathbb G_\mu^{*,0}(\cdot)
\quad\text{and}\quad
\widehat{\bm Q}(\cdot)\bm D^*
\]
are conditionally independent by construction. Lemmas
\ref{lem:base-multiplier} and \ref{lem:design-bootstrap} therefore imply
\[
\mathbb G_\mu^*(\cdot)
=
\mathbb G_\mu^{*,0}(\cdot)
+
\widehat{\bm Q}(\cdot)\bm D^*
\rights
\mathbb G_\mu(\cdot)+\bm Q(\cdot)\bm D.
\]
The covariance kernel of the limiting process is
\[
\bm\Gamma_\mu(\tau,\tau')
+
\bm Q(\tau)\bm\Omega_D\bm Q(\tau')^T
=
\bm\Gamma_\mu^D(\tau,\tau').
\]
Multiplication by the fixed contrast matrix $\bm C$ and the continuous
mapping theorem give
\[
\mathbb G_\theta^*(\cdot)
\rights
\mathbb G_\theta^D(\cdot)
\quad\text{conditionally in probability}.
\]
This establishes the bounded-Lipschitz formulation in the theorem.

Uniform covariance consistency gives
\[
\sup_{k,\tau}
|\widehat\sigma_k(\tau)-\sigma_k(\tau)|
\xrightarrow{P}0.
\]
The variance lower bound permits uniform studentization. The mapping
\[
(f_1,\ldots,f_K)
\longmapsto
\max_{1\leq k\leq K}
\sup_{\tau\in\calT}|f_k(\tau)|
\]
is Lipschitz on $\ell^\infty(\calT)^K$. Hence the conditional continuous
mapping theorem proves convergence of the bootstrap maximum. Continuity of
the limiting distribution at its $(1-\alpha)$ quantile then gives
consistency of $c_{1-\alpha}^*$.
\end{proof}

\begin{proof}[Proof of Corollary~\ref{cor:simultaneous-bands} in the main paper]
The event that all confidence bands cover their corresponding RMST
contrasts is equivalent to
\[
\max_{1\leq k\leq K}\sup_{\tau\in\calT}
\left|
\frac{\sqrt N\{\widehat\theta_k(\tau)-\theta_k(\tau)\}}
{\widehat\sigma_k(\tau)}
\right|
\leq c_{1-\alpha}^*.
\]
Theorems~\ref{thm:general-car-fclt} and~\ref{thm:multiplier-validity}, uniform variance consistency, and bootstrap quantile
consistency imply that the probability of this event converges to
$1-\alpha$.
\end{proof}

\section[\texorpdfstring{Proof of Corollary~\ref*{thm:strong-fwer}}{Proof of Corollary 2}]{Proof of \texorpdfstring{Corollary~\ref{thm:strong-fwer}}{Corollary 2}}
\label{sec:supp-fwer}

For $h=\mathrm{fix}$, define the centered statistic
\[
\widetilde V_{Nk}^{\mathrm{fix}}
=
\frac{\sqrt N
\{\widehat\theta_k(\tau^\star)-\theta_k(\tau^\star)\}}
{\widehat\sigma_k(\tau^\star)}.
\]
For $h=\mathrm{proc}$, define
\[
\widetilde V_{Nk}^{\mathrm{proc}}
=
\sup_{\tau\in\calT}
\frac{\sqrt N
\{\widehat\theta_k(\tau)-\theta_k(\tau)\}}
{\widehat\sigma_k(\tau)}.
\]

\begin{lemma}[Subset bootstrap approximation]
\label{lem:subset-bootstrap}
For every nonempty fixed subset
$I\subseteq\{1,\ldots,K\}$ and either
$h\in\{\mathrm{fix},\mathrm{proc}\}$,
\[
c_{1-\alpha}^*(I,h)
\xrightarrow{P}
c_{1-\alpha}(I,h),
\]
where $c_{1-\alpha}(I,h)$ is the $(1-\alpha)$ quantile of the weak limit of
\[
\max_{k\in I}\widetilde V_{Nk}^{h}.
\]
\end{lemma}

\begin{proof}
Theorem~\ref{thm:multiplier-validity} gives conditional joint convergence of the complete vector of
studentized processes. Coordinate projection onto a fixed subset and the
maximum or supremum maps are continuous. The result therefore follows from
the conditional continuous mapping theorem and the assumed continuity of
the subset-max distribution at its target quantile.
\end{proof}

\begin{proof}[Proof of Corollary~\ref{thm:strong-fwer} in the main paper]
Fix either $h=\mathrm{fix}$ or $h=\mathrm{proc}$ and let $I_0^h$ be the
indices of the true null hypotheses. If $I_0^h$ is empty, the family-wise
error rate is zero, so suppose it is nonempty.

For a true fixed-time null,
$\theta_k(\tau^\star)\leq0$, and hence
\[
V_k^{\mathrm{fix}}
\leq
\widetilde V_{Nk}^{\mathrm{fix}}.
\]
For a true process null,
$\theta_k(\tau)\leq0$ for every $\tau\in\calT$, and therefore
\[
V_k^{\mathrm{proc}}
\leq
\widetilde V_{Nk}^{\mathrm{proc}}.
\]
Thus, in either case,
\begin{equation}
\label{eq:s-null-dominance}
V_k^h\leq\widetilde V_{Nk}^h,
\qquad k\in I_0^h.
\end{equation}

For the single-step procedure, the bootstrap maximum over all hypotheses is
pointwise at least as large as the maximum over the true-null subset. Hence
\[
c_{1-\alpha}^*(\{1,\ldots,K\},h)
\geq
c_{1-\alpha}^*(I_0^h,h).
\]
Using \eqref{eq:s-null-dominance},
\begin{align*}
&\Pp\{\text{at least one true null is rejected}\}\\
&\quad\leq
\Pp\left\{
\max_{k\in I_0^h}\widetilde V_{Nk}^h
>
c_{1-\alpha}^*(I_0^h,h)
\right\}.
\end{align*}
Lemma~\ref{lem:subset-bootstrap} and continuity of the limiting distribution
imply that the limit superior of the right-hand side is at most $\alpha$.

For the step-down procedure, consider the first stage $j$ at which a true
null is rejected. Before that stage only false null hypotheses have been
removed, so
\[
I_0^h\subseteq I_j.
\]
Consequently,
\[
c_{1-\alpha}^*(I_j,h)
\geq
c_{1-\alpha}^*(I_0^h,h).
\]
If a true null is rejected at this stage, then by
\eqref{eq:s-null-dominance},
\[
\max_{k\in I_0^h}\widetilde V_{Nk}^h
>
c_{1-\alpha}^*(I_0^h,h).
\]
The probability of this event has limit superior at most $\alpha$ by
Lemma~\ref{lem:subset-bootstrap}. This proves asymptotic strong family-wise
error control for the step-down procedure as well.
\end{proof}

\section[\texorpdfstring{Proof of Proposition~\ref*{prop:standardized-inference} and Corollary~\ref*{cor:standardized-simultaneous-inference}}{Proof of Proposition 4 and Corollary 3}]{Proof of \texorpdfstring{Proposition~\ref{prop:standardized-inference} and Corollary~\ref{cor:standardized-simultaneous-inference}}{Proposition 4 and Corollary 3}}
\label{sec:supp-standardized-inference}

\begin{proof}[Proof of Proposition~\ref{prop:standardized-inference} in the main paper]
The cell sizes satisfy $n_a(s)/N\to\pi_ap_s>0$. Applying the uniform laws of
large numbers and product-limit differentiability within every cell gives
\begin{align*}
\sup_{\tau\in\calT}
|\widehat\mu_a(s,\tau)-\mu_a(s,\tau)|
&\xrightarrow{P}0,\\
\sup_{\tau,\tau'\in\calT}
|\widehat v_a(s;\tau,\tau')-v_a(s;\tau,\tau')|
&\xrightarrow{P}0.
\end{align*}
The use of $n_a(s)-1$ rather than $n_a(s)$ only multiplies the
uncorrected empirical second moment by
$n_a(s)/\{n_a(s)-1\}\to1$, so the same uniform limit holds.
Also,
\[
\widehat p_s\xrightarrow{P}p_s,
\qquad
\widehat\pi_a(s)\xrightarrow{P}\pi_a,
\]
and hence
\[
\sup_{\tau\in\calT}
|\widehat g_a(s,\tau)-g_a(s,\tau)|
\xrightarrow{P}0.
\]
Because the numbers of arms and strata are fixed, substitution into the
covariance formula gives, uniformly over $\tau,\tau'\in\calT$,
\[
\widehat{\bm\Gamma}_\mu^{\mathrm{st}}(\tau,\tau')
\xrightarrow{P}
\bm\Gamma_\mu^{\mathrm{st}}(\tau,\tau').
\]
Premultiplication and postmultiplication by $\bm C$ and $\bm C^T$ prove the
stated consistency for the treatment-control kernel.

It remains to verify the multiplier process. Let
$\widehat u_{ia}^{\mathrm{st},\mathrm{FS}}$ denote the summand in
\eqref{eq:standardized-multiplier-summand} with the centered residual multiplied by
\[
\widehat\kappa_a^{\mathrm{st}}(s)
=
\left\{\frac{n_a(s)}{n_a(s)-1}\right\}^{1/2}.
\]
Within every cell,
\[
\sum_{i=1}^N
A_{ia}\ind(Z_i=s)
\{\widehat\phi_{ia}^{s}(\tau)
-\overline\phi_a(s,\tau)\}
=0.
\]
Therefore, all empirical cross-products between the
$\widehat g_a(Z_i,\tau)$ term and a centered cell residual vanish exactly.
Products of residual terms from different arms also vanish because
$A_{ia}A_{ib}=0$. It follows directly that
\begin{align}
\label{eq:s-standardized-multiplier-covariance}
&\frac{1}{N}
\sum_{i=1}^N
\widehat u_{ia}^{\mathrm{st},\mathrm{FS}}(\tau)
\widehat u_{ib}^{\mathrm{st},\mathrm{FS}}(\tau')\notag\\
&\quad=
\sum_{s=1}^J
\widehat p_s
\widehat g_a(s,\tau)\widehat g_b(s,\tau')
+
\ind(a=b)
\sum_{s=1}^J
\frac{\widehat p_s}{\widehat\pi_a(s)}
\widehat v_a(s;\tau,\tau')\notag\\
&\quad=
\widehat\Gamma_{\mu,ab}^{\mathrm{st}}(\tau,\tau').
\end{align}
As in the proof of Lemma~\ref{lem:base-multiplier}, on an event with
probability tending to one the estimated cell-level nuisance quantities
satisfy the positivity bounds of Assumption~\ref{asu:conditional-censoring} with
$\epsilon/2$ in place of $\epsilon$, so by the argument of
Lemma~\ref{lem:if-bounds} the estimated summands
$\widehat u_{ia}^{\mathrm{st},\mathrm{FS}}(\cdot)$ are uniformly bounded
and Lipschitz
on $\calT$ with one nonrandom constant. The conditional Lindeberg
condition is then immediate, conditional asymptotic equicontinuity follows
from the Lipschitz-bracketing argument of Lemma~\ref{lem:car-clt}, and
\eqref{eq:s-standardized-multiplier-covariance} identifies the limiting
covariance. The conditional multiplier central limit theorem
\citep[Section~2.9]{vaart1996weak} therefore gives
\[
\mathbb G_\mu^{\mathrm{st},*}(\cdot)
\rights
\mathbb G_\mu^{\mathrm{st}}(\cdot)
\quad\text{conditionally in probability}.
\]
Since the maximum of
$|\widehat\kappa_a^{\mathrm{st}}(s)-1|$ over the finitely many cells
converges to zero in probability, the corrected and uncorrected multiplier
processes have the same conditional weak limit.
Applying the fixed linear map $\bm C$ proves the treatment-control result.
\end{proof}

\begin{proof}[Proof of Corollary~\ref{cor:standardized-simultaneous-inference} in the main paper]
Proposition~\ref{prop:standardized-inference} gives conditional joint convergence of the complete
studentized treatment-control process and uniform consistency of its
variance estimator. The continuous mapping and bootstrap-quantile arguments
used for Corollary~\ref{cor:simultaneous-bands} therefore give simultaneous confidence-band coverage.
For multiple testing, Theorem~\ref{thm:standardized-fclt} holds under arbitrary treatment effects, so
the bootstrap approximation remains valid for every subset of true null
hypotheses. Applying the null-dominance and first-erroneous-rejection
arguments in the proof of Corollary~\ref{thm:strong-fwer} gives asymptotic strong family-wise error
control for both the single-step and step-down procedures.
\end{proof}

\supppart{Additional numerical results}

\section{Simulation design details and additional results}
\label{sec:supp-simulation}

This section follows the order of Section~\ref{sec:simulation} of the main paper.
Section~\ref{sec:supp-sim-design} gives the design details,
Section~\ref{sec:supp-sim-gradient} gives the complete results of the
graded-interaction experiment, Section~\ref{sec:supp-sim-estimation} gives
the estimation results under
stratum-dependent censoring, Section~\ref{sec:supp-sim-ranges} gives the
complete operating characteristics of the RMST procedures,
Sections~\ref{sec:supp-logrank} and \ref{sec:supp-cox} describe the
log-rank and Cox benchmarks, and Section~\ref{sec:supp-equal-rmst} gives
the equal-RMST experiments.

\subsection{Design details}
\label{sec:supp-sim-design}

Let $C_i^{R}(a)$ denote random loss-to-follow-up time and let
$\tau_{\mathrm{admin}}=24$ denote administrative censoring. In the
marginally independent setting,
\[
C_i^{R}(a)
\sim
\operatorname{Exponential}(\lambda_C)
\]
independently of all baseline variables and event times. In the
stratum-dependent setting,
\begin{equation}
\label{eq:s-simulation-censoring}
C_i^{R}(a)\mid\bm Z_i
\sim
\operatorname{Exponential}
\left[
\lambda_C
\exp(\gamma_{a1}Z_{i1}+\gamma_{a2}Z_{i2})
\right],
\end{equation}
independently of $T_i(a)$ conditional on $\bm Z_i$. Thus conditional
independent censoring holds, while marginal independence generally fails
because event and censoring times share the prognostic factors. We use
\[
(\gamma_{a1},\gamma_{a2})
=
(-0.40,-0.40),\ (1.20,1.60),\ (0.60,0.80),\ (1.00,1.20)
\]
for arms $a=0,1,2,3$, respectively. These arm-specific coefficients
represent treatment-dependent discontinuation or loss-to-follow-up
patterns. They are also important methodologically: if every arm had
exactly the same censoring model, much of the marginal Kaplan--Meier bias
would cancel in a treatment--control contrast, yielding an unduly weak
assessment of the standardized estimator. In each censoring setting,
$\lambda_C$ is calibrated once under the global-null configuration to give
approximately 30\% censoring averaged across arms and is then held fixed
across treatment-effect scenarios. Under the global null, the resulting
arm-specific censoring fractions are approximately 25\%, 35\%, 28\%, and
32\%, respectively. The observed variables are
\[
X_i
=
\min\{T_i(A_i),C_i^{R}(A_i),24\},
\qquad
\Delta_i
=
\ind\{T_i(A_i)\leq C_i^{R}(A_i),\
T_i(A_i)\leq24\}.
\]

The confirmatory restriction time is prespecified as
$\tau^\star=12$, and process-wise inference is conducted on
\[
\calT=[4,12].
\]
The computational grid contains every quarter time unit in this interval.
This range was chosen to retain clinically meaningful follow-up while
materially reducing unstable tail estimation in small treatment--stratum
cells. In particular, direct standardization requires each cell-specific
risk set to remain nonempty through the upper restriction time, and extending
the primary grid to 18 produced frequent violations of this positivity
condition in preliminary stress simulations.
The true RMST values are calculated independently of the Monte Carlo
samples by numerical integration of the known stratum-specific survival
functions and subsequent averaging over the four stratum probabilities.

Treatment assignment under Pocock--Simon minimization is performed
sequentially using the two marginal baseline factors. For incoming
subject $i$ and candidate arm $a$, define
\begin{equation}
\label{eq:simulation-ps-score}
\Psi_i(a)
=
\sum_{r=1}^2
\sum_{b=0}^3
\left|
n_{b,r,Z_{ir}}^{(i-1)}
+\ind(a=b)
-
\frac{
n_{r,Z_{ir}}^{(i-1)}+1
}{4}
\right|,
\end{equation}
where $n_{b,r,l}^{(i-1)}$ and $n_{r,l}^{(i-1)}$ are the previous
arm-specific and total counts at level $l$ of factor $r$. If fewer
than four arms minimize $\Psi_i(a)$, total probability $0.80$ is
divided equally among the minimizing arms and probability $0.20$ is
divided equally among the remaining arms. A complete tie is resolved
by equal randomization.

The same minimization rule is used when estimating $\bm\Omega_D$ by design
simulation.

Treatment assignment under multi-arm Hu--Hu randomization uses the same
enrollment order and all three balance scales. Let
$n_b^{(i-1)}$, $n_b^{(i-1)}(s)$, and
$n_{b,r,l}^{(i-1)}$ denote the previous overall, joint-stratum, and
factor-level counts in arm $b$. For incoming subject $i$ and candidate arm
$a$, define
\begin{align}
\label{eq:simulation-hh-score}
\Psi_i^{\mathrm{HH}}(a)
={}&
0.30\sum_{b=0}^3
\left[
n_b^{(i-1)}+\ind(a=b)-\frac{i}{4}
\right]^2 \notag\\
&+
0.50\sum_{b=0}^3
\left[
n_b^{(i-1)}(Z_i)+\ind(a=b)
-\frac{n^{(i-1)}(Z_i)+1}{4}
\right]^2 \notag\\
&+
0.10\sum_{r=1}^2\sum_{b=0}^3
\left[
n_{b,r,Z_{ir}}^{(i-1)}+\ind(a=b)
-\frac{n_{r,Z_{ir}}^{(i-1)}+1}{4}
\right]^2 .
\end{align}
The weights are the representative two-binary-covariate specification of
\citet{hu2012asymptotic}, written in the order overall, joint stratum, and
the two margins, and satisfy $0.30+0.50+0.10+0.10=1$. They are fixed before
outcome generation. Candidate arms are ranked from the smallest to largest
value of $\Psi_i^{\mathrm{HH}}(a)$ and receive probabilities
\[
\left(0.80,\frac{1}{15},\frac{1}{15},\frac{1}{15}\right).
\]
Within a tied block, the probabilities of the occupied ranks are averaged,
which is equivalent to randomly ordering the tied arms before applying the
ranked probability vector. This is the treatment-label-symmetric multi-arm
construction of \citet{hu2023multi}. Since the within-stratum weight is
positive, the joint-stratum imbalances are bounded in probability and their
$N^{-1/2}$-scaled covariance limit is zero. Consequently, no design-only
simulation is required for Hu--Hu randomization.

Scenarios WAN and EBL keep arms 2 and 3 of P-PH and change only the shape
of the arm-1 effect. In WAN the arm-1 hazard multiplier is $0.551345$
through month 6 and 1 afterward, and in EBL it is $0.497230$ through month
6 and $1.20$ afterward. The early multipliers were obtained by root finding
so that the true 12-month RMST gain of arm 1 equals that of P-PH, which is
1.075 months, and a unit test checks this equality to four decimal places.
The three scenarios therefore share the confirmatory effect size, while the
true arm-1 gain at $\tau=6$ is 0.341, 0.444, and 0.502 months under P-PH,
WAN, and EBL. Among the 80,000 standardized fits in each of WAN and EBL,
162 (0.20\%) and 273 (0.34\%) had an empty treatment--stratum risk set at
$\tau_U$. The reported rates are computed among the remaining trials, and
counting the flagged trials as nonrejections would lower the M3 power by
at most 0.006.

\subsection{Graded-interaction experiment}
\label{sec:supp-sim-gradient}

The graded-interaction experiment uses the global null, marginally
independent censoring calibrated to about 30\% for each value of $c$, simple
randomization and Pocock--Simon minimization, and $N\in\{800,1200\}$. For
interaction weight $c$, the control log-hazard coefficients are
$(0.45+0.75c,\ 0.65+0.55c,\ -2.4c)$ for $Z_{i1}$, $Z_{i2}$, and
$Z_{i1}Z_{i2}$. The non-additive share of $m_a(\cdot,12)$ reported in the
main paper is $\{\sum_sp_s\eta_a(s,12)^2\}^{1/2}/
\{\sum_sp_sm_a(s,12)^2\}^{1/2}$, where $\eta_a$ is the remainder of the
$\bm p$-weighted least-squares projection of $m_a(\cdot,12)$ onto additive
functions. It was computed by numerical integration of the influence
function under the true stratum hazards and the calibrated censoring
distribution. Table~\ref{tab:supp-interaction-gradient} gives the
fixed-time and process-wise family-wise error rates and the standard-error
ratios of M1, M2, M5, and M3. The process-wise rates follow the fixed-time
pattern. Under minimization, the process-wise error rate of M5 rose from
0.024--0.028 at $c=0$ to 0.033--0.034 at $c=1$, while those of M2 and M3
stayed between 0.022 and 0.029.

\begin{table}[t!]
\centering
\scriptsize
\caption{Graded-interaction null experiment. The interaction weight $c$ interpolates between the additive model ($c=0$) and G0-I ($c=1$). Each cell gives the fixed-time/process-wise family-wise error rate, and the SE column gives the ratio of the average estimated standard error of the 12-month contrasts to their empirical standard deviation. The first and second lines of each cell correspond to $N=800$ and $N=1200$. Each configuration uses 10,000 trials and 5,000 Gaussian draws, and PS uses 20,000 design-only replicates.}
\label{tab:supp-interaction-gradient}
\begin{tabular}{llcccccccc}
\toprule
&&\multicolumn{2}{c}{M1}&\multicolumn{2}{c}{M2}&\multicolumn{2}{c}{M5}&\multicolumn{2}{c}{M3}\\
\cmidrule(lr){3-4}\cmidrule(lr){5-6}\cmidrule(lr){7-8}\cmidrule(lr){9-10}
Design & $c$ & FWER & SE & FWER & SE & FWER & SE & FWER & SE\\
\midrule
PS & 0.00 & \shortstack{0.021/0.025\\0.019/0.021} & \shortstack{1.031\\1.039} & \shortstack{0.025/0.028\\0.025/0.024} & \shortstack{1.000\\1.007} & \shortstack{0.026/0.028\\0.026/0.024} & \shortstack{0.997\\1.005} & \shortstack{0.025/0.029\\0.025/0.025} & \shortstack{0.998\\1.006}\\
 & 0.25 & \shortstack{0.023/0.025\\0.023/0.025} & \shortstack{1.018\\1.021} & \shortstack{0.025/0.025\\0.026/0.028} & \shortstack{1.003\\1.004} & \shortstack{0.027/0.026\\0.027/0.028} & \shortstack{0.996\\0.999} & \shortstack{0.025/0.026\\0.028/0.027} & \shortstack{1.001\\1.002}\\
 & 0.50 & \shortstack{0.027/0.028\\0.026/0.026} & \shortstack{1.011\\1.010} & \shortstack{0.026/0.028\\0.027/0.027} & \shortstack{1.006\\1.004} & \shortstack{0.031/0.031\\0.031/0.029} & \shortstack{0.984\\0.983} & \shortstack{0.027/0.028\\0.026/0.026} & \shortstack{1.002\\1.004}\\
 & 0.75 & \shortstack{0.027/0.028\\0.025/0.027} & \shortstack{0.998\\1.006} & \shortstack{0.027/0.028\\0.025/0.026} & \shortstack{1.000\\1.007} & \shortstack{0.036/0.034\\0.034/0.034} & \shortstack{0.951\\0.959} & \shortstack{0.026/0.028\\0.028/0.027} & \shortstack{0.995\\1.003}\\
 & 1.00 & \shortstack{0.022/0.025\\0.025/0.023} & \shortstack{1.007\\1.008} & \shortstack{0.021/0.023\\0.024/0.022} & \shortstack{1.012\\1.013} & \shortstack{0.036/0.033\\0.040/0.034} & \shortstack{0.927\\0.928} & \shortstack{0.021/0.022\\0.025/0.025} & \shortstack{1.013\\1.007}\\
\addlinespace
SR & 0.00 & \shortstack{0.026/0.027\\0.025/0.028} & \shortstack{0.995\\1.004} & \shortstack{0.025/0.026\\0.024/0.027} & \shortstack{1.001\\1.007} & \shortstack{0.034/0.032\\0.031/0.031} & \shortstack{0.962\\0.970} & \shortstack{0.027/0.028\\0.025/0.026} & \shortstack{0.996\\1.002}\\
 & 0.25 & \shortstack{0.024/0.026\\0.026/0.028} & \shortstack{0.999\\1.005} & \shortstack{0.023/0.025\\0.026/0.027} & \shortstack{1.005\\1.009} & \shortstack{0.027/0.028\\0.030/0.030} & \shortstack{0.976\\0.983} & \shortstack{0.023/0.025\\0.026/0.028} & \shortstack{1.001\\1.007}\\
 & 0.50 & \shortstack{0.024/0.024\\0.024/0.024} & \shortstack{1.008\\0.999} & \shortstack{0.022/0.023\\0.023/0.024} & \shortstack{1.014\\1.003} & \shortstack{0.027/0.027\\0.027/0.029} & \shortstack{0.980\\0.971} & \shortstack{0.026/0.026\\0.024/0.024} & \shortstack{1.009\\0.999}\\
 & 0.75 & \shortstack{0.025/0.025\\0.025/0.025} & \shortstack{1.002\\0.994} & \shortstack{0.025/0.025\\0.025/0.024} & \shortstack{1.008\\0.997} & \shortstack{0.034/0.031\\0.035/0.032} & \shortstack{0.954\\0.946} & \shortstack{0.025/0.025\\0.025/0.026} & \shortstack{1.004\\0.992}\\
 & 1.00 & \shortstack{0.028/0.028\\0.024/0.026} & \shortstack{0.997\\0.998} & \shortstack{0.027/0.026\\0.024/0.025} & \shortstack{1.002\\1.001} & \shortstack{0.043/0.040\\0.040/0.037} & \shortstack{0.916\\0.917} & \shortstack{0.027/0.026\\0.024/0.027} & \shortstack{0.998\\0.997}\\
\addlinespace
\bottomrule
\end{tabular}
\end{table}

\subsection{Estimation under stratum-dependent censoring}
\label{sec:supp-sim-estimation}

Table~\ref{tab:supp-conditional-bias} gives the estimation results under
G0-C, in which censoring is independent of the event times only within
strata. The unstratified Kaplan--Meier contrasts had a maximum absolute bias
of 0.075 to 0.094 that did not decrease from $N=800$ to $N=1200$, which
reflects a failure of identification rather than a small-sample effect. The
standardized contrasts had a maximum absolute bias of at most 0.009,
standard-error ratios of 0.995 to 1.009, and pointwise coverage of 0.948 to
0.953, and the proportion of trials with adequate treatment--stratum risk
sets through $\tau^\star$ was 0.976 to 0.998. The bias of the unstratified
contrasts also raised the error rates of the unstratified procedures above
the nominal level. Their fixed-time family-wise error rate was 0.029 to 0.034
for M1 and 0.033 to 0.038 for M2, compared with 0.024 to 0.028 for M3 in
Table~\ref{tab:simulation-error-by-method} of the main paper. M2 was affected more than M1 because its standard
error is correct, whereas the conservative standard error of M1 partly
masks the bias.

\begin{table}[p!]
\centering
\small
\caption{Estimation under stratum-dependent censoring (G0-C). Bias columns
give the maximum absolute bias across the three treatment-control RMST
contrasts at $\tau^\star=12$. The final three columns refer to the
standardized analysis, and ``stable'' is the proportion of trials with adequate
treatment-by-stratum risk sets through $\tau^\star$. Marginal-KM bias is
estimated from 5,000 trials per configuration for SR, SPB, and PS and from
20,000 trials for HH, and the FWER columns give the fixed-time step-down
family-wise error rates of M1 and M2 from the same runs. All standardized
results use the formal 20,000-trial runs. M1 and M2 have the same marginal
point estimator.}
\label{tab:supp-conditional-bias}
\resizebox{\textwidth}{!}{%
\begin{tabular}{cclcccccc}
\toprule
$N$ & Design & \multicolumn{2}{c}{Maximum absolute bias}
& \multicolumn{2}{c}{FWER} & SE/SD range & Coverage range & Stable\\
\cmidrule(lr){3-4}\cmidrule(lr){5-6}
&& Marginal KM & Standardized & M1 & M2 &&&\\
\midrule
800  & SR  & 0.082 & 0.002 & 0.034 & 0.036 & 1.001--1.007 & 0.949--0.951 & 0.976\\
     & SPB & 0.075 & 0.002 & 0.029 & 0.035 & 1.003--1.006 & 0.950--0.953 & 0.982\\
     & PS  & 0.083 & 0.001 & 0.032 & 0.036 & 1.001--1.006 & 0.949--0.950 & 0.980\\
     & HH  & 0.094 & 0.009 & 0.031 & 0.037 & 1.001--1.006 & 0.950--0.951 & 0.980\\
1200 & SR  & 0.079 & 0.002 & 0.034 & 0.034 & 1.001--1.005 & 0.950--0.951 & 0.997\\
     & SPB & 0.084 & 0.002 & 0.029 & 0.036 & 0.998--0.999 & 0.948--0.950 & 0.998\\
     & PS  & 0.084 & 0.003 & 0.029 & 0.033 & 0.995--1.002 & 0.949--0.951 & 0.997\\
     & HH  & 0.088 & 0.005 & 0.032 & 0.038 & 1.005--1.009 & 0.949--0.951 & 0.998\\
\bottomrule
\end{tabular}}
\end{table}

\subsection{Operating characteristics of the RMST procedures}
\label{sec:supp-sim-ranges}

Table~\ref{tab:supp-alternative-ranges} gives the complete range of
operating characteristics across simple randomization, stratified permuted
blocks, Pocock--Simon minimization, and multi-arm Hu--Hu randomization for
each alternative scenario. All entries use the finite-sample-corrected
implementation. The
unstratified methods are included as misspecified comparators under
stratum-dependent censoring. Their apparently higher power should not be
interpreted as greater efficiency because their RMST contrasts can be
asymptotically biased. The standardized analysis is the primary valid
analysis for these scenarios.

\begin{table}[p!]
\centering
\scriptsize
\caption{Ranges across the four randomization procedures under alternative
scenarios. Results use the step-down procedure. DP and CP denote
disjunctive and conjunctive power, and ``Band'' is simultaneous-band coverage.
There is no true null under A-PH. Each configuration uses 10,000 trials,
5,000 multiplier draws per trial, and, for Pocock--Simon minimization,
20,000 design-only randomization replicates. The unstratified analyses are
misspecified comparators under conditional censoring.}
\label{tab:supp-alternative-ranges}
\resizebox{\textwidth}{!}{
\begin{tabular}{lclrrrrrrr}
\toprule
&&&\multicolumn{3}{c}{Fixed time}&\multicolumn{3}{c}{RMST process}&\\
\cmidrule(lr){4-6}\cmidrule(lr){7-9}
Scenario & $N$ & Analysis & FWER & DP & CP & FWER & DP & CP & Band\\
\midrule
P-PH & 800 & Naive & 0.019--0.025 & 0.681--0.692 & 0.208--0.218 & 0.019--0.021 & 0.599--0.613 & 0.157--0.168 & 0.945--0.952 \\
P-PH & 800 & CAR-adjusted & 0.024--0.025 & 0.675--0.712 & 0.212--0.238 & 0.019--0.022 & 0.593--0.636 & 0.162--0.182 & 0.943--0.948 \\
P-PH & 800 & Standardized & 0.020--0.022 & 0.633--0.645 & 0.188--0.201 & 0.020--0.021 & 0.552--0.570 & 0.144--0.155 & 0.946--0.949 \\
P-PH & 1200 & Naive & 0.021--0.026 & 0.854--0.864 & 0.331--0.340 & 0.020--0.022 & 0.798--0.809 & 0.267--0.269 & 0.947--0.953 \\
P-PH & 1200 & CAR-adjusted & 0.025--0.027 & 0.855--0.877 & 0.336--0.361 & 0.023--0.024 & 0.796--0.825 & 0.265--0.291 & 0.942--0.948 \\
P-PH & 1200 & Standardized & 0.021--0.023 & 0.814--0.824 & 0.299--0.306 & 0.020--0.022 & 0.752--0.761 & 0.237--0.246 & 0.946--0.950 \\
\midrule
A-PH & 800 & Naive & -- & 0.708--0.713 & 0.126--0.135 & -- & 0.626--0.637 & 0.090--0.097 & 0.951--0.961 \\
A-PH & 800 & CAR-adjusted & -- & 0.700--0.736 & 0.132--0.149 & -- & 0.619--0.661 & 0.094--0.108 & 0.946--0.954 \\
A-PH & 800 & Standardized & -- & 0.664--0.675 & 0.117--0.125 & -- & 0.586--0.604 & 0.085--0.091 & 0.947--0.956 \\
A-PH & 1200 & Naive & -- & 0.876--0.886 & 0.221--0.230 & -- & 0.821--0.832 & 0.165--0.175 & 0.942--0.955 \\
A-PH & 1200 & CAR-adjusted & -- & 0.875--0.900 & 0.220--0.249 & -- & 0.822--0.850 & 0.164--0.191 & 0.944--0.949 \\
A-PH & 1200 & Standardized & -- & 0.845--0.855 & 0.196--0.208 & -- & 0.784--0.793 & 0.150--0.161 & 0.949--0.952 \\
\midrule
DEL & 800 & Naive & 0.015--0.017 & 0.225--0.239 & 0.047--0.053 & 0.011--0.016 & 0.171--0.186 & 0.028--0.032 & 0.945--0.954 \\
DEL & 800 & CAR-adjusted & 0.016--0.020 & 0.234--0.252 & 0.053--0.061 & 0.014--0.018 & 0.182--0.197 & 0.031--0.037 & 0.944--0.947 \\
DEL & 800 & Standardized & 0.014--0.017 & 0.204--0.209 & 0.037--0.043 & 0.012--0.015 & 0.158--0.169 & 0.022--0.024 & 0.945--0.950 \\
DEL & 1200 & Naive & 0.017--0.021 & 0.336--0.345 & 0.083--0.086 & 0.015--0.016 & 0.266--0.274 & 0.052--0.055 & 0.946--0.955 \\
DEL & 1200 & CAR-adjusted & 0.020--0.022 & 0.338--0.370 & 0.085--0.097 & 0.016--0.019 & 0.267--0.295 & 0.052--0.064 & 0.944--0.948 \\
DEL & 1200 & Standardized & 0.015--0.018 & 0.299--0.311 & 0.059--0.068 & 0.014--0.016 & 0.240--0.253 & 0.036--0.041 & 0.947--0.950 \\
\midrule
CRS & 800 & Naive & 0.011--0.017 & 0.204--0.221 & 0.028--0.032 & 0.012--0.016 & 0.160--0.176 & 0.015--0.019 & 0.945--0.956 \\
CRS & 800 & CAR-adjusted & 0.014--0.017 & 0.217--0.232 & 0.031--0.033 & 0.014--0.015 & 0.169--0.181 & 0.018--0.020 & 0.946--0.949 \\
CRS & 800 & Standardized & 0.011--0.015 & 0.192--0.197 & 0.022--0.023 & 0.012--0.013 & 0.152--0.154 & 0.012--0.013 & 0.948--0.952 \\
CRS & 1200 & Naive & 0.012--0.017 & 0.303--0.316 & 0.035--0.041 & 0.012--0.016 & 0.241--0.252 & 0.019--0.023 & 0.946--0.958 \\
CRS & 1200 & CAR-adjusted & 0.016--0.017 & 0.311--0.337 & 0.039--0.048 & 0.014--0.017 & 0.248--0.270 & 0.023--0.027 & 0.944--0.950 \\
CRS & 1200 & Standardized & 0.011--0.014 & 0.278--0.286 & 0.026--0.030 & 0.012--0.015 & 0.223--0.231 & 0.014--0.018 & 0.946--0.950 \\
\midrule
WAN & 800 & Naive & 0.017--0.025 & 0.704--0.716 & 0.212--0.222 & 0.016--0.021 & 0.696--0.710 & 0.168--0.177 & 0.947--0.956 \\
WAN & 800 & CAR-adjusted & 0.021--0.023 & 0.703--0.740 & 0.217--0.235 & 0.019--0.022 & 0.694--0.725 & 0.171--0.188 & 0.946--0.949 \\
WAN & 800 & Standardized & 0.018--0.020 & 0.651--0.660 & 0.190--0.194 & 0.018--0.020 & 0.661--0.676 & 0.157--0.161 & 0.949--0.953 \\
WAN & 1200 & Naive & 0.022--0.025 & 0.878--0.882 & 0.339--0.343 & 0.021--0.023 & 0.870--0.880 & 0.273--0.278 & 0.950--0.955 \\
WAN & 1200 & CAR-adjusted & 0.025--0.028 & 0.878--0.897 & 0.337--0.368 & 0.024--0.026 & 0.869--0.890 & 0.275--0.300 & 0.944--0.951 \\
WAN & 1200 & Standardized & 0.021--0.025 & 0.840--0.841 & 0.306--0.315 & 0.021--0.023 & 0.851--0.853 & 0.252--0.259 & 0.947--0.953 \\
\midrule
EBL & 800 & Naive & 0.018--0.024 & 0.721--0.733 & 0.214--0.226 & 0.018--0.022 & 0.785--0.794 & 0.176--0.185 & 0.942--0.956 \\
EBL & 800 & CAR-adjusted & 0.021--0.024 & 0.714--0.755 & 0.215--0.246 & 0.019--0.022 & 0.782--0.809 & 0.179--0.201 & 0.945--0.950 \\
EBL & 800 & Standardized & 0.019--0.021 & 0.660--0.678 & 0.193--0.205 & 0.017--0.020 & 0.762--0.770 & 0.166--0.178 & 0.946--0.953 \\
EBL & 1200 & Naive & 0.022--0.025 & 0.889--0.899 & 0.341--0.348 & 0.019--0.023 & 0.930--0.940 & 0.277--0.287 & 0.946--0.953 \\
EBL & 1200 & CAR-adjusted & 0.024--0.027 & 0.886--0.912 & 0.339--0.373 & 0.021--0.025 & 0.928--0.945 & 0.279--0.310 & 0.943--0.948 \\
EBL & 1200 & Standardized & 0.021--0.022 & 0.851--0.856 & 0.308--0.318 & 0.020--0.024 & 0.919--0.926 & 0.257--0.269 & 0.947--0.950 \\
\bottomrule
\end{tabular}}
\end{table}

\subsection{Stratified log-rank benchmark}
\label{sec:supp-logrank}

For M6, let $Y_{as}(t)$ and $d_{as}(t)$ denote the risk-set size and number
of failures at time $t$ in arm $a$ and stratum $s$. For treatment arm
$a\geq1$, the benefit-oriented score and its conventional null variance are
\begin{align}
U_a
&=
\sum_s\sum_{t\leq\tau^\star}
\left[
\frac{Y_{as}(t)}
{Y_{0s}(t)+Y_{as}(t)}
\{d_{0s}(t)+d_{as}(t)\}
-d_{as}(t)
\right],
\label{eq:logrank-score}\\
V_a
&=
\sum_s\sum_{t\leq\tau^\star}
\frac{
Y_{0s}(t)Y_{as}(t)d_s^{(a)}(t)
\{Y_s^{(a)}(t)-d_s^{(a)}(t)\}
}{
\{Y_s^{(a)}(t)\}^2\{Y_s^{(a)}(t)-1\}
},
\label{eq:logrank-variance}
\end{align}
where $Y_s^{(a)}=Y_{0s}+Y_{as}$ and
$d_s^{(a)}=d_{0s}+d_{as}$. Thus $Z_a^{\mathrm{LR}}=U_a/\sqrt{V_a}$ is
positive when treatment has fewer failures than expected under equality of
the two survival curves. At a general analysis horizon $h$, the sums run
over $t\leq h$ instead of $t\leq\tau^\star$.

The shared-control correlation is estimated without replacing the
conventional hypergeometric studentization. Specifically, define the
subject-level event-score contribution at analysis horizon $h$ by
\[
\widehat\psi_{ia}(h)
=
\sum_s\sum_{t\leq h}
\ind\{X_i=t,\Delta_i=1,\bm Z_i=s,\ A_i\in\{0,a\}\}
\left\{
\frac{Y_{as}(t)}{Y_{0s}(t)+Y_{as}(t)}
-\ind(A_i=a)
\right\}.
\]
Let $\widehat{\bm C}(h)$ be the empirical covariance matrix obtained from
the centered vectors
$\widehat{\bm\psi}_i(h)
=(\widehat\psi_{i1}(h),\widehat\psi_{i2}(h),
\widehat\psi_{i3}(h))^T$, and let
$\widehat{\bm R}_{\mathrm{LR}}(h)$ be its correlation matrix. The covariance
used for max-$T$ calibration is the hybrid matrix
\[
\widehat{\bm\Sigma}_{\mathrm{LR}}(h)
=
\operatorname{diag}\{\sqrt{V_1(h)},\ldots,\sqrt{V_K(h)}\}
\widehat{\bm R}_{\mathrm{LR}}(h)
\operatorname{diag}\{\sqrt{V_1(h)},\ldots,\sqrt{V_K(h)}\}.
\]
Thus each marginal statistic retains its ordinary null log-rank variance,
while the off-diagonal entries capture the fact that a control-arm failure
contributes simultaneously to every treatment--control comparison. The
single-step and step-down critical values are calculated from centered
Gaussian draws with correlation
$\widehat{\bm R}_{\mathrm{LR}}(h)$. Under a partial null, only the joint
limit of the true-null subset is required, and using the same Gaussian draw for
all subsets makes the step-down critical values subset monotone.

Under a true pairwise null, censoring that is independent conditionally on
the joint strata gives a mean-zero stratified score. The true-null subset
therefore remains valid across the four randomization procedures, and the
subset-monotone step-down critical values provide strong FWER control.

The main comparison censors M6 administratively at 12 months to match
$\tau^\star$. Table~\ref{tab:supp-logrank-sensitivity} compares it with an
analysis that uses all follow-up through 24 months. Both horizons were
evaluated on the same simulated trials as the RMST procedures, so their
differences contain no between-run Monte Carlo noise. Type I error remained
controlled at both horizons. Using all
follow-up materially increased power when the benefit began after month 6
in DEL and when the early and late effects opposed one another in CRS. It
reduced power in WAN and EBL, because events after month 12 carry no
benefit in WAN and an excess hazard in EBL, and both dilute the early
log-rank signal.

\begin{table}[p!]
\centering
\scriptsize
\caption{Sensitivity analysis for the stratified log-rank benchmark (M6) using follow-up through 12 or 24 months. Each entry is 12-month/24-month operating performance, and the first and second lines correspond to $N=800$ and $N=1200$.}
\label{tab:supp-logrank-sensitivity}
\begin{tabular}{llcccc}
\toprule
Scenario & Quantity & SR & SPB & PS & HH\\
\midrule
G0-M & FWER & \shortstack{0.025/0.025\\0.024/0.026} & \shortstack{0.027/0.028\\0.027/0.025} & \shortstack{0.025/0.026\\0.027/0.026} & \shortstack{0.025/0.025\\0.027/0.027}\\
G0-C & FWER & \shortstack{0.025/0.026\\0.026/0.024} & \shortstack{0.023/0.024\\0.028/0.027} & \shortstack{0.027/0.025\\0.023/0.025} & \shortstack{0.027/0.026\\0.026/0.025}\\
G0-I & FWER & \shortstack{0.028/0.028\\0.024/0.025} & \shortstack{0.027/0.029\\0.025/0.025} & \shortstack{0.026/0.025\\0.026/0.027} & \shortstack{0.026/0.028\\0.027/0.026}\\
P-PH & Strong FWER & \shortstack{0.022/0.024\\0.022/0.024} & \shortstack{0.022/0.023\\0.023/0.024} & \shortstack{0.023/0.023\\0.021/0.020} & \shortstack{0.021/0.022\\0.022/0.026}\\
 & Disjunctive power & \shortstack{0.703/0.855\\0.879/0.966} & \shortstack{0.714/0.866\\0.876/0.962} & \shortstack{0.712/0.861\\0.876/0.965} & \shortstack{0.699/0.855\\0.874/0.963}\\
A-PH & Disjunctive power & \shortstack{0.735/0.880\\0.900/0.976} & \shortstack{0.744/0.881\\0.904/0.976} & \shortstack{0.739/0.884\\0.899/0.973} & \shortstack{0.742/0.885\\0.897/0.975}\\
DEL & Strong FWER & \shortstack{0.017/0.024\\0.021/0.024} & \shortstack{0.019/0.023\\0.021/0.026} & \shortstack{0.019/0.023\\0.022/0.025} & \shortstack{0.016/0.020\\0.021/0.024}\\
 & Disjunctive power & \shortstack{0.312/0.661\\0.442/0.840} & \shortstack{0.309/0.650\\0.451/0.844} & \shortstack{0.309/0.660\\0.456/0.848} & \shortstack{0.307/0.655\\0.452/0.844}\\
CRS & Strong FWER & \shortstack{0.017/0.022\\0.020/0.026} & \shortstack{0.016/0.023\\0.018/0.023} & \shortstack{0.017/0.020\\0.017/0.024} & \shortstack{0.014/0.021\\0.018/0.024}\\
 & Disjunctive power & \shortstack{0.283/0.609\\0.403/0.798} & \shortstack{0.276/0.608\\0.401/0.795} & \shortstack{0.282/0.606\\0.407/0.796} & \shortstack{0.276/0.609\\0.402/0.791}\\
WAN & Strong FWER & \shortstack{0.020/0.021\\0.024/0.021} & \shortstack{0.023/0.023\\0.024/0.023} & \shortstack{0.021/0.022\\0.021/0.022} & \shortstack{0.019/0.020\\0.024/0.021}\\
 & Disjunctive power & \shortstack{0.503/0.444\\0.693/0.624} & \shortstack{0.512/0.458\\0.686/0.631} & \shortstack{0.494/0.446\\0.687/0.629} & \shortstack{0.502/0.446\\0.695/0.629}\\
EBL & Strong FWER & \shortstack{0.019/0.018\\0.021/0.019} & \shortstack{0.017/0.017\\0.020/0.021} & \shortstack{0.021/0.020\\0.023/0.019} & \shortstack{0.021/0.018\\0.021/0.020}\\
 & Disjunctive power & \shortstack{0.402/0.336\\0.573/0.497} & \shortstack{0.403/0.344\\0.582/0.498} & \shortstack{0.404/0.339\\0.573/0.491} & \shortstack{0.413/0.352\\0.577/0.494}\\
\bottomrule
\end{tabular}
\begin{minipage}{0.98\textwidth}
\footnotesize
The 12-month analysis is the horizon-matched main-text benchmark, and 24 months uses all available follow-up. M6 tests equality of survival curves over the selected interval, not an RMST contrast. Both horizons use the same trials. G0-M and G0-C use 20,000 trials and 10,000 Gaussian draws per trial, and all other rows use 10,000 trials and 5,000 draws. There were no fitting failures.
\end{minipage}
\end{table}

\subsection{Stratified Cox benchmark}
\label{sec:supp-cox}

Table~\ref{tab:supp-cox-ranges} summarizes the secondary stratified Cox
benchmark. To make its power comparison with fixed-time RMST meaningful,
event follow-up was administratively censored at $\tau^\star=12$ before
fitting the Cox model. Thus, the two analyses use the same generated trials
and the same time horizon, although they target different estimands. The
Cox model was stratified by the four joint randomization strata and included
the three treatment indicators. Its max-$T$ procedure used the complete
model-based covariance matrix of the negative fitted log hazard ratios.

\begin{table}[p!]
\centering
\small
\caption{Ranges of stratified Cox operating characteristics. The
global-null rows cover simple randomization, stratified permuted blocks,
and Pocock--Simon minimization, and the P-PH and A-PH rows additionally cover
multi-arm Hu--Hu randomization. All analyses are administratively censored
at 12 months. FWER is global-null FWER under G0-M and G0-C and strong FWER
under P-PH. DP and CP denote disjunctive and conjunctive power. Maximum
absolute bias and SE/SD are calculated across treatment arms and included
randomization procedures for the negative log hazard-ratio estimators. The
secondary global-null benchmark uses 5,000 trials and 2,000 multiplier
draws per configuration. The matched P-PH and A-PH benchmarks use 10,000
trials and 5,000 draws.}
\label{tab:supp-cox-ranges}
\begin{tabular}{lcrrrrr}
\toprule
Scenario & $N$ & FWER & DP & CP & Maximum $|\mathrm{bias}|$ & SE/SD\\
\midrule
G0-M & 800  & 0.021--0.027 & --          & --          & 0.001 & 0.987--1.016\\
     & 1200 & 0.021--0.024 & --          & --          & 0.002 & 0.993--1.014\\
G0-C & 800  & 0.022--0.024 & --          & --          & 0.004 & 0.982--1.007\\
     & 1200 & 0.022--0.024 & --          & --          & 0.004 & 0.983--1.014\\
\midrule
P-PH & 800  & 0.021--0.023 & 0.690--0.708 & 0.222--0.229 & 0.006 & 0.984--1.007\\
     & 1200 & 0.021--0.023 & 0.873--0.877 & 0.353--0.360 & 0.001 & 0.973--1.008\\
A-PH & 800  & --          & 0.728--0.734 & 0.142--0.149 & 0.002 & 0.991--1.004\\
     & 1200 & --          & 0.892--0.899 & 0.238--0.253 & 0.003 & 0.985--1.014\\
\bottomrule
\end{tabular}
\end{table}

\begin{table}[p!]
\centering
\small
\caption{Comparison of standardized fixed-time RMST and stratified Cox
step-down inference under Pocock--Simon minimization. Both analyses use the
same 5,000 global-null or 10,000 treatment-effect trials and the same
12-month analysis horizon. The corresponding multiplier counts are 2,000
and 5,000, respectively. FWER is
global-null FWER under G0-M and G0-C, strong FWER for the true null arm
under P-PH, and is not applicable under A-PH. DP and CP denote disjunctive
and conjunctive power.}
\label{tab:supp-cox-comparison}
\begin{tabular}{lcrrrrrr}
\toprule
&&\multicolumn{3}{c}{Standardized RMST}
&\multicolumn{3}{c}{Stratified Cox}\\
\cmidrule(lr){3-5}\cmidrule(lr){6-8}
Scenario & $N$ & FWER & DP & CP & FWER & DP & CP\\
\midrule
G0-M & 800  & 0.028 & --    & --    & 0.027 & --    & --\\
     & 1200 & 0.027 & --    & --    & 0.024 & --    & --\\
G0-C & 800  & 0.027 & --    & --    & 0.024 & --    & --\\
     & 1200 & 0.023 & --    & --    & 0.024 & --    & --\\
\midrule
P-PH & 800  & 0.022 & 0.643 & 0.201 & 0.023 & 0.704 & 0.227\\
     & 1200 & 0.022 & 0.824 & 0.305 & 0.023 & 0.873 & 0.357\\
A-PH & 800  & --    & 0.670 & 0.121 & --    & 0.730 & 0.145\\
     & 1200 & --    & 0.847 & 0.208 & --    & 0.894 & 0.253\\
\bottomrule
\end{tabular}
\end{table}

All 220,000 Cox fits represented in the two benchmark sets converged without
an analysis error. The estimated
standard errors tracked the empirical standard deviations closely in every
configuration. As expected under the correctly specified PH scenarios, Cox
inference was generally more powerful than fixed-time RMST inference.
Matching the analysis horizon ensures that this gain cannot be attributed
to Cox using additional follow-up beyond $\tau^\star$. A
constant-hazard-ratio benchmark is not reported for DEL or CRS because
those scenarios deliberately violate proportional hazards.

\subsection{Equal-RMST power-trend experiment}
\label{sec:supp-equal-rmst}

The primary alternatives in the main paper represent clinically familiar
hazard patterns, but apart from P-PH, WAN, and EBL their RMST gains are not
equal. We therefore added
parallel paired experiments under marginally and conditionally independent
censoring. The former compares both proposed RMST procedures, M2 and M3,
with M6, and the latter compares M3 with M6 because M2 does not identify the
target marginal RMST under that censoring condition. Arms 2 and 3 retain the
control hazard. For each target
\[
d\in\{0.50,0.75,1.00,1.25\},
\qquad
\theta_1(12)
=
\int_0^{12}\{S_1(t)-S_0(t)\}\,dt
=d,
\]
numerical integration over the four stratum-specific survival curves
calibrates four arm-1 profiles. The 0.50-month gain is the prespecified
moderate/local alternative, and the other targets quantify the power trend.
The delayed profile changes at month 3, which preserves a clinically
interpretable delay while leaving a feasible calibration range through a
1.25-month RMST gain. The treatment-limited-benefit (TLB) profile has a
reduced hazard through month 6 and the control hazard thereafter. The
immunotherapy-type crossing (IOX) profile has hazard multiplier 1.10 through
month 3 and a subsequently reduced hazard calibrated to the target RMST
gain.

\begin{table}[p!]
\centering
\small
\caption{\textbf{Marginally independent censoring.} Hazard-multiplier
calibration for the equal-RMST power-trend experiment. For PH, the
multiplier is constant. For DEL and IOX, $r_{\rm pre}$
applies through month 3, and for TLB it applies through month 6.
$U_{\rm LR}/N$ is the expected benefit-oriented pairwise stratified
log-rank score per randomized subject through month 12.}
\label{tab:supp-equal-rmst-marginal-calibration}
\begin{tabular}{cclrrr}
\toprule
$d$ & Profile & Change time & $r_{\rm pre}$ & $r_{\rm post}$ & $U_{\rm LR}/N$\\
\midrule
0.50 & PH  & -- & 0.8286 & 0.8286 & 0.01136\\
0.50 & DEL & 3  & 1.0000 & 0.6748 & 0.01404\\
0.50 & TLB & 6  & 0.7815 & 1.0000 & 0.00904\\
0.50 & IOX & 3  & 1.1000 & 0.5869 & 0.01558\\
\addlinespace
0.75 & PH  & -- & 0.7487 & 0.7487 & 0.01689\\
0.75 & DEL & 3  & 1.0000 & 0.5278 & 0.02080\\
0.75 & TLB & 6  & 0.6788 & 1.0000 & 0.01343\\
0.75 & IOX & 3  & 1.1000 & 0.4421 & 0.02231\\
\addlinespace
1.00 & PH  & -- & 0.6722 & 0.6722 & 0.02232\\
1.00 & DEL & 3  & 1.0000 & 0.3897 & 0.02738\\
1.00 & TLB & 6  & 0.5801 & 1.0000 & 0.01774\\
1.00 & IOX & 3  & 1.1000 & 0.3060 & 0.02886\\
\addlinespace
1.25 & PH  & -- & 0.5990 & 0.5990 & 0.02766\\
1.25 & DEL & 3  & 1.0000 & 0.2597 & 0.03376\\
1.25 & TLB & 6  & 0.4852 & 1.0000 & 0.02196\\
1.25 & IOX & 3  & 1.1000 & 0.1779 & 0.03522\\
\bottomrule
\end{tabular}
\end{table}

\begin{table}[p!]
\centering
\small
\caption{\textbf{Conditionally independent censoring given strata.}
Hazard-multiplier calibration for the equal-RMST power-trend experiment.
Notation follows Table~\ref{tab:supp-equal-rmst-marginal-calibration}.}
\label{tab:supp-equal-rmst-conditional-calibration}
\begin{tabular}{cclrrr}
\toprule
$d$ & Profile & Change time & $r_{\rm pre}$ & $r_{\rm post}$ & $U_{\rm LR}/N$\\
\midrule
0.50 & PH  & -- & 0.8286 & 0.8286 & 0.01115\\
0.50 & DEL & 3  & 1.0000 & 0.6748 & 0.01371\\
0.50 & TLB & 6  & 0.7815 & 1.0000 & 0.00892\\
0.50 & IOX & 3  & 1.1000 & 0.5869 & 0.01518\\
\addlinespace
0.75 & PH  & -- & 0.7487 & 0.7487 & 0.01660\\
0.75 & DEL & 3  & 1.0000 & 0.5278 & 0.02035\\
0.75 & TLB & 6  & 0.6788 & 1.0000 & 0.01326\\
0.75 & IOX & 3  & 1.1000 & 0.4421 & 0.02182\\
\addlinespace
1.00 & PH  & -- & 0.6722 & 0.6722 & 0.02196\\
1.00 & DEL & 3  & 1.0000 & 0.3897 & 0.02685\\
1.00 & TLB & 6  & 0.5801 & 1.0000 & 0.01752\\
1.00 & IOX & 3  & 1.1000 & 0.3060 & 0.02829\\
\addlinespace
1.25 & PH  & -- & 0.5990 & 0.5990 & 0.02724\\
1.25 & DEL & 3  & 1.0000 & 0.2597 & 0.03319\\
1.25 & TLB & 6  & 0.4852 & 1.0000 & 0.02170\\
1.25 & IOX & 3  & 1.1000 & 0.1779 & 0.03462\\
\bottomrule
\end{tabular}
\end{table}

The four practical hazard profiles are common across censoring conditions
because censoring does not change the RMST estimand. The IOX profile is a
crossing-survival-curve alternative in the clinically familiar direction of
a mild early disadvantage followed by sustained benefit, and it is not
calibrated to attenuate M6. The unrounded multipliers and numerical checks
are reproduced by
\texttt{calibrate\_equal\_rmst\_scenarios.py} with its
\texttt{--censoring} option.

For a separate methodological stress test, we retain the reverse-crossing
profiles used in the earlier numerical investigation. They have an early
hazard reduction followed by a hazard increase and satisfy two calibration
constraints: the specified positive RMST gain and zero expected
benefit-oriented M6 score. Unlike IOX, these profiles are not presented as a
representative clinical efficacy pattern. Because the second constraint
depends on censoring-weighted risk sets, their multipliers differ between
the two censoring experiments.

\begin{table}[p!]
\centering
\small
\caption{Calibration of the log-rank-orthogonal reverse-crossing stress
test. All profiles change at month 6 and satisfy
$E\{U_{\rm LR}\}/N=0$ through month 12.}
\label{tab:supp-equal-rmst-stress-calibration}
\begin{tabular}{clrrr}
\toprule
$d$ & Censoring & $r_{\rm pre}$ & $r_{\rm post}$ & $U_{\rm LR}/N$\\
\midrule
0.50 & Marginal    & 0.6035 & 1.6906 & 0.00000\\
0.50 & Conditional & 0.5979 & 1.7137 & 0.00000\\
0.75 & Marginal    & 0.4208 & 2.0201 & 0.00000\\
0.75 & Conditional & 0.4126 & 2.0550 & 0.00000\\
1.00 & Marginal    & 0.2473 & 2.3403 & 0.00000\\
1.00 & Conditional & 0.2368 & 2.3872 & 0.00000\\
1.25 & Marginal    & 0.0822 & 2.6519 & 0.00000\\
1.25 & Conditional & 0.0696 & 2.7109 & 0.00000\\
\bottomrule
\end{tabular}
\end{table}

Tables~\ref{tab:supp-equal-rmst-marginal} and
\ref{tab:supp-equal-rmst-conditional} report all four randomization
procedures under marginally and conditionally independent censoring,
respectively. All methods within a censoring experiment use identical trial
random-number streams. The 0.50-month target produces deliberately local
power, whereas the larger targets show how the procedures approach
conventional design power. Under marginally independent censoring, both M2
and M3 are valid estimators of the target RMST and are compared directly.
Under conditionally independent censoring, M2 is omitted because it targets a
censoring-distorted marginal KM quantity.

Under marginally independent censoring, M2 and M3 had very similar power
under the three balancing designs: their maximum absolute difference was
0.005 across SPB, PS, and HH, including the PS configurations displayed in
the main paper. Under simple randomization, M3 exceeded M2 by as much as
0.033 because standardization removes chance between-stratum imbalance.
Strong FWER over the practical marginal-censoring grid was 0.018--0.029 for
M2 and 0.019--0.030 for M3. Relative to M6, M3 gained 0.030--0.207 power
under treatment-limited benefit. M6 gained 0.020--0.089 under PH,
0.097--0.339 under delayed benefit, and 0.106--0.404 under
immunotherapy-type crossing, and M2 showed the same qualitative ordering.

Under conditionally independent censoring, M3 exceeded M6 by
0.035--0.207 under treatment-limited benefit. M6 exceeded M3 by
0.019--0.070 under PH, 0.097--0.324 under delayed benefit, and
0.111--0.393 under immunotherapy-type crossing. The IOX result is important:
crossing curves do not by themselves favor RMST. Here the mild early
disadvantage is followed by a sustained and sufficiently strong hazard
reduction, so the event-weighted log-rank score remains highly informative.
By contrast, after month 6 the TLB profile has equal treatment and control
hazards. The previously accumulated survival separation continues to enter
the RMST area, while no further directional hazard score accrues to M6.

For M6, strong FWER over the practical profiles was 0.020--0.031 under
marginal censoring and 0.020--0.028 under conditional censoring. The
marginal upper endpoint was
0.0312 (Monte Carlo standard error 0.0017) for delayed benefit with
$d=1.25$, $N=800$, and Pocock--Simon minimization. In that configuration,
the single-step and Holm error rates were 0.0230 and 0.0277, respectively,
and the step-down rate was 0.0238 at $N=1200$. The empirical standard
deviations of the two true-null log-rank statistics at $N=800$ were 1.020
and 1.018, indicating mild finite-sample under-studentization. Thus the
isolated step-down overrun does not contradict the asymptotic result, but it
does caution against interpreting the Gaussian calibration as exact in
finite samples.

\begin{table}[p!]
\centering
\scriptsize
\caption{\textbf{Marginally independent censoring.} Ranges of arm-1 power across the four randomization designs in the equal-RMST trend experiment.}
\label{tab:supp-equal-rmst-marginal}
\resizebox{\textwidth}{!}{%
\begin{tabular}{llrrrrrrrr}
\toprule
&& \multicolumn{4}{c}{$N=800$} & \multicolumn{4}{c}{$N=1200$}\\
\cmidrule(lr){3-6}\cmidrule(lr){7-10}
Effect shape & Method & $0.50^\dagger$ & $0.75$ & $1.00$ & $1.25$ & $0.50^\dagger$ & $0.75$ & $1.00$ & $1.25$\\
\midrule
PH & M2 & 0.118--0.131 & 0.282--0.310 & 0.509--0.546 & 0.750--0.773 & 0.179--0.194 & 0.425--0.463 & 0.710--0.750 & 0.911--0.928\\
 & M3 & 0.123--0.131 & 0.296--0.309 & 0.534--0.546 & 0.767--0.773 & 0.187--0.197 & 0.451--0.466 & 0.735--0.751 & 0.922--0.927\\
 & M6 & 0.143--0.153 & 0.349--0.356 & 0.610--0.634 & 0.841--0.844 & 0.223--0.229 & 0.523--0.538 & 0.809--0.830 & 0.961--0.963\\
\addlinespace
Delayed benefit & M2 & 0.113--0.128 & 0.258--0.292 & 0.475--0.513 & 0.702--0.731 & 0.172--0.192 & 0.410--0.433 & 0.680--0.708 & 0.875--0.901\\
 & M3 & 0.116--0.128 & 0.278--0.291 & 0.503--0.514 & 0.723--0.730 & 0.178--0.192 & 0.427--0.434 & 0.706--0.710 & 0.891--0.901\\
 & M6 & 0.226--0.247 & 0.532--0.563 & 0.835--0.845 & 0.972--0.975 & 0.353--0.373 & 0.749--0.758 & 0.958--0.962 & 0.997--0.998\\
\addlinespace
Treatment-limited benefit & M2 & 0.121--0.138 & 0.294--0.314 & 0.535--0.569 & 0.776--0.803 & 0.180--0.203 & 0.436--0.471 & 0.737--0.771 & 0.932--0.941\\
 & M3 & 0.129--0.137 & 0.302--0.316 & 0.559--0.569 & 0.800--0.804 & 0.192--0.204 & 0.461--0.474 & 0.764--0.771 & 0.940--0.944\\
 & M6 & 0.096--0.104 & 0.205--0.217 & 0.385--0.402 & 0.593--0.599 & 0.137--0.149 & 0.317--0.330 & 0.569--0.580 & 0.787--0.799\\
\addlinespace
Immunotherapy-type crossing & M2 & 0.112--0.124 & 0.259--0.280 & 0.468--0.502 & 0.688--0.718 & 0.167--0.190 & 0.397--0.426 & 0.665--0.704 & 0.871--0.892\\
 & M3 & 0.120--0.124 & 0.277--0.281 & 0.488--0.503 & 0.712--0.720 & 0.180--0.190 & 0.423--0.427 & 0.688--0.703 & 0.886--0.894\\
 & M6 & 0.297--0.305 & 0.630--0.631 & 0.888--0.893 & 0.981--0.985 & 0.448--0.461 & 0.825--0.830 & 0.979--0.980 & 0.999--1.000\\
\addlinespace
\bottomrule
\end{tabular}%
}
\begin{minipage}{0.98\textwidth}
\footnotesize
Each cell is the minimum--maximum rejection probability across the four designs. M2 and M3 target the 12-month RMST contrast; M6 tests equality of survival curves through month 12. $^\dagger$The 0.50-month gain is the prespecified moderate/local alternative. Results use 10,000 matched trials. M3 stability was at least 0.992; strong FWER ranged from 0.018 to 0.031.
\end{minipage}
\end{table}

\begin{table}[p!]
\centering
\scriptsize
\caption{\textbf{Conditionally independent censoring given strata.} Ranges of arm-1 power across the four randomization designs in the equal-RMST trend experiment.}
\label{tab:supp-equal-rmst-conditional}
\resizebox{\textwidth}{!}{%
\begin{tabular}{llrrrrrrrr}
\toprule
&& \multicolumn{4}{c}{$N=800$} & \multicolumn{4}{c}{$N=1200$}\\
\cmidrule(lr){3-6}\cmidrule(lr){7-10}
Effect shape & Method & $0.50^\dagger$ & $0.75$ & $1.00$ & $1.25$ & $0.50^\dagger$ & $0.75$ & $1.00$ & $1.25$\\
\midrule
PH & M3 & 0.121--0.128 & 0.291--0.303 & 0.529--0.541 & 0.759--0.772 & 0.190--0.196 & 0.446--0.455 & 0.731--0.744 & 0.919--0.925\\
 & M6 & 0.142--0.147 & 0.341--0.349 & 0.599--0.610 & 0.822--0.837 & 0.218--0.226 & 0.513--0.521 & 0.800--0.812 & 0.953--0.959\\
\addlinespace
Delayed benefit & M3 & 0.123--0.130 & 0.276--0.284 & 0.501--0.507 & 0.714--0.726 & 0.178--0.185 & 0.427--0.432 & 0.694--0.704 & 0.889--0.897\\
 & M6 & 0.220--0.233 & 0.523--0.534 & 0.821--0.826 & 0.964--0.969 & 0.339--0.353 & 0.728--0.735 & 0.949--0.954 & 0.997--0.998\\
\addlinespace
Treatment-limited benefit & M3 & 0.128--0.138 & 0.296--0.310 & 0.554--0.565 & 0.788--0.792 & 0.190--0.203 & 0.461--0.471 & 0.758--0.769 & 0.936--0.938\\
 & M6 & 0.093--0.101 & 0.205--0.215 & 0.378--0.392 & 0.580--0.596 & 0.131--0.143 & 0.321--0.331 & 0.564--0.576 & 0.788--0.791\\
\addlinespace
Immunotherapy-type crossing & M3 & 0.119--0.126 & 0.268--0.286 & 0.486--0.495 & 0.702--0.720 & 0.180--0.189 & 0.406--0.423 & 0.685--0.696 & 0.878--0.888\\
 & M6 & 0.283--0.289 & 0.600--0.615 & 0.869--0.874 & 0.978--0.980 & 0.434--0.441 & 0.799--0.804 & 0.969--0.976 & 0.998--0.999\\
\addlinespace
\bottomrule
\end{tabular}%
}
\begin{minipage}{0.98\textwidth}
\footnotesize
Each cell is the minimum--maximum rejection probability across the four designs. M3 targets the 12-month RMST contrast; M6 tests equality of survival curves through month 12. M2 is omitted because its marginal estimator does not identify that target under stratum-dependent censoring. $^\dagger$The 0.50-month gain is the prespecified moderate/local alternative. Results use 10,000 matched trials. M3 stability was at least 0.886; strong FWER ranged from 0.019 to 0.029.
\end{minipage}
\end{table}

Tables~\ref{tab:supp-equal-rmst-marginal-stress} and
\ref{tab:supp-equal-rmst-conditional-stress} report the separate
log-rank-orthogonal reverse-crossing stress test. Its near-zero M6 power is a
consequence of the calibration, not a numerical failure. In the
deterministic risk-set limit, the benefit-oriented score has the form
\[
u_{\rm LR}
=
\sum_s p_s\int_0^{12}
w_s(t)\{\lambda_{0s}(t)-\lambda_{1s}(t)\}\,dt,
\qquad w_s(t)>0.
\]
The early period has $\lambda_{1s}(t)<\lambda_{0s}(t)$ and contributes
positively, whereas the post-crossing period has
$\lambda_{1s}(t)>\lambda_{0s}(t)$ and contributes negatively. The stress
profiles set these contributions to sum to zero. Consequently,
$E\{U_{\rm LR}\}/N=u_{\rm LR}=0$, so the standardized log-rank statistic has
no $\sqrt N$ location drift and is not consistent against this particular
alternative. At each larger RMST target, both opposing hazard components are
recalibrated to preserve the same score cancellation. Multiplicity adjustment
then leaves an arm-1 rejection
probability close to that of a comparison with no directional score signal,
about 0.01 here because the one-sided family-wise level 0.025 is calibrated
for the maximum of three shared-control statistics. The 12-month RMST gain
nevertheless remains positive. This deliberately
constructed blind spot does not imply that log-rank tests lack power under
crossing hazards generally and is not used as the primary practical
comparison in the main paper.

\begin{table}[p!]
\centering
\scriptsize
\caption{\textbf{Marginally independent censoring.} Ranges of arm-1 power across the four randomization designs in the log-rank-orthogonal crossing stress test.}
\label{tab:supp-equal-rmst-marginal-stress}
\resizebox{\textwidth}{!}{%
\begin{tabular}{llrrrrrrrr}
\toprule
&& \multicolumn{4}{c}{$N=800$} & \multicolumn{4}{c}{$N=1200$}\\
\cmidrule(lr){3-6}\cmidrule(lr){7-10}
Effect shape & Method & $0.50$ & $0.75$ & $1.00$ & $1.25$ & $0.50$ & $0.75$ & $1.00$ & $1.25$\\
\midrule
Reverse-crossing stress test & M2 & 0.125--0.148 & 0.335--0.357 & 0.626--0.658 & 0.881--0.909 & 0.199--0.218 & 0.491--0.531 & 0.822--0.851 & 0.977--0.984\\
 & M3 & 0.135--0.148 & 0.353--0.360 & 0.647--0.660 & 0.897--0.905 & 0.207--0.218 & 0.522--0.535 & 0.851--0.852 & 0.982--0.984\\
 & M6 & 0.010--0.011 & 0.010--0.013 & 0.012--0.013 & 0.013--0.016 & 0.010--0.012 & 0.011--0.015 & 0.012--0.016 & 0.012--0.015\\
\addlinespace
\bottomrule
\end{tabular}%
}
\begin{minipage}{0.98\textwidth}
\footnotesize
Each cell is the minimum--maximum rejection probability across the four designs. M2 and M3 target the 12-month RMST contrast; M6 tests equality of survival curves through month 12. The crossing profile is deliberately calibrated to have zero expected M6 score and is retained as a methodological stress test rather than a representative clinical crossing pattern. Results use 10,000 matched trials. M3 stability was at least 0.992; strong FWER ranged from 0.016 to 0.028.
\end{minipage}
\end{table}

\begin{table}[p!]
\centering
\scriptsize
\caption{\textbf{Conditionally independent censoring given strata.} Ranges of arm-1 power across the four randomization designs in the log-rank-orthogonal crossing stress test.}
\label{tab:supp-equal-rmst-conditional-stress}
\resizebox{\textwidth}{!}{%
\begin{tabular}{llrrrrrrrr}
\toprule
&& \multicolumn{4}{c}{$N=800$} & \multicolumn{4}{c}{$N=1200$}\\
\cmidrule(lr){3-6}\cmidrule(lr){7-10}
Effect shape & Method & $0.50$ & $0.75$ & $1.00$ & $1.25$ & $0.50$ & $0.75$ & $1.00$ & $1.25$\\
\midrule
Reverse-crossing stress test & M3 & 0.137--0.143 & 0.330--0.348 & 0.612--0.624 & 0.803--0.824 & 0.210--0.221 & 0.525--0.529 & 0.828--0.837 & 0.947--0.957\\
 & M6 & 0.011--0.012 & 0.010--0.013 & 0.012--0.012 & 0.012--0.013 & 0.010--0.011 & 0.009--0.013 & 0.010--0.013 & 0.011--0.014\\
\addlinespace
\bottomrule
\end{tabular}%
}
\begin{minipage}{0.98\textwidth}
\footnotesize
Each cell is the minimum--maximum rejection probability across the four designs. M3 targets the 12-month RMST contrast; M6 tests equality of survival curves through month 12. M2 is omitted because its marginal estimator does not identify that target under stratum-dependent censoring. The crossing profile is deliberately calibrated to have zero expected M6 score and is retained as a methodological stress test rather than a representative clinical crossing pattern. Results use 10,000 matched trials. M3 stability was at least 0.886; strong FWER ranged from 0.015 to 0.027.
\end{minipage}
\end{table}

\section{Application details}
\label{sec:supp-application}

This section gives the details of the ICON6-based application in Section~\ref{sec:application}
of the main paper, in the order in which the main paper refers to them, and
then reanalyzes the ACTG 175 trial, whose public data allow the
standardized procedure to be applied directly to stratified data.

\subsection{Reconstruction of the ICON6 data}
\label{sec:supp-icon6-reconstruction}

Figure~2 of \citet{ledermann2016cediranib} plots the PFS curves of the three
ICON6 arms as vector paths. We read the line segments of each colored path
from the PDF, as in \citet{zhao2025synthipd}, and mapped page coordinates to months and survival
probabilities using the axis tick marks, which lie at 0 and 24 months on
the horizontal axis and at 0 and 100\% on the vertical axis. For each
distinct time we kept the lowest survival value, which gives the level of
the curve after each downward step. The numbers at risk at 0, 6, 12, 18,
and 24 months were then combined with these coordinates by the algorithm of
\citet{guyot2012enhanced}, which distributes censoring times uniformly
within each risk interval and adjusts their number until the reconstructed
risk set at the start of the next interval matches the published value.
Table~\ref{tab:supp-icon6-validation} compares the reconstruction with the
published summaries.

\begin{table}[t!]
\centering
\small
\caption{Validation of the reconstructed ICON6 data against
\citet{ledermann2016cediranib}. Numbers at risk are given at 0, 6, 12, 18,
and 24 months. Events are counted within 24 months.}
\label{tab:supp-icon6-validation}
\resizebox{\textwidth}{!}{%
\begin{tabular}{llccc}
\toprule
Quantity & Source & Arm A & Arm B & Arm C\\
\midrule
Median PFS (months) & Published & 8.7 & 9.9 & 11.0\\
 & Reconstructed & 8.7 & 9.9 & 11.0\\
\addlinespace
Numbers at risk & Published & 118, 90, 24, 7, 2 & 174, 151, 51, 18, 11 & 164, 146, 63, 22, 7\\
 & Reconstructed & 118, 90, 24, 7, 2 & 174, 151, 50, 17, 10 & 164, 146, 63, 22, 7\\
\addlinespace
Events within 24 months & Published & 113 & 152 & 140\\
 & Reconstructed & 113 & 152 & 138\\
\bottomrule
\end{tabular}}
\end{table}

\subsection{Generating model of the ICON6-based emulation}
\label{sec:supp-icon6-model}

Let $S_a^{\mathrm{rec}}(t)$ denote the PFS curve of arm $a$ extracted from
the published figure. The two factors $Z_{i1}$ (relapse-free interval of 6
to 12 months) and $Z_{i2}$ (planned regimen other than carboplatin plus
paclitaxel) are independent Bernoulli variables with probabilities 0.33 and
0.25, and they define four joint strata with probabilities $p_s$. With
$\eta_s=\log(1.6)Z_{1}+\log(1.15)Z_{2}$ evaluated at the profile of stratum
$s$, the survival function of arm $a$ in stratum $s$ is
\[
S_{as}(t)=\{B_a(t)\}^{\exp(\eta_s)},
\qquad
\sum_{s}p_s\{B_a(t)\}^{\exp(\eta_s)}=S_a^{\mathrm{rec}}(t),
\]
where $B_a(t)$ is obtained by solving the second equation for each $t$ on a
grid of 0.01 months. The marginal survival curve of every arm therefore
equals its ICON6 curve, while the strata differ in prognosis. Event times
are drawn by inverting $S_{as}$, so that the treatment effect keeps the
time-varying shape observed in ICON6. Censoring times are drawn from the
Kaplan--Meier estimate of the censoring distribution in the reconstructed
data, independently of the event times, and all follow-up ends at 24
months. Under the global null, every arm uses the curve of arm A. Simple
randomization assigns arms with probabilities $(2,3,3)/8$, stratified
permuted blocks use blocks of eight containing two, three, and three
assignments, and Pocock--Simon minimization assigns probability 0.80 to the
arms minimizing the summed absolute deviation of the marginal counts from
the 2:3:3 target, with the remaining probability shared in proportion to
the target allocation. Positivity was never violated in a way that left a
cell's Kaplan--Meier curve undefined at $\tau^\star$, because censoring
before 12 months was rare and cells without patients at risk at
$\tau^\star$ had already reached zero survival through an event.

\subsection{Sensitivity to the restriction time}
\label{sec:supp-icon6-tau}

Because the 12-month restriction time could be suspected of favoring the
RMST analysis, we repeated the emulation at the realized sample size with
$\tau^\star\in\{9,15,18,24\}$, using the process-wise range
$[4,\tau^\star]$ and the log-rank test both with follow-up censored at
$\tau^\star$ and with all follow-up. Table~\ref{tab:supp-icon6-tau} gives
the results. The design-adjusted marginal procedure M2 had power between
0.795 and 0.871 for the waning benefit of arm B at every restriction time,
compared with 0.679 to 0.832 for the log-rank test censored at the same
time and 0.791 to 0.811 for the log-rank test with all follow-up. Its
advantage was smallest at $\tau^\star=9$, where the hazards of arms A and
B had not yet converged, and it persisted at every longer restriction
time, although at $\tau^\star=18$ the fixed-time advantage over the
log-rank test with all follow-up was only 0.003 to 0.019 while the
process-wise advantage was 0.032 to 0.042. The power gain therefore does
not depend on the choice $\tau^\star=12$. The family-wise error rate of M2
was 0.017 to 0.029.

The standardized procedure M3 coincided with M2 at $\tau^\star=9$ and
$\tau^\star=12$. At longer restriction times it was limited by
positivity. Censoring in ICON6 was rare before 12 months but frequent
afterward, so in the smaller treatment--stratum cells the last
observation was often a censoring time before $\tau^\star$, which leaves
the cell's Kaplan--Meier curve undefined. This happened in about 12\% of
trials at $\tau^\star=15$, 34\% at $\tau^\star=18$, and 75\% at
$\tau^\star=24$. Counting these trials as nonrejections reduced the power
of M3 to 0.744 to 0.753 at $\tau^\star=15$, and we do not report M3 for
longer restriction times. In practice, the restriction time for the
standardized analysis should be chosen so that every treatment--stratum
cell is expected to retain patients at risk, and the marginal procedure M2
is the appropriate choice when a longer horizon is clinically required.

\begin{table}[t!]
\centering
\scriptsize
\caption{Sensitivity of the ICON6-based emulation to the restriction time at the realized sample size $N=456$. Entries are FWER under the global null and power for arm B (waning benefit), each given as SR/SPB/PS. The process-wise analysis uses $[4,\tau^\star]$. M3 is not reported for $\tau^\star\geq18$ because the standardized estimator was undefined in a large fraction of trials. The last column gives the fraction of trials with an undefined standardized estimate under PS when arms B and C are active.}
\label{tab:supp-icon6-tau}
\begin{tabular}{llccc}
\toprule
$\tau^\star$ & Method & FWER & Power, arm B & Undefined\\
\midrule
9 & M2, fixed time & 0.021/0.027/0.022 & 0.845/0.857/0.855 & 0.000\\
 & M2, process & 0.019/0.024/0.017 & 0.778/0.796/0.793 & \\
 & M3, fixed time & 0.022/0.027/0.022 & 0.843/0.858/0.854 & \\
 & M3, process & 0.020/0.024/0.019 & 0.783/0.797/0.791 & \\
 & M6, horizon $\tau^\star$ & 0.026/0.029/0.029 & 0.820/0.828/0.832 & \\
 & M6, 24 months & 0.027/0.030/0.028 & 0.797/0.809/0.801 & \\
\addlinespace
12 & M2, fixed time & 0.024/0.025/0.024 & 0.839/0.853/0.860 & 0.000\\
 & M2, process & 0.021/0.021/0.019 & 0.818/0.831/0.838 & \\
 & M3, fixed time & 0.025/0.025/0.024 & 0.848/0.854/0.860 & \\
 & M3, process & 0.021/0.021/0.019 & 0.822/0.832/0.837 & \\
 & M6, horizon $\tau^\star$ & 0.030/0.026/0.030 & 0.679/0.683/0.690 & \\
 & M6, 24 months & 0.029/0.027/0.026 & 0.791/0.802/0.811 & \\
\addlinespace
15 & M2, fixed time & 0.026/0.025/0.023 & 0.836/0.847/0.855 & 0.117\\
 & M2, process & 0.021/0.022/0.021 & 0.826/0.836/0.844 & \\
 & M3, fixed time & 0.022/0.021/0.020 & 0.744/0.748/0.753 & \\
 & M3, process & 0.019/0.018/0.017 & 0.736/0.742/0.744 & \\
 & M6, horizon $\tau^\star$ & 0.030/0.026/0.026 & 0.700/0.702/0.711 & \\
 & M6, 24 months & 0.031/0.027/0.028 & 0.798/0.799/0.808 & \\
\addlinespace
18 & M2, fixed time & 0.024/0.024/0.028 & 0.795/0.814/0.814 & 0.338\\
 & M2, process & 0.021/0.019/0.023 & 0.824/0.834/0.837 & \\
 & M6, horizon $\tau^\star$ & 0.027/0.025/0.030 & 0.697/0.712/0.704 & \\
 & M6, 24 months & 0.028/0.027/0.030 & 0.792/0.798/0.795 & \\
\addlinespace
24 & M2, fixed time & 0.027/0.027/0.029 & 0.846/0.871/0.860 & 0.755\\
 & M2, process & 0.023/0.025/0.024 & 0.848/0.870/0.858 & \\
 & M6, 24 months & 0.028/0.028/0.030 & 0.793/0.810/0.799 & \\
\addlinespace
\bottomrule
\end{tabular}
\end{table}

\subsection{Emulation with interacting factors}
\label{sec:supp-icon6-interaction}

To examine the design term in a realistic setting, we repeated the null
emulation with stratum hazard ratios of 1, 2.5, 2.5, and 1 for the four
joint strata, so that each factor is prognostic only in the absence of the
other. Table~\ref{tab:supp-icon6-interaction} reports the fixed-time
family-wise error rates. Under minimization, M5, which omits the design
term, had rates of 0.029 and 0.033, while M2 and M3 stayed at or below
0.025. Under stratified blocks, the conventional analysis M1 was
conservative, with rates of 0.018.

\begin{table}[t!]
\centering
\scriptsize
\caption{Fixed-time family-wise error rate in the ICON6-based emulation when the two randomization factors interact.}
\label{tab:supp-icon6-interaction}
\begin{tabular}{llccc}
\toprule
Quantity & Method & SR & SPB & PS\\
\midrule
FWER, interacting factors & M1 & \shortstack{0.027\\0.026} & \shortstack{0.018\\0.018} & \shortstack{0.021\\0.025}\\
 & M2 & \shortstack{0.025\\0.025} & \shortstack{0.024\\0.024} & \shortstack{0.021\\0.025}\\
 & M5 & \shortstack{0.036\\0.033} & \shortstack{0.024\\0.025} & \shortstack{0.029\\0.033}\\
 & M3 & \shortstack{0.023\\0.024} & \shortstack{0.024\\0.024} & \shortstack{0.021\\0.023}\\
 & M6, 12 months & \shortstack{0.027\\0.028} & \shortstack{0.028\\0.029} & \shortstack{0.029\\0.031}\\
\addlinespace
\bottomrule
\end{tabular}
\begin{minipage}{0.92\textwidth}
\footnotesize
Same layout as Table~\ref{tab:icon6-emulation} of the main paper. The stratum hazard ratios are 1, 2.5, 2.5, and 1 for the four joint strata.
\end{minipage}
\end{table}

\subsection{Emulation with stratum-dependent dropout}
\label{sec:supp-icon6-censoring}

In ICON6 almost no patient was censored before 12 months, so censoring
could not distort the 12-month RMST in the main emulation. To show what
happens when it does, we added a hypothetical dropout process to the null
emulation. Each patient is subject to exponential dropout with monthly
rate $0.008\exp(\gamma_a Z_{i1})$, where $Z_{i1}$ indicates a relapse-free
interval of 6 to 12 months and $(\gamma_0,\gamma_1,\gamma_2)=(0,2.0,2.5)$
for arms A, B, and C. Dropout before month 12 then affects about 7\%,
16\%, and 22\% of patients in the three arms and is concentrated among
patients with a short relapse-free interval in the cediranib arms, as may
happen when toxicity leads frailer patients to withdraw. Censoring is
independent of the event time within each joint stratum but not marginally.
We used $N=456$ and $N=912$, with 10,000 trials per configuration.

Table~\ref{tab:supp-icon6-censoring} gives the results. Because the
patients who dropped out had worse prognosis, the unstratified estimator
overstated the RMST of the cediranib arms, with a bias of about 0.1 months
for arm C that did not shrink with the sample size. The family-wise error
rate of M1 and M2 therefore rose from 0.034--0.040 at $N=456$ to
0.046--0.050 at $N=912$. The standardized estimator remained unbiased, and
M3 kept the error rate at or below 0.022, or 0.022 to 0.026 among trials in
which its estimate was defined. The stratified log-rank test was also
unaffected, because it too compares arms within strata. The same early
dropout that biases the unstratified estimator also thins the smaller
treatment--stratum cells, and at least one cell's Kaplan--Meier curve was
undefined at 12 months in about 24\% of trials at $N=456$ and 12\% at
$N=912$. When early dropout is expected to depend on the randomization
factors, the protocol should therefore prespecify how sparse strata are
pooled or choose a restriction time at which every cell is expected to
retain patients at risk.

\begin{table}[t!]
\centering
\scriptsize
\caption{ICON6-based emulation with dropout that depends on the relapse-free interval and on the arm. Each cell gives the value for $N=456$ (first line) and $N=912$ (second line), based on 10,000 trials under the global null. Bias is that of the 12-month RMST contrast in months. M1 and M2 share the unstratified point estimator. M3 counts trials with an undefined standardized estimate as nonrejections, and $^\dagger$ gives its rate among the trials in which the estimate is defined.}
\label{tab:supp-icon6-censoring}
\begin{tabular}{llccc}
\toprule
Quantity & Method & SR & SPB & PS\\
\midrule
Bias, arm C & M1 and M2 & \shortstack{0.104\\0.107} & \shortstack{0.105\\0.113} & \shortstack{0.107\\0.106}\\
 & M3 & \shortstack{0.002\\0.001} & \shortstack{0.002\\0.006} & \shortstack{0.001\\0.002}\\
\addlinespace
Undefined M3 estimate & Global null & \shortstack{0.246\\0.124} & \shortstack{0.239\\0.122} & \shortstack{0.239\\0.123}\\
\addlinespace
FWER, global null & M1 & \shortstack{0.036\\0.049} & \shortstack{0.037\\0.046} & \shortstack{0.038\\0.046}\\
 & M2 & \shortstack{0.034\\0.048} & \shortstack{0.038\\0.050} & \shortstack{0.040\\0.048}\\
 & M3 & \shortstack{0.017\\0.021} & \shortstack{0.018\\0.022} & \shortstack{0.019\\0.022}\\
 & M3, defined$^\dagger$ & \shortstack{0.022\\0.024} & \shortstack{0.024\\0.025} & \shortstack{0.025\\0.026}\\
 & M6, 12 months & \shortstack{0.024\\0.025} & \shortstack{0.025\\0.025} & \shortstack{0.026\\0.026}\\
\addlinespace
\bottomrule
\end{tabular}
\end{table}

\subsection{Reanalysis of ACTG 175}
\label{sec:actg175}

Unlike the reconstructed ICON6 data, the public ACTG 175 data contain a
stratification variable, so the standardized procedure can be applied
directly. ACTG 175 compared zidovudine (ZDV) monotherapy with ZDV plus didanosine
(ddI), ZDV plus zalcitabine, and ddI monotherapy in adults with HIV and
baseline CD4 counts of 200--500 cells per cubic millimeter
\citep{hammer1996actg175}. The primary endpoint was the first of a decline
of at least 50\% in CD4 count, an AIDS-defining event, or death, and
randomization was blocked within sites and stratified by duration of prior
antiretroviral therapy. We analyzed the public analytic dataset
\citep{uci2023actg175}, which contains 2,139 of the 2,467 randomized
participants and the three-level prior-therapy variable but not study site.
The analysis is therefore a reproducible methodological reanalysis rather
than a reconstruction of the original intention-to-treat analysis.

With ZDV as the shared control, we applied the stratum-standardized
procedure over the three prior-therapy strata, with $\tau^\star=730$ days and
process-wise inference over $\calT=[180,730]$ on a 14-day grid, using the
one-sided step-down max-$T$ procedure at level 0.025, two-sided 95\%
simultaneous intervals and bands, and 100,000 Gaussian multiplier draws. A
Cox model stratified by prior therapy, censored at day 730 and calibrated
with the same max-$T$ construction, served as a secondary benchmark. The 12
treatment--stratum cells contained 96--238 participants, with 64--183 still
at risk at day 730, so no pooling of strata was needed
(Table~\ref{tab:supp-actg-data}).

The standardized 730-day RMST was 645.5 days under ZDV, 696.2 days under ZDV
plus ddI, 694.0 days under ZDV plus zalcitabine, and 685.2 days under ddI,
corresponding to gains of 50.7, 48.5, and 39.7 event-free days
(Table~\ref{tab:actg175-rmst}). All three simultaneous intervals excluded
zero, and all fixed-time and process-wise step-down adjusted $p$-values were
below 0.0001.

\begin{table}[t!]
\centering
\small
\caption{Stratum-standardized ACTG 175 RMST results. RMST and differences are
in days at $\tau^\star=730$. Confidence intervals are two-sided and
simultaneous across the three fixed-time treatment-control contrasts.
Adjusted $p$-values are from the one-sided step-down max-$T$ procedures at
the fixed time and over $[180,730]$. The standardized control-arm RMST was
645.5 days.}
\label{tab:actg175-rmst}
\resizebox{\textwidth}{!}{%
\begin{tabular}{lrrrrrr}
\toprule
Treatment & RMST & Difference & SE & Simultaneous 95\% CI & Fixed $p_{\mathrm{adj}}$ & Process $p_{\mathrm{adj}}$\\
\midrule
ZDV + ddI & 696.2 & 50.7 & 8.8 & $(30.6,\ 70.8)$ & $<0.0001$ & $<0.0001$\\
ZDV + Zal & 694.0 & 48.5 & 8.9 & $(28.1,\ 69.0)$ & $<0.0001$ & $<0.0001$\\
ddI & 685.2 & 39.7 & 9.1 & $(18.7,\ 60.6)$ & $<0.0001$ & $<0.0001$\\
\bottomrule
\end{tabular}

}
\end{table}

The estimated correlations among the three fixed-time contrasts were
0.680--0.706, which is the shared-control dependence exploited by the
max-$T$ procedure. The simultaneous lower band
(Figure~\ref{fig:actg175-rmst-process}) remained above zero from day 180 for
ZDV plus ddI, from day 250 for ZDV plus zalcitabine, and from day 222 for
ddI. These times are descriptive summaries of the band, not estimated onset
parameters.

\begin{figure}[t!]
\centering
\includegraphics[width=0.86\textwidth]
{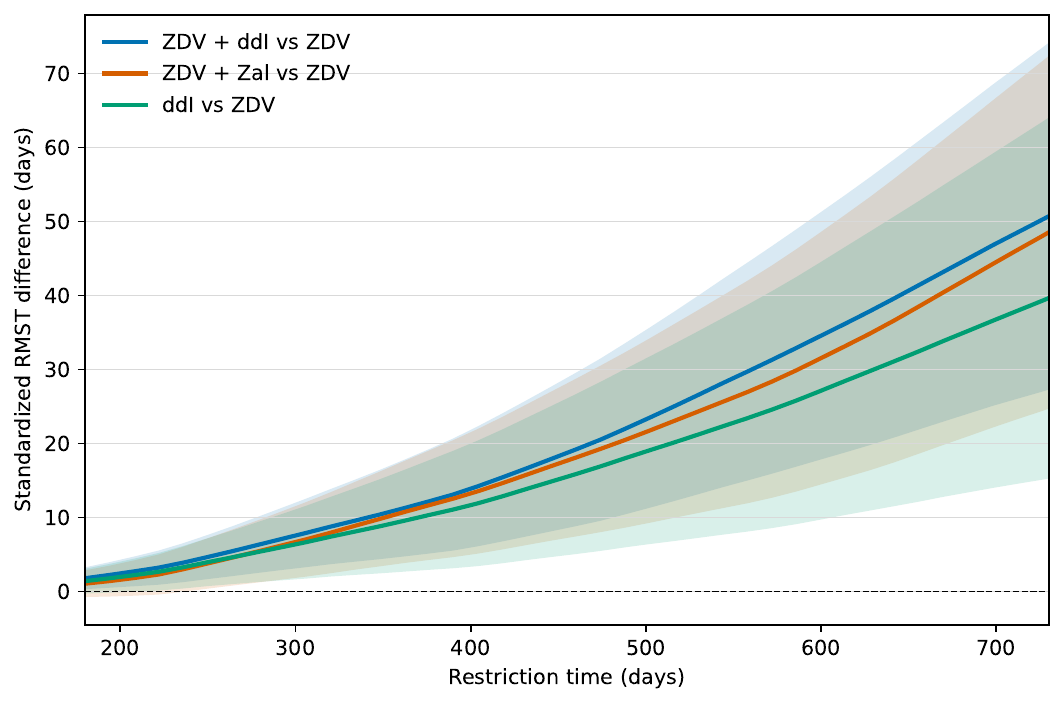}
\caption{Stratum-standardized RMST differences relative to ZDV in ACTG 175.
Shaded regions are two-sided 95\% confidence bands simultaneous over all
three comparisons and all restriction times on the grid from day 180 to day
730. Positive values favor the displayed treatment.}
\label{fig:actg175-rmst-process}
\end{figure}

The stratified Cox benchmark gave the same qualitative conclusion, with
hazard ratios relative to ZDV of 0.418--0.551, all simultaneous intervals
excluding one, and no evidence against proportional hazards through day 730
(Table~\ref{tab:supp-actg-cox}). The RMST analysis adds the absolute
interpretation in event-free days and the simultaneous restriction-time
profile. All results are reproduced by
\texttt{run\_actg175\_analysis.py}.

\begin{table}[t!]
\centering
\small
\caption{ACTG 175 public analytic dataset. ``Full events'' uses all recorded
follow-up in the public file, and the final two columns refer to the selected
730-day analysis horizon.}
\label{tab:supp-actg-data}
\resizebox{0.94\textwidth}{!}{%
\begin{tabular}{lrrrr}
\toprule
Study arm & $N$ & Full events & Events by day 730 & At risk at day 730\\
\midrule
ZDV & 532 & 181 & 134 & 342\\
ZDV + ddI & 522 & 103 & 67 & 411\\
ZDV + Zal & 524 & 109 & 62 & 417\\
ddI & 561 & 128 & 86 & 425\\
\bottomrule
\end{tabular}

}
\end{table}

\begin{table}[t!]
\centering
\small
\caption{Secondary ACTG 175 stratified Cox benchmark through day 730.
Confidence intervals are two-sided and simultaneous across the three hazard
ratios. Adjusted $p$-values test a treatment benefit using the one-sided
step-down max-$T$ procedure. PH-test $p$-values are treatment-specific
Schoenfeld-residual checks and are not multiplicity adjusted.}
\label{tab:supp-actg-cox}
\resizebox{0.86\textwidth}{!}{%
\begin{tabular}{lrrrr}
\toprule
Treatment & Hazard ratio & Simultaneous 95\% CI & $p_{\mathrm{adj}}$ & PH-test $p$\\
\midrule
ZDV + ddI & 0.448 & $(0.314,\ 0.640)$ & $<0.0001$ & 0.2884\\
ZDV + Zal & 0.418 & $(0.290,\ 0.603)$ & $<0.0001$ & 0.2958\\
ddI & 0.551 & $(0.397,\ 0.766)$ & $<0.0001$ & 0.7379\\
\bottomrule
\end{tabular}

}
\end{table}

\end{document}